\documentclass[letterpaper,11pt]{article}
\usepackage{cite}

\usepackage[utf8]{inputenc}
\usepackage[letterpaper,hmargin=1in,vmargin=1in,footskip=36pt]{geometry}
\usepackage{mathptmx} 
\DeclareMathAlphabet\mathcal{OMS}{cmsy}{m}{n}
\SetMathAlphabet\mathcal{bold}{OMS}{cmsy}{b}{n}

\usepackage[utf8]{inputenc}
\usepackage[english]{babel}
\usepackage{amsthm}
\usepackage{enumerate}
\usepackage[fleqn]{amsmath}
\usepackage{amsfonts}
\usepackage{amssymb}
\usepackage{amsbsy}
\usepackage[svgnames]{xcolor}
\usepackage{colortbl}
\usepackage{lineno}
\usepackage{xspace}
\usepackage{bm}
\usepackage{scalerel}
\usepackage{stackengine,wasysym}

\usepackage{algorithm}  
\usepackage{algpseudocode}  
\usepackage{amsmath}

\def\ve#1{\mathchoice{\mbox{\boldmath$\displaystyle\bf#1$}}
	{\mbox{\boldmath$\textstyle\bf#1$}}
	{\mbox{\boldmath$\scriptstyle\bf#1$}}
	{\mbox{\boldmath$\scriptscriptstyle\bf#1$}}}

\makeatletter
\let\bfseries=\undefined
\DeclareRobustCommand\bfseries
{\not@math@alphabet\bfseries\mathbf
	\boldmath\fontseries\bfdefault\selectfont}
\makeatother

\def\Orthant_j{{\mathcal O}_{j}}

\newcommand\vea{{\ve a}}
\newcommand\veb{{\ve b}}

\newcommand\ved{{\ve d}}

\newcommand\veg{{\ve g}}

\newcommand\vep{{\ve p}}

\newcommand\vev{{\ve v}}

\newcommand\vex{{\ve x}}
\newcommand\vey{{\ve y}}

\newcommand{\OO}{{\mathcal{O}}}

\usepackage[textsize=tiny,textwidth=2cm,shadow,color=green]{todonotes}

\newenvironment{psmallmatrix}{\left(\smallmatrix}{\endsmallmatrix\right)}

\newcommand\FourBlockBig[5][\relax]{\begin{pmatrix}#2& #3\\#4&#5 \end{pmatrix}\ifx#1\relax\else^{(#1)}\fi}
\newcommand\FourBlock[5][\relax]{\begin{psmallmatrix}#2& #3\\#4&#5 \end{psmallmatrix}\ifx#1\relax\else{^{(#1)}}\fi}
\newcommand\TwoBlock[3][\relax]{\begin{psmallmatrix}#2\\#3 \end{psmallmatrix}\ifx#1\relax\else{^{(#1)}}\fi}

\newtheorem{theorem}{Theorem}
\newtheorem{claim}{Claim}

\newtheorem{lemma}{Lemma}

\newtheorem{observation}{Observation}

\makeatletter
\newtheorem*{rep@theorem}{\rep@title}
\newcommand{\newreptheorem}[2]{%
	\newenvironment{rep#1}[1]{%
		\def\rep@title{#2 \ref{##1}}%
		\begin{rep@theorem}}%
		{\end{rep@theorem}}}
\makeatother

\newreptheorem{theorem}{Theorem}
\newreptheorem{corollary}{Corollary}
\newreptheorem{lemma}{Lemma}
\newreptheorem{proposition}{Proposition}

\title{Approximation Algorithms for Interdiction Problem with Packing Constraints}

\author{Lin Chen\thanks{Department of Computer Science, Texas Tech University.
		\texttt{chenlin198662@gmail.com}.} \ \ \ \  
		Xiaoyu Wu\thanks{School of Mathematical Sciences, Zhejiang University.
		\texttt{xiaoyu\_wu@zju.edu.cn}.} \ \ \ \ 
		Guochuan Zhang\thanks{School of Computer Science, Zhejiang University. 
		\texttt{zgc@zju.edu.cn}.} }

\date{\today}

\begin{document}
	
	\maketitle
	
	\thispagestyle{empty}

\begin{abstract}

	We study a bilevel optimization problem which is a zero-sum Stackelberg game. In this problem, there are two players, a leader and a follower, who pick items from a common set. Both the leader and the follower have their own (multi-dimensional) budgets, respectively. Each item is associated with a profit, which is the same to the leader and the follower, and will consume the leader's (follower's) budget if it is selected by the leader (follower). The leader and the follower will select items in a sequential way: First, the leader selects items within the leader's budget. Then the follower selects items from the remaining items within the follower's budget. The goal of the leader is to minimize the maximum profit that the follower can obtain. Let $s_A$ and $s_B$ be the dimension of the leader's and follower's budget, respectively. A special case of our problem is the bilevel knapsack problem studied by Caprara et al. [SIAM Journal on Optimization, 2014], where $s_A=s_B=1$. We consider the general problem and obtain an $(s_B+\epsilon)$-approximation algorithm when $s_A$ and $s_B$ are both constant. In particular, if $s_B=1$, our algorithm implies a PTAS for the bilevel knapsack problem, which is the first $\OO(1)$-approximation algorithm. We also complement our result by showing that there does not exist any $(4/3-\epsilon)$-approximation algorithm even if $s_A=1$ and $s_B=2$. We also consider a variant of our problem with resource augmentation when $s_A$ and $s_B$ are both part of the input. We obtain an $\OO(1)$-approximation algorithm with $\OO(1)$-resource augmentation, that is, we give an algorithm that returns a solution which exceeds the given leader's budget by $\OO(1)$ times, and the objective value achieved by the solution is $\OO(1)$ times the optimal objective value that respects the leader's budget.

\end{abstract}

\vspace{2mm}
\hspace{3.5mm}\textbf{Keywords:} Bilevel Integer Programming, Interdiction Constraints, Knapsack

\clearpage
\setcounter{page}{1}

\section{Introduction}\label{sec:introduction}

In recent years, there is an increasing interest in adopting the {\it Stackelberg competition model}~\cite{dixit1980role} to address the critical security concern that arises in protecting our ports, airports, transportation, and other critical national infrastructures (see, e.g.,~\cite{an2013deployed,paruchuri2008efficient,pita2008deployed}). 
In these problems, the attacker's target is to maximize the illicit gain, while the defender tries to mitigate the attack by minimizing the attacker's objective through deploying defending resources. 

In this paper, we consider an abstract model for general defending problems called \textit{interdiction with packing constraints} ({IPC}). In IPC, given are a set of items, together with a {\it leader} and a {\it follower}. Both the leader and the follower have their own (multi-dimensional) budgets, respectively. Each item is associated with a profit, which is the same to the leader and the follower, and will consume the leader's (follower's) budget if it is selected by the leader (follower). The leader and the follower will select items in a sequential way: First, the leader selects items within the leader's budget. Then the follower selects items from the remaining items within the follower's budget. The goal of the leader is to minimize the maximum profit that the follower can obtain. IPC captures the general setting where the follower is the attacker who gets profit by attacking items, and the leader is the defender who tries to minimize the attacker's gain by protecting a subset of items.

IPC can be formulated as a bilevel integer program (IP) as follows. Denote by $I = \{1,2,\cdots,n\}$ the set of items. Each item $j \in I$ is associated with a \textit{profit} $p_j \in \mathbb{Q}_{>0}$, an $s_A$-dimensional \textit{cost vector} $\ve A_j \in \mathbb{Q}_{\ge 0}^{s_A}$ to the leader and an $s_B$-dimensional \textit{weight vector} $\ve B_j\in\mathbb{Q}_{\ge 0}^{s_B}$ to the follower. The leader and the follower have their own \textit{budget vectors}, denoted by $\vea\in\mathbb{Q}_{\ge 0}^{s_A}$ and $\veb\in\mathbb{Q}_{\ge 0}^{s_B}$, respectively. We introduce 0-1 variables $x_j$ and $y_j$ for each $j \in I$ as the decision variables for the leader and the follower. More precisely, if the leader chooses item $j$, then $x_j =1$, otherwise $x_j = 0$. Similarly, $y_j=1$ if the follower chooses item $j$ and $y_j=0$ otherwise. Denote by  $\vex = (x_1,x_2,\dots, x_n)$, $\vey = (y_1,y_2,\dots, y_n)$, $\vep = (p_1,p_2,\dots, p_n)$ and $\ve1 = (\underbrace{1,\dots ,1}_{n})$. IPC can be formulated as a bilevel program \textbf{IPC}$(I,\vea,\veb)$ as follows:
\begin{subequations} 
	\begin{align}
	\textbf{IPC}(I,\vea,\veb):
	\min\limits_{\vex}\hspace{1mm} & \vep\vey&\label{IPC:a1} \\
	s.t. \hspace{1mm}&\ve A\vex\le \vea& \label{IPC:b1} \\
	& \hspace{1mm} \vex \in\{0,1\}^n& \label{IPC:c1}\\
	& \text{where }  \vey \text{ solves the following:}& \nonumber\\
	&\max\limits_{\vey}\quad \vep\vey &  \label{IPC:a2} \\
	& s.t. \hspace{5mm} \ve B\vey\le\veb& \label{IPC:b2}\\
	&\hspace{10mm} \vex+\vey\le \ve1 & \label{IPC:c2} \\
	& \hspace{10mm} \vey\in\{0,1\}^n& \label{IPC:d2}
	\end{align}
\end{subequations}
where $\ve A=(\ve A_{1},\cdots,\ve A_n)$ and $\ve B=(\ve B_{1},\cdots,\ve B_n)$ are $s_A\times n$ and $s_B\times n$ non-negative rational matrices, respectively.

The most relevant prior work to our IPC model is the well-known knapsack interdiction problem introduced by DeNegre~\cite{denegre2011interdiction}, which is the special case of IPC where $s_A = 1$ and $s_B=1$. Very recently, Caprara et al.~\cite{DBLP:conf/ipco/CapraraCLW13} proved that DeNegre's knapsack interdiction problem is $\sum_{2}^{p}$-complete and strongly NP-hard, which also implies the $\sum_{2}^{p}$-completeness and strongly NP-hardness for IPC. Caprara et al. showed a polynomial time approximation scheme (PTAS) for a special case of knapsack interdiction problem where the profit of an item is equal to its weight to the follower. 

Except for the knapsack interdiction problem, we are not aware of any approximation algorithms for other special cases of IPC. However, if we relax the follower's problem by allowing $\vey$ to take fractional value, then there are several research works in the literature. The most relevant work is the packing interdiction problem studied by Dinitz and Gupta~\cite{DBLP:conf/ipco/DinitzG13}, where the follower's problem is given by $\max\{\sum_j p_j(1-x_j)y_j:B\vey\le \veb, \vey\ge \ve 0\}$, while the leader's constraints are the same as \eqref{IPC:b1} and \eqref{IPC:c1} except that $s_A=1$. Dinitz and Gupta provided an approximation algorithm whose ratio depends on the sparsity of the matrix $B$. Their techniques crucially rely on the fact that $\vey$ can take fractional value, and therefore duality theory can be applied to the follower's problem, allowing  the bilevel problem to be transformed to a single level problem. Besides the packing interdiction problem, quite a few graph interdiction problems have been studied in the literature, where the follower's problem is a standard graph optimization problem, and the leader can remove edges or vertices to minimize the follower's optimal objective value on the graph after edge-removal or vertex-removal. On planar graphs, polynomial time approximation schemes (PTASs) were obtained for network flow interdiction~\cite{DBLP:conf/stoc/Phillips93}\cite{DBLP:journals/dam/Zenklusen10} and matching interdiction~\cite{DBLP:journals/dam/PanS16}. On general graphs, approximation algorithms were also obtained for, e.g., connectivity interdiction~\cite{DBLP:journals/orl/Zenklusen14}, minimum spanning tree interdiction~\cite{DBLP:conf/focs/Zenklusen15}\cite{DBLP:conf/icalp/LinharesS17}, matching interdiction \cite{DBLP:journals/dam/Zenklusen10a}\cite{DBLP:conf/ipco/DinitzG13}, network flow interdiction~\cite{chestnut2017hardness}~\cite{burch2003decomposition}~\cite{DBLP:journals/mor/ChestnutZ17}, etc. All of these algorithms crucially rely on the follower's specific graph optimization problem and do not apply directly to IPC. 


\paragraph*{Our Contributions}

The main contribution of this paper is an $(s_B+\epsilon)$-approximation polynomial time algorithm for IPC when $s_A$ and $s_B$ are both constant. In particular, when $s_B=1$, our algorithm is a PTAS. Since the knapsack interdiction problem is a special case of IPC when $s_A=s_B=1$, our result gives the first $\OO(1)$-approximation algorithm for this problem. To complement our result, we also show that IPC does not admit any $({4}/{3}-\epsilon)$-approximation algorithm even if $s_B = 2$ and $s_A=1$, assuming $P\neq NP$. This implies that the PTAS for $s_B=1$ cannot be further extended to the case of $s_B\ge 2$. 

We also consider a natural variant of IPC where the leader's budget can be violated. For this variant we obtain a $(\frac{\rho}{1-\alpha},\frac{1}{\alpha})$-bicriteria approximation algorithm for any $\alpha\in (0,1)$, which runs in polynomial time when $s_A$ and $s_B$ are arbitrary (not necessarily polynomial in the input size). More precisely, the algorithm takes as input two oracles, a $\rho$-approximation algorithm to the follower's optimization problem $\max\{\vep\vey:\ve B\vey\le \veb,\vey\in\{0,1\}^n\}$; and a separation oracle for the leader's problem that given any $\vex=\vex^0$, it either asserts that $\ve A\vex^0\le \vea$ or returns a violating constraint. Then in polynomial oracle time the algorithm returns a solution $\vex^*$ for the leader such that $\ve A\vex^*\le \frac{1}{\alpha}\vea$, and the objective value is at most $\frac{\rho T^*}{1-\alpha}$, where $T^*$ is the optimal objective value with the leader's budget being $\vea$. When we take, e.g., $\alpha=1/2$, we achieve an objective of $2\rho T^*$ with the leader's budget augmented to $2\vea$.

In terms of techniques, our main contribution is a general method for bilevel optimization problems where the leader's and follower's decision variables are both integral. Most prior works on bilevel optimization require follower's decision variables to take fractional values, which accommodates the application of LP duality to transform the bilevel optimization problem to a standard (single level) optimization, and are thus inapplicable when the follower's decision variables become integral. A common technique used in many single level optimization problems is to first classify items into large and small based on whether they can make a significant contribution to the objective value, then guess out large items via enumeration, and handle small items fractionally via LP (see, e.g.,~\cite{jansen2010parameterized,jansen2010eptas,caprara2000approximation}. However, such a technique encounters a fundamental challenge in IPC: we can guess out all large items selected by the leader, however, the follower may still select arbitrarily from the remaining large items. In other word, the follower's choice on large items can never be guessed out, and therefore we cannot apply duality to the follower's problem. We overcome the challenge based on the following two ideas: First, we show that given leader's choice on large items, there is a fixed number of "dominant choices" such that the follower's choice on remaining large items always belong to the dominant choices. Second, we show that there exists a subset of "critical items" such that the follower's choice on small items can be characterized through linear constraints given that these critical items are known. The two observations allow us to transform the bilevel program for IPC to an LP without utilizing duality.    

The characterization of dominant choices and critical items become sophisticated in the general case when $s_B$ is an arbitrary constant, but is much simpler in the special case $s_B=1$. Hence, for ease of presentation, in the main part we present our algorithm for the special case to give an overview on the technical insights, and meanwhile provide a proof sketch towards generalizing the algorithm for the general case. Our techniques may be of separate interest to other bilevel optimization problems.



\paragraph*{Related work}

Our IPC problem lies generally in the area of \textit{bilevel optimization}, which has received extensive research in the literature. Jeroslow~\cite{DBLP:journals/mp/Jeroslow85} showed that in general, bilevel optimization problems are NP-hard even when the objectives and the constraints are linear. We refer readers to~Colson et al.~\cite{DBLP:journals/anor/ColsonMS07} for a comprehensive survey on bilevel optimization.

Within the area of bilevel optimization, \textit{Mixed-Integer Bilevel Linear Problem} (MIBLP) is related to our IPC. MIBLP is a bilevel optimization problem where the objective functions and the constraints for the leader and follower are both linear. MIBLP has been studied extensively in the literature, see, e.g.,~\cite{DBLP:journals/jgo/TangRS16}~\cite{DBLP:journals/eor/FischettiMS18}~\cite{DBLP:journals/informs/FischettiLMS19}.  We also refer the reader to \cite{DBLP:journals/ior/FischettiLMS17}\cite{DBLP:journals/ejco/KleinertLLS21} for an overview on MIBLP solvers and related applications. Most of these algorithmic results are for finding exact solutions through, e.g., branch and bound based approach. For DeNegre's knapsack interdiction problem, an improved exact algorithm was derived by Federico Della Croce and Rosario Scatamacchia \cite{DBLP:journals/mp/CroceS20}.

It is worth mentioning that besides DeNegre's knapsack interdiction problem (i.e., $s_A=s_B=1$ in IPC), other variants of {bilevel knapsack problems} have also been studied in which the leader interferes the follower's program in a different way. One kind of bilevel knapsack problem was introduced by Dempe and Richter~\cite{dempe2000bilevel} where two players hold one knapsack, the leader determines the knapsack's capacity while the follower picks items into the knapsack to maximize his own total profit. The goal is to maximize the objective of the leader. Brotcorne et al.~\cite{DBLP:journals/disopt/BrotcorneHM13} gave a dynamic programming algorithm for both cases of this model. Chen and Zhang~\cite{DBLP:journals/tcs/ChenZ13} proposed a bilevel knapsack variant where two players hold their own knapsacks and the leader can only influence the profit of the items. The follower is interested in his own revenue while the leader aims at maximizing the total profit of both players. The improved approximation results for this problem were derived by Xian Qiu and Walter Kern~\cite{DBLP:journals/tcs/QiuK15}. Another bilevel knapsack variant occurred in the work of Pferschy et al.~\cite{DBLP:journals/tcs/PferschyNP19} where the leader controls the weights of a subset of the follower’s items and the follower aims at maximizing his own profit. The leader's payoff is the total weight of the items he controls and selected by the follower. Very recently, Pferschy et al.\cite{pferschy2021stackelberg} tackled a “symmetrical” problem in which the leader can control the profits instead of item weights. In addition to these works, a matrix interdiction problem was studied by Kasiviswanathan and Pan~\cite{kasiviswanathan2010matrix}.

It is also worth mentioning that the continuous version of DeNegre's knapsack interdiction problem, where the leader and the follower can both fractionally choose an item, has also been studied in recent years. Carvalho et al.~\cite{DBLP:journals/orl/CarvalhoLM18} gave the first polynomial time optimal algorithm. Later on, a faster optimal algorithm was proposed by Woeginger and Fischer~\cite{DBLP:journals/orl/FischerW20}.

\smallskip
\noindent\textbf{Some notations.}
We write column vectors in boldface, e.g. $\vex,\vey$, and their entries in normal font. For a vector $\vex$, we either denote its entries by $\vex=(x_1,x_2,\cdots,x_n)$, or by $\vex=(\vex[1],\vex[2],\cdots,\vex[n])$. Given two vectors $\vex$ and $\vey$ with the same dimension, their dot product is $\sum_{j}x_{j}y_{j}$, which we denote by $\vex\vey$.

\section{Hardness results}\label{sec:hardness}
    \begin{theorem}\label{thm:hardness}
    	Assuming $P\neq NP$, for arbitrary small $\epsilon>0$, there does not exist a $(4/3-\epsilon)$-approximation polynomial time algorithm for IPC when $s_A\ge 1$ and $s_B\ge 2$.
    \end{theorem}
     Towards the proof, we need the 3 hitting set (3HS) problem.

    Problem: 3 Hitting Set

    Instance: A ground set $U=\{u_1,u_2,\cdots,u_n\}$; a collection $C$ of $m$ subsets $S_1,S_2,\cdots,S_m$ whose union is $U$, where each subset $S_h$ contains exactly 3 elements; a positive integer $k$.

    Question: Is there a hitting subset $S\subseteq U$ such that $|S|\le k$, and $S$ contains at least one element from each subset in $C$?

\begin{proof}
Recall that 3HS problem is a natural generalization of the well-known Vertex Cover problem, and both are NP-complete \cite{hartmanis1982computers}. Our reduction is from the 3HS problem. Given an instance of the 3HS, we construct an instance of the IPC where $s_A=1$, $s_B=2$ as follows. Let $E=10\cdot\sum_{i=1}^n10^i$, and $Q$ be any sufficiently large integer, say, $Q=10E$. Let $\vea=k$ and $\veb=(E,4Q-E)$. The profit of every item constructed below is $1$. For every element $u_i$, we construct an {\it element-item} (item $i$) whose interdiction cost is $1$, and whose weight vector is $(10^i,Q-10^i)$. For every subset $S_h=\{u_i,u_j,u_k\}$, we construct a {\it set-item} (item $n+h$) whose interdiction cost is $k+1$ (that is, the leader cannot interdict a set-item), and whose weight vector is $(E-10^i-10^j-10^k,Q-E+10^i+10^j+10^k)$. In total we construct $n+m$ items.
	
We first claim that the objective value of any feasible solution for the IPC instance is at most $4$. Suppose on the contrary the claim is false, then the follower is able to select at least 5 items under the budget $\veb=(E,4Q-E)$. Notice that for any $1\le i\le n$, $Q-10^i>Q-E\ge 0.9Q$, and for any $1\le i,j,k\le n$ we have $Q-E+10^i+10^j+10^k>0.9Q$, if we sum up the weight vectors of any 5 items, then the second coordinate is at least $4.5Q$, which exceeds the budget $4Q-E$, hence the claim is true.
	
Suppose the 3HS instance admits hitting set $S$ of size at most $k$, we show that the optimal objective value of the IPC instance is at most $3$. Let $S=\{u_{\ell_1},u_{\ell_2},\cdots,u_{\ell_k}\}$ (if $S$ contains less than $k$ elements, we simply add arbitrary elements to make it contain exactly $k$ items), then we consider the solution $\vex$ where $x_{\ell_i}=1$ for $1\le i\le k$, and $x_j=0$ otherwise. We claim that for any $\vey$ satisfying $\vex+\vey\le \ve 1$, $\vep\vey=\sum_{j} y_j\le 3$. Suppose on the contrary that the claim is false, then the follower can select at least $4$, and hence exactly 4 items (given our claim in the above paragraph that shows the objective value cannot exceed 4). Notice that for every item, if we add the first and second coordinate of its weight vector, then the sum is exactly $Q$. Hence, if we add up the weight vector of the 4 items, it must be $(z,4Q-z)$ for some $z$, and meanwhile, we have $(z,4Q-z)\le (E,4Q-E)$, that is $z\le E$ and $4Q-z\le 4Q-E$. Hence, $z=E$, which means the sum of the first coordinate of the weight vectors of the 4 items is exactly $E=10\sum_i 10^i$. We first observe that it is impossible for the 4 items to be all element-items, this is because the first coordinate of the weight vector for any element-item is at most $10^n<0.1E$. We then observe that there cannot be two set-items among the 4 items, because the first coordinate of the weight vector for any set-item is at least $E-0.1E=0.9E$. Hence, among the 4 items, there must be exactly 1 set-item and 3 element-items. Let the 3 element-items be those corresponding to $u_i,u_j,u_k$ and the set-item be the one corresponding to $\{u_{i'},u_{j'},u_{k'}\}$, then it follows that $10^i+10^j+10^k+E-10^{i'}-10^{j'}-10^{k'}=E$, implying that $\{i,j,k\}=\{i',j',k'\}$. However, this is not possible because the hitting set $S$ contains at least one element from $\{u_{i'},u_{j'},u_{k'}\}=\{u_{i},u_{j},u_{k}\}$, which implies that $x_i+x_j+x_k\ge 1$, and whereas $y_i$, $y_j$, $y_k$ cannot be 1 simultaneously. Thus, the optimal objective value of the IPC instance is at most 3.
	
Suppose the optimal objective value of the IPC instance is at most $3$, we show that the 3HS problem admits a hitting set of size at most $k$. Let $\vex^*$ be the optimal solution for IPC. Consider the set $S^*=\{u_i:x^{*}_i=1\}$. Given that $\sum_ix_i\le k$, we know $|S^*|\le k$. We claim that $S^*$ is a hitting set. Suppose on the contrary that the claim is false, then there exists some subset $\{u_i,u_j,u_k\}$ such that $S^*\cap \{u_i,u_j,u_k\}=\emptyset$. Then we consider the 3 element-items whose weight vectors are $(10^i,Q-10^i)$, $(10^j,Q-10^j)$, $(10^k,Q-10^k)$, and the set-item whose weight vector is $(E-10^i-10^j-10^k,Q-E+10^i+10^j+10^k)$. It is easy to see that the follower can select all the 4 items, leading to an objective value of 4, contradicting the fact that the optimal objective value is at most 3. 
	
Now suppose there exists a $(4/3-\epsilon)$-approximation polynomial time algorithm for the IPC. We apply the algorithm to the IPC instance constructed from the 3HS instance. If the 3HS instance admits a hitting set of size at most $k$, then the approximation algorithm returns a solution with objective value at most $4-\epsilon<4$, which means it must return a solution with objective value at most $3$. If the 3HS instance does not admit a hitting set of size at most $k$, then the approximation algorithm returns a solution with objective value at least $4$. Hence the polynomial time approximation algorithm can be used to determine whether 3HS problem admits a feasible solution, contradicting the NP-hardness of 3HS problem.
\end{proof}

\section{A PTAS for IPC where $s_{B}=1$ and $s_{A}$ is a fixed constant}\label{sec:constant_KIP_1}
The goal of this section is to prove the following Theorem~\ref{theorem:s_B=1}. Theorem~\ref{theorem:s_B=1} is a special case of our main result, however, its proof shares similar key ideas as the general case (where $s_A$ and $s_B$ are arbitrary fixed constants). Therefore, we provide a full presentation to demonstrate the technical insights, and in the next section we will show how to extend the techniques when $s_B\ge 2$.



\begin{theorem}\label{theorem:s_B=1}
When $s_B = 1$ and $s_A$ is an arbitrary fixed constant, there exists a polynomial time approximation scheme for IPC.
\end{theorem}

The rest of this section is dedicated to proving the following Lemma~\ref{lemma:residue-main_1}, which implies Theorem~\ref{theorem:s_B=1} directly by scaling item profits (here we write \textbf{IPC}$(I,\vea,b)$ instead of \textbf{IPC}$(I,\vea,\veb)$ as $\veb$ becomes 1-dimensional {given that $s_B=1$}).
\begin{lemma}\label{lemma:residue-main_1}
Let $OPT$ be the optimal objective value of \textbf{IPC}$(I,\vea,b)$. If $OPT\le 1$, then for an arbitrarily small number $\epsilon>0$, there exists a polynomial time algorithm that returns a feasible solution to \textbf{IPC}$(I,\vea,b)$ with an objective value of at most $1+\OO(\epsilon)$.
\end{lemma}

\subsection{Preprocessing}\label{subsec:prepro}
From now on we assume $OPT\le 1$. Without loss of generality, we further assume that $\max_{j}p_j \le 1$.

\smallskip
\noindent\textbf{Scaling.}  We scale the matrix $\ve A$ and $\ve B$ such that $\vea=\ve1$ and $b=1$. From now on we denote this IPC instance as \textbf{IPC}$(I,\ve1,1)$. Without loss of generality, we further assume that $\max_{j} B_j \le 1$.

\smallskip
\noindent\textbf{Rounding down the profits.}
We apply the standard geometric rounding. Let $\delta>0$ be some small parameter to be fixed later (in particular, we can choose $\delta=\epsilon^2$). Consider each item profit $p_j$. If $p_j \le \delta^2$, we keep it as it is; otherwise $p_j>\delta^2$, we round the profit down to the largest value of the form $\delta^2(1+\delta)^h$. For profits whose values are at least $\delta^2$, simple calculation shows there are at most $\tilde\OO(1/\delta)$ distinct rounded profits. This rounding scheme introduces an additive loss of at most $\OO(\delta)$ times the objective value. For simplicity, we still denote the rounded profits by $p_j$'s. 

\smallskip
\noindent\textbf{Item classification.} Recall that each item $j$ is associated with a profit $p_j$ and a weight vector $\ve B_j$. Since $s_B=1$, we write $B_j$ as its weight. 


\textbf{Classifying Weights:}
 We say an item $j$ has a {\it large} weight if $ B_j > \delta$; otherwise, it has a {\it small} weight.
 
\textbf{Classifying Profits:} We say an item $j$ has a {\it large} profit if $p_j>\delta$; a {\it medium} profit if $\delta^2<p_j\le \delta$; and a {\it small} profit if $p_j \le \delta^2$.

We say an item is {\em large} if it has a large-profit, or a large-weight. Otherwise, the item is {\em small}. Large items and small items will be handled separately.

Denote by ${S}^*$ the items selected by the leader in an optimal solution of \textbf{IPC}$(I,\ve1,1)$. 

\subsection{Handling Large Items}\label{subsec:large_items}
 \subsubsection{Determining the leader's choice on large items}
The goal of this subsection is to guess large items in ${S}^*$ in polynomial time. 

\smallskip
\noindent\textbf{Large-profit small-weight items.} Notice that if there are at least $1/\delta$ such items for the follower to select, then selecting any $1/\delta$ of them gives a solution with an objective value strictly larger than $1$, contradicting to the assumption that $OPT\le 1$. Thus ${S}^*$ must include all except at most $1/\delta-1$ such items, which can be guessed out via $n^{\OO({1/\delta})}$ enumerations. Hence, we have the following observation.

\begin{observation}\label{obs:lpsw}
	With $n^{\OO({1/\delta})}$ enumerations,  we can guess out all large-profit small-weight items in ${S}^*$. 
\end{observation}

\smallskip
\noindent\textbf{Small-profit large-weight items.} Notice that the follower can select at most $1/\delta$ items from this subgroup and their total profit is at most $\delta^2*\frac{1}{\delta} = \delta$. Hence, even if the leader does not select any such item, the objective value can increase by at most $\delta$, which leads to the following observation.

\begin{observation}\label{obs:splw}
With $\OO(\delta)$ additive error, we may assume that ${S}^*$ does not contain small-profit large-weight items. 
\end{observation}

\smallskip
\noindent\textbf{Large/medium-profit large-weight items.} Notice that the follower can select at most $1/\delta$ items from this group. Since we are considering the case of $s_B = 1$, if there are two items that are not selected by the leader, and they have the same profit, then the follower always prefers the one with a smaller weight. Hence, we have the following lemma.

\begin{lemma}\label{lemma:lmplw2}
With $n^{\tilde\OO(1/\delta^{2})}$ enumerations, we can guess out all large/medium-profit large-weight items in ${S}^*$. 
\end{lemma}
\begin{proof}
Recall that there are at most $\tilde{\OO}(1/\delta)$ distinct large/medium profits. Let $S_h$ be the set of large-weight items whose profits are all $\delta^2(1+\delta)^h$. We observe two facts: (i). Among items in $S_h\setminus S^*$, the follower always selects the ones with the smallest weights; (ii). The follower can select at most $1/\delta$ items from $S_h\setminus S^*$. We claim that, $S^*\cap S_h$ can be determined through guessing out the following $1/\delta$ key items in $S_h$: among items in $S_h\setminus S^*$, which is the item that has the $k$-th smallest weight for $k=1,2,\cdots,1/\delta$?\footnote{If there are less than $1/\delta$ items in $S_{h}$, we can simply guess out all items in $S_{h}\setminus S^*$ via $n^{\OO(\frac{1}{\delta})}$ enumerations.} To see the claim, let $w_{h}^{max}$ be the weight of the item in $S_h\setminus S^*$ that has the $1/\delta$-th smallest weight. Consider any item in $S_h$: if its weight is smaller than $w_{h}^{max}$, and is not one of the key items, then this item must belong to ${S}^*$ (by the definition of key items); if its weight is larger than or equal to $w_{h}^{max}$, and is not one of the key items, then it is not in ${S}^*$ (there is no need for the leader to select such an item since the follower will never select this item even if it is available). Thus via 
$n^{\OO{(1/\delta)}}$ enumerations, we can guess out all large/medium-profit large-weight items in ${S}^* \cap S_h$. Moreover, via total $n^{\tilde\OO{(1/\delta^2)}}$ enumerations, we can guess out all large/medium-profit large-weight items in ${S}^*$.   
\end{proof}

To summarize, our above analysis leads to the following lemma:
\begin{lemma}
With $\OO(\delta)$ additive error, we can guess out all the large items in the optimal solution ${S}^*$, i.e., all items that either have a large profit or a large weight, by $n^{\tilde{\OO}(\frac{1}{\delta^{2}})}$ enumerations.
\end{lemma}


Let $\bar{I}\subseteq I =\{1,2,\cdots,n\}$ be the set of small items, i.e., items of medium/small-profit and small-weight. Then $I\setminus \bar{I}$ is the set of large items. Denote by $\vex^*$ the optimal solution to \textbf{IPC}$(I,\ve1,1)$, which is corresponding to $S^*$. In the following we assume a correct guess on large items. Hence, the values of $\{x_j^*:j\in I\setminus \bar{I} \}$ are known. We let $\vea'$ be the total cost of these guessed-out large items. 
 
\subsubsection{Finding the follower's dominant choices on large items.}\label{subsection:connect}
Consider all the large items. Even if the leader's choice on large items is fixed, the follower may still have exponentially many different choices on the remaining large items. The goal of this subsection is to show that, among these choices of the follower, it suffices to restrict our attention to a few "dominant" choices that always outperform other choices.

For simplicity, we re-index items such that $\bar{I}=\{1,2,\cdots,\bar{n}\}$, where $\bar{n}\le n$. 
We further assume that items in $\bar{I}$ are sorted in decreasing order of the profit-weight ratios $p_j/B_j$. For any $\bar{\vea}\le \ve1$ and $\bar{b}\le 1$, denote by \textbf{IPC}$(\bar{I},\bar{\vea},\bar{b})$ the "residual instance" where the item set is $\bar{I}$, the budget vector of the leader is $\bar{\vea}$ and the budget of the follower is $\bar{b}$. 

Denote by $I^{\prime}\subseteq I\setminus \bar{I}$ the subset of large items which are {\it not} selected by the leader. Note that due to the assumption $OPT\le 1$ and that we have guessed out correct large items in $S^*$, the follower cannot select items from $I^{\prime}$ with total profit larger than 1. Hence, for each integer $k\in [1, 1+1/\epsilon]$, we can define the following sub-problem: among items in $I^{\prime}$, find out a subset of items with minimal total weight such that their total profit is within $[(k-1)\epsilon,k\epsilon)$. Denote by $SP(k)$ this sub-problem and by $KP(k\epsilon)$ its optimal solution, if it exists. We claim that the follower can select at most $\OO(1/\delta)$ items from $I^{\prime}$, thus via $n^{\OO({1}/{\delta})}$ enumerations, we can return $KP(k\epsilon)$ or assert there does not exist a feasible solution to $SP(k)$. The claim is guaranteed by the following two facts: (i). The total profit the follower could obtain from $I^{\prime}$ is at most 1; (ii). Items in $I^{\prime}$ either have a large-profit, or a large-weight.  
Let $\Theta = \{KP(k\epsilon): k \in \{1,2,\cdots,1+\frac{1}{\epsilon}\} \}$, which contains the follower's $\OO(1/\epsilon)$ possible choices on $I^{\prime}$. For $\ell \in \{1,2,\cdots,1+\frac{1}{\epsilon}\}$, we let $b_\ell$ and $P_\ell$ be the total weight and the total profit of the items selected by the follower, respectively\footnote{If there is no feasible solution to $SP(\ell)$, we let $b_\ell=1$ and $P_\ell=0$.}. Then the follower has a residual budget of $1-b_\ell$ for items in $\bar{I}$. Recall that the leader has a residual budget vector of $\ve1-\vea'$ for items in $\bar{I}$.  

Define $\vey[\bar{I} ] = (y_1,y_2\cdots, y_{\bar{n}}) $. Recall that by guessing we already know the value of $x_j^*$ for $j\in I\setminus \bar{I}$. Consider the following bilevel program:
\begin{subequations}
	\begin{align}
	\textbf{Bi-IP}(I,\ve1,1):
	\min_{\vex }\hspace{1mm} & P_{\ell}+ \sum_{j=1}^{\bar{n}}p_{j}y_{j}& \nonumber \\
	s.t. \hspace{1mm}& \sum_{j=1}^{\bar{n}}\ve A_{j}x_{j}\le \ve1-\vea' &\label{Bi-IP:a1}\\
	&x_j = x_j^*, \hspace{2mm} \forall j\in I\setminus \bar{I}& \label{Bi-IP:b1}\\
	& x_j \in \{0,1\}, \hspace{2mm} \forall j\in \bar{I} & \label{Bi-IP:c1}\\
	& \text{where integer }\ell, \vey[\bar{I} ] \text{ solves the following:}& \nonumber\\
	&\max_{1\le \ell \le 1+\frac{1}{\epsilon}} \max_{\vey[\bar{I} ]} \quad P_{\ell}+ \sum_{j=1}^{\bar{n}}p_{j}y_{j} & \label{Bi-IP:d}\\
	& s.t. \quad\hspace{10mm}  \sum_{j=1}^{\bar{n}} B_{j}y_{j}\le 1-b_\ell &\label{Bi-IP:a2}\\
	&\hspace{20mm} y_{j} \le 1-x_j, \hspace{2mm} \forall j\in \bar{I} \label{Bi-IP:b2} \\
	& \hspace{20mm}  y_j\in \{0,1\} , \hspace{2mm} \forall j\in \bar{I} \label{Bi-IP:c2}&
	\end{align}
\end{subequations}

\noindent\textbf{What is the difference between \textbf{Bi-IP}$(I,\ve1,1)$ and \textbf{IPC}$(I,\ve1,1)$, assuming the correct guess of $x_j^*$ for $j\in I \setminus \bar{I}$}? In \textbf{Bi-IP}$(I,\ve1,1)$, the follower's choices on remaining large items are restricted to the $\OO(1/\epsilon)$ choices in $\Theta$, while in \textbf{IPC}$(I,\ve1,1)$, the follower can choose any remaining large items. However, we observe that $\Theta$ contains all the follower's "dominant choices of remaining large items" in the sense that the follower uses the smallest budget to achieve a profit within $[(k-1)\epsilon,k\epsilon)$. Consequently, the objective value of \textbf{Bi-IP}$(I,\ve1,1)$ differs by at most $\epsilon$ to that of \textbf{IPC}$(I,\ve1,1)$. A formal description is given below.


\begin{lemma}\label{claim:s_1}

Let $\bar{\vex}$ be any feasible solution to \textbf{Bi-IP}$(I,\ve1,1)$. Then $\bar{\vex}$ is also feasible to \textbf{IPC}$(I,\ve1,1)$. Let ${Obj}_{Bi}(\bar{\vex})$ and $Obj({\bar{\vex}})$ be the objective values of \textbf{Bi-IP}$(I,\ve1,1)$ and \textbf{IPC}$(I,\ve1,1)$ for $\vex = \bar{\vex}$, respectively. We have 
$${Obj}_{Bi}({\bar{\vex}}) \le Obj({\bar{\vex}}) \le {Obj}_{Bi}({\bar{\vex}}) +\epsilon.$$
Furthermore, let ${OPT}_{Bi}$ and $OPT$ be the optimal objective values of \textbf{Bi-IP}$(I,\ve1,1)$ and \textbf{IPC}$(I,\ve1,1)$, respectively, then we have $${OPT}_{Bi} \le OPT \le {OPT}_{Bi}+\epsilon.$$
\end{lemma}
\begin{proof}
Compare the follower's possible choices in \textbf{Bi-IP}$(I,\ve1,1)$ and $\textbf{IPC}(I,\ve1,1)$ when the leader's solution is fixed to $\bar{\vex}$. It is easy to see that in $\textbf{IPC}(I,\ve1,1)$, the follower's feasible choices on the remaining large items contain $\Theta$, it thus follows that  ${Obj}_{Bi}({\bar{\vex}}) \le Obj({\bar{\vex}})$. Particularly, since the optimal solution $\vex^*$ of $\textbf{IPC}(I,\ve1,1)$ is a feasible solution of \textbf{Bi-IP}$(I,\ve1,1)$ 
and the optimal solution of \textbf{Bi-IP}$(I,\ve1,1)$ may achieve an even smaller value, it follows that ${OPT}_{Bi} \le OPT$. It remains to prove that $Obj(\bar{\vex}) \le {Obj}_{Bi}(\bar{\vex}) +\epsilon$ and $OPT \le {OPT}_{Bi} +\epsilon$.

Note that $Obj(\bar{\vex})$ is exactly the optimal objective value of the following integer program:
	\begin{align*}
	\textbf{IP($\bar{\vex}$)}:
	\max\limits _{\vey}\quad& \vep\vey&\\
	 s.t.  \hspace{1mm}&\sum_{j=1}^{n}B_jy_j\le 1&\\
	&\hspace{2mm} \vey\le \ve1- \bar{\vex} \\
	& \hspace{2mm}\vey\in\{0,1\}^n&
	\end{align*}
Let $\bar{\vey}$ be an optimal solution of \textbf{IP}($\bar{\vex}$), then $Obj(\bar{\vex}) = \sum_{j\in{I\backslash \bar{I}}}p_j\bar{y}_j+\sum_{j\in{\bar{I}}}p_j\bar{y}_j$. Recall that $\vex^*$ is an optimal solution of \textbf{IPC}$(I,\ve1,1)$, and $\bar{\vex}_j=\vex_j^*$ for $j\in I\setminus \bar{I}$ by \eqref{Bi-IP:b1}. Consequently, if we compare the follower in \textbf{IPC}$(I,\ve1,1)$ and the follower in \textbf{Bi-IP}$(I,\ve1,1)$, {the subset of items in $I\setminus \bar{I}$ available for the two followers to select is the same}, and we let this subset be $R=\{j:x_j^*=0,j\in I\setminus \bar{I}\}$. Given that we assume $OPT\le 1$, the maximal profit the follower could obtain from $R$ is at most 1, thus there exists some integer $\bar{\ell} \in [1, 1+1/\epsilon]$ such that $\sum_{j\in I\backslash \bar{I}}p_j\bar{y}_j  \in [(\bar{\ell}-1)\epsilon, \bar{\ell}\epsilon)$. 
By the definitions of $P_{\bar{\ell}}$ and $b_{\bar{\ell}}$, we have $\sum_{j\in I\backslash \bar{I}}p_j\bar{y}_j \le P_{\bar{\ell}}+\epsilon$ and $b_{\bar{\ell} }\le \sum_{j \in I\backslash \bar{I}}B_j\bar{y}_j$. Define $\vey^{\prime}\in\{0,1\}^{n}$ such that the $\vey^{\prime}$ is a combination of two partial solutions: in $I\setminus\bar{I}$, $\vey^{\prime}$ is the same as $KP(\bar{\ell}\epsilon)$; and in $\bar{I}$, $\vey^{\prime}$ is the same as $\bar{\vey}$.
Then $\vey^{\prime}$ is a feasible solution of the following program:
	\begin{align*}
	\overline{\textbf{IP}}(\bar{\vex}):
	&\max_{\ell} \max_{\vey[\bar{I} ]} \quad P_{\ell}+ \sum_{j=1}^{\bar{n}}p_{j}y_{j} &\\
	& \quad s.t. \hspace{10mm}  \sum_{j=1}^{\bar{n}} B_{j}y_{j}\le 1-b_\ell &\\
	&\hspace{19mm} y_{j} \le 1-\bar{x}_j, \hspace{2mm}\forall j\in \bar{I} \\
	& \hspace{19mm}  y_j\in \{0,1\} , \hspace{2mm}\forall j\in \bar{I} &
	\end{align*}
Notice that the optimal objective value of $\overline{\textbf{IP}}(\bar{\vex})$ is ${Obj}_{Bi}(\bar{\vex})$, thus $\vep \vey^{\prime} = P_{\bar{\ell}}+ \sum_{j=1}^{\bar{n}}p_{j}\bar{y}_{j} \le {Obj}_{Bi}(\bar{\vex})$. To conclude, we have
$$
Obj(\bar{\vex}) = \sum_{j\in{I\backslash \bar{I}}}p_j\bar{y}_j+\sum_{j=1}^{\bar{n}}p_j\bar{y}_j \le P_{\bar{\ell}}+\epsilon+\sum_{j=1}^{\bar{n}}p_{j}\bar{y}_{j} \le {Obj}_{Bi}(\bar{\vex}) + \epsilon
$$
Particularly, given an optimal solution $\bar{\vex}^*$ of \textbf{Bi-IP}$(I,\ve1,1)$, we have $Obj(\bar{\vex}^*) \le {OPT}_{Bi} +\epsilon$. Since the optimal solution of $\textbf{IPC}(I,\ve1,1)$ may achieve an even smaller objective value, it follows that $OPT \le {OPT}_{Bi} +\epsilon$. Hence Lemma~\ref{claim:s_1} is proved.
\end{proof}

\subsection{Handling Small Items}\label{subsection:small-items_1}
According to Lemma~\ref{claim:s_1}, to solve \textbf{IPC}$(I,\ve1,1)$, it suffices to solve \textbf{Bi-IP}$(I,\ve1,1)$, which is the goal of
this subsection. Towards this, we first obtain a linear relaxation of \textbf{Bi-IP}$(I,\ve1,1)$ where both the leader and the follower can select items fractionally. Then we reformulate this bilevel linear relaxation as a single level linear program and find an extreme point optimal fractional solution. {Finally we round this fractional solution to obtain a feasible solution to \textbf{Bi-IP}$(I,\ve1,1)$ with an objective value of at most $1+\OO(\epsilon)$, which is thus also a feasible solution to \textbf{IPC}$(I,\ve1,1)$ with an objective value of at most $1+\OO(\epsilon)$.}

Replace (\ref{Bi-IP:c1}) and (\ref{Bi-IP:c2}) in \textbf{Bi-IP}$(I,\ve1,1)$ with $x_j \in [0,1]( \forall j\in \bar{I})$ and $y_j\in [0,1] (\forall j\in \bar{I})$, respectively, we obtain a relaxation of \textbf{Bi-IP}$(I,\ve1,1)$ as follows.
\begin{subequations}
	\begin{align}
	\textbf{Bi-IP}_r(I,\ve1,1):
	\min_{\vex }\hspace{1mm} & P_{\ell}+ \sum_{j=1}^{\bar{n}}p_{j}y_{j}&\nonumber \\
	s.t. \hspace{1mm}& \sum_{j=1}^{\bar{n}}\ve A_{j}x_{j}\le \ve1-\vea' & \label{Bi-IP_r:a1}\\
	&x_j = x_j^*, \hspace{2mm} \forall j\in I\setminus \bar{I}& \label{Bi-IP_r:b1}\\
	& x_j \in [0,1], \hspace{2mm} \forall j\in \bar{I} &\label{Bi-IP_r:c1}\\
	& \text{where integer } \ell, 
	\vey[\bar{I} ] \text{ solves the following:}&\nonumber\\
	&\max_{1\le \ell \le 1+\frac{1}{\epsilon}} \max_{\vey[\bar{I} ]} \quad P_{\ell}+ \sum_{j=1}^{\bar{n}}p_{j}y_{j} &\label{Bi-IP_r:d}\\
	& s.t. \quad\hspace{10mm}  \sum_{j=1}^{\bar{n}} B_{j}y_{j}\le 1-b_\ell &\label{Bi-IP_r:a2}\\
	&\hspace{20mm} y_{j} \le 1-x_j, \hspace{2mm} \forall j\in \bar{I} \label{Bi-IP_r:b2}\\
	& \hspace{20mm}  y_j\in [0,1], \hspace{2mm} \forall j\in \bar{I} \label{Bi-IP_r:c2}&
	\end{align}
\end{subequations}
Denote by ${OPT}_{Bi}^r$ the optimal objective value of \textbf{Bi-IP}$_r(I,\ve1,1)$. Note that items in $\bar{I}$ are sorted in decreasing order of the profit-weight ratios $p_j/B_j$. Consider any fixed leader's solution $\vex\in [0,1]^n$ and any fixed $\ell$, the follower is solving a knapsack problem in the remaining (fractional) items. The maximal objective value of the follower, given $\vex\in [0,1]^n$ and $\ell$, is obtained by a simple greedy algorithm that selects remaining fractional items in $\bar{I}$ in the natural order of indices (recall that items are re-indexed in non-increasing order of ratios), until the budget $1-b_\ell$ is exhausted. Note that the greedy algorithm will stop at some (fractional) item when the budget $1-b_\ell$ is exhausted\footnote{The greedy algorithm may pack all remaining (fractional) items without using up the budget $1-b_\ell$. To patch this case, we add a dummy item whose profit is $0$, cost vector is $\ve 0$ and weight is sufficiently large. We assume the last item $\bar{n}$ is the dummy item.}. We say this item is {\it critical} and let its index be $c_{\ell}$. 
Given any fixed $\vex\in [0,1]^n$ and $\ell$, the maximal objective value of program (\ref{Bi-IP_r:d})-(\ref{Bi-IP_r:c2}) for $\vex$ is 
\begin{subequations}
	\begin{align}
		P_{\ell}+\sum_{j=1}^{c_{\ell}-1}p_j(1-x_j)+p_{c_{\ell}}\frac{1-b_{\ell}-\sum_{j=1}^{c_{\ell}-1}B_j(1-x_j)}{B_{c_{\ell}}}, \label{refor_1}
	\end{align}
\end{subequations}
where $c_{\ell}$ is the critical item given $\vex$ and $\ell$. The following two formulas are directly given by the definition of critical.
\begin{subequations}
	\begin{align}
		& \sum_{j=1}^{c_{\ell}-1}B_j(1-x_j) \le 1-b_\ell \label{cons_1}\\
		&B_{c_{\ell}}+ \sum_{j=1}^{c_{\ell}-1}B_j(1-x_j) \ge 1-b_\ell \label{cons_2}
	\end{align}
\end{subequations}

{We first show that the optimal objective value of the \textbf{Bi-IP}$_{r}(I,\ve1,1)$ is at most ${OPT}_{Bi}+\delta$.}

\begin{lemma}\label{claim:ptas_1}
Let ${OPT}_{Bi}$ and ${OPT}_{Bi}^r$ be the optimal objective values of \textbf{Bi-IP}$(I,\ve1,1)$ and \textbf{Bi-IP}$_r(I,\ve1,1)$, respectively, then ${OPT}_{Bi}^r\le {OPT}_{Bi}+\delta$.
\end{lemma}
\begin{proof}
It is not straightforward if we compare \textbf{Bi-IP}$(I,\ve1,1)$ with \textbf{Bi-IP}$_r(I,\ve1,1)$ directly, as both the follower and the leader become stronger in the relaxation (in the sense they can pack items fractionally). Towards this, we introduce an intermediate bilevel program \textbf{Bi-IP}$_{in}(I,\ve1,1)$, which is obtained by replacing (\ref{Bi-IP_r:c1}) in \textbf{Bi-IP}$_r(I,\ve1,1)$ with $x_j \in \{0,1\}(\forall j\in \bar{I})$, that is, we only allow the follower to select items fractionally but not the leader. Denote by $OPT_{in}$ the optimal objective value of \textbf{Bi-IP}$_{in}(I,\ve1,1)$. 

First, we compare \textbf{Bi-IP}$_{in}(I,\ve1,1)$ with \textbf{Bi-IP}$_r(I,\ve1,1)$. We see that in \textbf{Bi-IP}$_r(I,\ve1,1)$ the follower is facing a  stronger leader who is allowed to fractionally pack items, and it thus follows that ${OPT}_{Bi}^r \le OPT_{in}$. 

Next, we compare \textbf{Bi-IP}$_{in}(I,\ve1,1)$ and \textbf{Bi-IP}$(I,\ve1,1)$. Note that the leader's solution must be integral in both programs. Any feasible solution of \textbf{Bi-IP}$_{in}(I,\ve1,1)$ is a feasible solution of \textbf{Bi-IP}$(I,\ve1,1)$, and vice versa.
Let $\bar{\vex}^*$ be an optimal solution of \textbf{Bi-IP}$(I,\ve1,1)$, then $\bar{\vex}^*$ is also a feasible solution of  \textbf{Bi-IP}$_{in}(I,\ve1,1)$. Once the leader fixes his solution as $\bar{\vex}^*$ in \textbf{Bi-IP}$_{in}(I,\ve1,1)$, there exist $\ell$ and $c_\ell \in \bar{I}$, such that the objective value of \textbf{Bi-IP}$_{in}(I,\ve1,1)$ is 
$$
Obj_{in}=P_{\ell}+\sum_{j=1}^{c_{\ell}-1}p_j(1-\bar{x}^*_j)+p_{c_{\ell}}\frac{1-b_{\ell}-\sum_{j=1}^{c_{\ell}-1}B_j(1-\bar{x}^*_j)}{B_{c_{\ell}}},
$$
where $c_{\ell}$ is the critical item corresponding to $\bar{\vex}^*$ and $\ell$. We have the following two observations:
\begin{itemize}
    \item ${OPT}_{Bi}\ge Obj_{in}-p_{c_{\ell}}\ge Obj_{in}-\delta$. This is because the follower in \textbf{Bi-IP}$(I,\ve1,1)$ can guarantee an objective value of $P_{\ell}+\sum_{j=1}^{c_{\ell}-1}p_j(1-\bar{x}^*_j)$, and $p_{j} \le \delta$ for $j\in \bar{I}$;
    \item $OPT_{in}\le Obj_{in}$. This is because  $\bar{\vex}^*$ is just a feasible solution of \textbf{Bi-IP}$_{in}(I,\ve1,1)$, while the optimal solution of the leader may achieve an even smaller objective value. 
\end{itemize}
To summarize, we know ${OPT}_{Bi}^r \le OPT_{in}\le Obj_{in}\le {OPT}_{Bi}+\delta$.
 Lemma~\ref{claim:ptas_1} is proved.
\end{proof}

Given Lemma~\ref{claim:ptas_1}, we are still facing two questions: how can we solve \textbf{Bi-IP}$_r(I,\ve1,1)$; and even if we obtain a fractional solution to \textbf{Bi-IP}$_r(I,\ve1,1)$, how can we transform it to an integral solution without incurring a huge loss. Towards this, consider the optimal solution $\vex^r$ to \textbf{Bi-IP}$_r(I,\ve1,1)$. Note that leader's choice on large items is guessed out in \textbf{Bi-IP}$_r(I,\ve1,1)$. Consider the scenario when the follower adopts the $\ell$-th dominant choice on the remaining large items, and recall the definition of critical items (see \eqref{refor_1}). Given solution $\vex^r$, for any $\ell\in\{1,2,\cdots,1+\frac{1}{\epsilon}\}$ there must exist a critical item. Therefore, there are $1+\frac{1}{\epsilon}$ critical items corresponding to $\vex^r$. The crucial fact is that, while we cannot guess out $\vex^r$ directly, we can guess out all the critical items corresponding to $\vex^r$. More precisely, with $\bar{n}^{\OO(1/\epsilon)}$ enumerations, we can guess out the critical item $c_\ell^r$ for $\vex^r$ and each $\ell$. Suppose we have guessed out the correct $c_\ell^r$'s corresponding to the optimal solution $\vex^r$, we consider the following LP:  
 \begin{subequations}
	\begin{align*}
		\textbf{LP}_{\text{Bi-IP}}: \quad \min\limits_{\vex,M}&\hspace{2mm} M&\\\
		\hspace{1mm}&  \sum_{j=1}^{\bar{n}}\ve A_{j}x_{j}\le \ve1-\vea' &\\
		\hspace{1mm}& P_{\ell}+\sum_{j=1}^{{c_{\ell}^r}-1}p_j(1-x_j)+p_{{c_{\ell}^r}}\frac{1-b_{\ell}-\sum_{j=1}^{{c_{\ell}^r}-1}B_j(1-x_j)}{B_{{c_{\ell}^r}}} \le M, \hspace{2mm} \forall \ \ell \in \{1,2,\cdots,1+\frac{1}{\epsilon}\} &\\
		& \sum_{j=1}^{{c_{\ell}^r}-1}B_j(1-x_j) \le 1-b_\ell, \hspace{2mm} \forall \ \ \ell \in \{1,2,\cdots,1+\frac{1}{\epsilon}\} &\\
		&B_{{c_{\ell}^r}}+ \sum_{j=1}^{{c_{\ell}^r}-1}B_j(1-x_j) \ge 1-b_\ell, \hspace{2mm} \forall \ \ \ell \in \{1,2,\cdots,1+\frac{1}{\epsilon}\} &\\
		&x_j = x_j^*, \hspace{2mm} j\in I\setminus \bar{I}& \\
		& x_j \in [0,1] , \hspace{2mm} j\in \bar{I}& 
	\end{align*}
\end{subequations}

We have the following simple observation.
\begin{observation}
Let $M^*$ and ${OPT}_{Bi}^r$ be the optimal objective values of $\textbf{LP}_{\text{Bi-IP}}$ and \textbf{Bi-IP}$_r(I,\ve1,1)$, respectively, then $M^*\le {OPT}_{Bi}^r$.
\end{observation}
Let $\vex^r$ be an optimal solution to \textbf{Bi-IP}$_r(I,\ve1,1)$. The observation follows directly as $\vex^r$ together with ${OPT}_{Bi}^r$ form a feasible solution to $\textbf{LP}_{\text{Bi-IP}}$.

In the meantime, we also have the following observation.
\begin{observation}
Let $\{\vex^{ex},M^*\}$ be an extreme point optimal solution to \textbf{LP}$_{\text{Bi-IP}}$, then $\vex^{ex}$ is also a feasible solution to \textbf{Bi-IP}$_r(I,\ve1,1)$ whose objective value is at most $M^*$. 
\end{observation}
The observation follows since by the definition of critical, \eqref{refor_1} is the largest profit the follower can achieve. Hence, when $\vex=\vex^{ex}$ in \textbf{Bi-IP}$_r(I,\ve1,1)$, the objective value is bounded by $M^*$. Given the two observations above, we know $\vex^{ex}$ is an optimal solution to \textbf{Bi-IP}$_r(I,\ve1,1)$ and we have $M^*= {OPT}_{Bi}^r$.
Finally, a near-optimal solution to \textbf{IPC}$(I,\ve1,1)$ can be obtained through the optimal solution to $\textbf{LP}_{\text{Bi-IP}}$, as implied by the following lemma.

\begin{lemma}\label{claim:ptas_2}
Let $\{\vex^{ex},M^*\}$ be an extreme point optimal solution to \textbf{LP}$_{\text{Bi-IP}}$. Define $\tilde{\vex}$ such that $\tilde{x}_j = 1$ if $x_j^{ex}=1$, and $\tilde{x}_j = 0$ otherwise. Then $\tilde{\vex}$ is a feasible solution of \textbf{IPC}$(I,\ve1,1)$ with an objective value of at most $OPT+\OO(\epsilon)$, where $OPT$ is the optimal objective value of \textbf{IPC}$(I,\ve1,1)$.
\end{lemma}
\begin{proof}
The feasibility of $\tilde{\vex}$ to \textbf{IPC}$(I,\ve1,1)$ is straightforward since 
$$\sum_{j=1}^n\ve A_{j}\tilde{x}_{j} = \sum_{j\in I \backslash \bar{I}}\ve A_{j}{x}^*_{j}+\sum_{j\in \bar{I}}\ve A_{j}\tilde{x}_{j} \le \vea'+ \sum_{j\in \bar{I}}\ve A_{j}x^{ex}_{j}\le \vea'+ \ve1-\vea'\le \ve1.
$$

Notice that $\tilde{\vex}$ is also feasible for \textbf{Bi-IP}$(I,\ve1,1)$ and \textbf{Bi-IP}$_r(I,\ve1,1)$. Let  $Obj(\tilde{\vex})$, ${Obj}_{Bi}(\tilde{\vex})$ and ${Obj}_{Bi}^r(\tilde{\vex})$ be the objective values of \textbf{IPC}$(I,\ve1,1)$, \textbf{Bi-IP}$(I,\ve1,1)$ and \textbf{Bi-IP}$_r(I,\ve1,1)$ by taking $\vex = \tilde{\vex}$, respectively. Since in \textbf{Bi-IP}$_r(I,\ve1,1)$, the leader is facing a stronger follower who can select fractional items in $\bar{I}$, it follows that ${Obj }_{Bi}(\tilde{\vex})\le {Obj}_{Bi}^r(\tilde{\vex})$.

Now we compare the objective values of two solutions to \textbf{Bi-IP}$_r(I,\ve1,1)$,  $\vex^{ex}$ and $\tilde{\vex}$. It is easy to see that in $\vex^{ex}$, there are at most $(s_A+\frac{3(1+\epsilon)}{\epsilon})$ variables taking fractional values, and all these variables are in $\{x_j^{ex}:j\in \bar{I}\}$. So the leader in $\vex^{ex}$ selects at most  $(s_A+\frac{3(1+\epsilon)}{\epsilon})$ more items in $\bar{I}$, compared with the leader in $\tilde{\vex}$. Consequently, the follower in $\vex^{ex}$ may select at most  $(s_A+\frac{3(1+\epsilon)}{\epsilon})$ less items in $\bar{I}$, compared with the follower in $\tilde{\vex}$. Given that the objective value of $\vex^{ex}$ is ${OPT}_{Bi}^r=M^*$, we have that
$$
{Obj}_{Bi}^r(\tilde{\vex}) \le  {OPT}_{Bi}^r +(s_A+\frac{3(1+\epsilon)}{\epsilon})\delta.
$$ 
According to Lemma~\ref{claim:s_1} and Lemma~\ref{claim:ptas_1}, we have $Obj(\tilde{\vex}) \le {Obj}_{Bi}(\tilde{\vex})+\epsilon$ and $ {OPT}_{Bi}^r \le {OPT}_{Bi}+\delta$. In conclusion, we have 
$$Obj(\tilde{\vex})  \le {OPT}_{Bi}+(s_A+1+\frac{3(1+\epsilon)}{\epsilon})\delta+\epsilon.
$$
Furthermore, ${OPT}_{Bi} \le OPT$ by Lemma~\ref{claim:s_1}. By setting $\delta = \epsilon^2$,  Lemma~\ref{claim:ptas_2} is proved.
\end{proof}
Hence, if $OPT\le 1$, then a feasible solution with objective value of at most $OPT+\OO(\epsilon)\le 1+\OO(\epsilon)$ is found, thus Lemma~\ref{lemma:residue-main_1} is proved, and Theorem~\ref{theorem:s_B=1} follows.

\section{Approximation algorithm for IPC where $s_B$ and $s_A$ are constant}\label{sec:constant_KIP_2}

In this section, we prove our main result -- Theorem~\ref{theorem:s_B}.

\begin{theorem}\label{theorem:s_B}
When $s_B$ and $s_A$ are fixed constants, for an arbitrarily small number $\epsilon >0$, there exists an $(s_B+\OO(\epsilon))$-approximation polynomial time algorithm for IPC.
\end{theorem}

Similar to the special case, by scaling item profits it suffices to show the following: 

\begin{lemma}\label{lemma:residue-main}
Let $OPT$ be the optimal objective value of \textbf{IPC}$(I,\vea,\veb)$. If $OPT\le 1$, then for an arbitrarily small number $\epsilon>0$, there exists a polynomial time algorithm that returns a feasible solution to $\textbf{IPC}(I,\vea,\veb)$ with an objective value of at most $s_B+\OO(\epsilon)$. 
\end{lemma}
By further scaling the cost vectors and weight vectors, it suffices to find a near-optimal solution for \textbf{IPC}$(I,\ve 1,\ve 1)$. 

\noindent\textbf{Major technical challenge.} Recall that the key to solving IPC for the special case of $s_B=1$ is the establishment of \textbf{Bi-IP}$(I,\ve1,1)$, which is essentially equivalent to \textbf{IPC}$(I,\ve1,1)$. \textbf{Bi-IP}$(I,\ve1,1)$ is built upon the observation that the follower admits only $\OO(\frac{1}{\epsilon})$ dominant choices on large items, where each dominant choice corresponds to the minimal budget needed by the follower to ensure a profit of $[(k-1)\epsilon,k\epsilon)$ where $k\in\{1,2,\cdots,1+\frac{1}{\epsilon}\}$. Because the number of follower's choices on large items is small, its relaxation \textbf{Bi-IP}$_r(I,\ve1,1)$ has a small number of constraints (see \eqref{Bi-IP_r:a2}), and therefore we can further transform \textbf{Bi-IP}$_r(I,\ve1,1)$ to $\textbf{LP}_{\text{Bi-IP}}$ with a small number of constraints whose extreme point solution promises a good rounding. We aim to follow a similar method, however, when $s_B\ge 2$, we can no longer bound the follower's choices on large items. This is because to achieve a profit of $[(k-1)\epsilon,k\epsilon)$ where $k\in\{1,2,\cdots,1+\frac{1}{\epsilon}\}$, the follower may have a huge number of different choices utilizing different budgets, where the budgets are now vectors instead of numbers, and are thus incomparable. To handle the problem, we use the idea of rounding: let $\veb=(\veb[1],\veb[2],\cdots,\veb[s_B])$ and $\ve B_j=(\ve B_j[1],\ve B[2],\cdots, \ve B[s_B])$. We call the first dimension ($\veb[1]$ and $\ve B_j[1]$'s) the principal dimension. The principal dimension will be treated the same as the special case and will {\it not} be rounded. The coordinates of other dimensions ($\veb[h]$ and $\ve B_j[h]$'s for $2\le h\le s_B$) will be rounded. Then, we will be able to compare follower's different choices on large items: if there are two choices both achieving profit within $[(k-1)\epsilon,k\epsilon)$ for the same $k\in \OO(\frac{1}{\epsilon})$, and furthermore, the summation of their weight vectors share the same rounded value in each dimension $h\in [2,s_B]$, then the choice with smaller value in the principal dimension of the summed weight vectors dominates the other choice. By doing so, our argument for the special case can be carried over to the principal dimension.

There is one problem with the idea of rounding above, that is, if we round up weight vectors and meanwhile enlarge the follower's budget vector in dimension $h\in [2,s_B]$ (to accommodate the rounding up), the optimal objective value of IPC may increase. However, we are able to show that the optimal objective value only increases by a factor of $s_B$ (see Lemma~\ref{lemma:augment_budget}). This explains our approximation ratio of $s_B+\epsilon$.  

Below we give a very brief walk through and the reader is referred to Appendix~\ref{appsec:constant_KIP_2} for details.

\noindent\textbf{Step 1}. We pick a small parameter $\delta$ as the rounding precision, {keep the coordinates of weight vectors on principal dimension 
intact, and round the coordinates on other dimensions. We also round the profits.} By doing so we obtain a rounded instance $\tilde{I}_\delta$. Then we pick another small parameter $\tau$ and enlarge the weight budget on dimension $h\in [2,s_B]$ by a factor of $1+\tau$. By doing so we obtain:
\begin{subequations}
	\begin{align*}
	\textbf{IPC}_{\tau}(\tilde{I}_\delta,\ve1,\ve1):\nonumber
	\min\limits_{\vex} \hspace{2mm} & \tilde{\vep}\vey&\\
	s.t. \hspace{1mm}&\ve A\vex\le \ve1& \\
	&\hspace{1mm} \vex \in\{0,1\}^n&\\
	& \text{where }  \vey \text{ solves the following:}&\\
	&\max\limits_{\vey}\quad \tilde{\vep}\vey &\\
	& s.t. \quad\hspace{1mm} \sum_{j=1}^n \tilde{\ve B}_j[1]y_j\le 1&\\
	&\hspace{10mm} \sum_{j=1}^n \tilde{\ve B}_j[i]y_j\le 1+\tau, \hspace{2mm} \forall 2\le i\le s_B&\\
	&\hspace{10mm} \vex+\vey\le \ve1\\
	&\hspace{10mm} \vey\in\{0,1\}^n&
	\end{align*}
\end{subequations}
where $\tilde{\ve B}_j[1]={\ve B_j}[1]$, $\tilde{\ve B}_j[h]$, $2\le h\le s_B$ and $\tilde{\vep}$ are rounded weights and profits. We are able to prove the following lemma which ensures that solving \textbf{IPC}$_{\tau}(\tilde{I}_\delta,\ve1,\ve1)$ gives a good approximate solution to \textbf{IPC}$(I,\ve1,\ve1)$:
\begin{lemma}\label{lemma:r_theta}
Let $0<\tau\le 1/2$. Let $\tilde{\vex}$ be any feasible solution to \textbf{IPC}$_{\tau}(\tilde{I}_\delta,\ve1,\ve1)$. Then $\tilde{\vex}$ is a feasible solution to \textbf{IPC}$(I,\ve1,\ve1)$. Let $\widetilde{Obj}_{\tau}(\tilde{\vex})$ and $Obj(\tilde{\vex})$ be the objective values of \textbf{IPC}$_{\tau}(\tilde{I}_\delta,\ve1,\ve1)$ and \textbf{IPC}$(I,\ve1,\ve1)$ for $\vex = \tilde{\vex}$, respectively. If $2\delta\le \tau\le 1/2$, we have
$$ 
Obj(\tilde{\vex}) \le (1+\delta)\widetilde{Obj}_{\tau}(\tilde{\vex}) \le s_{B} (1+\delta)Obj(\tilde{\vex}).
$$
Furthermore, let $\widetilde{OPT}_{\tau}$ and $OPT$ be the optimal objective values of \textbf{IPC}$_{\tau}(\tilde{I}_\delta,\ve1,\ve1)$ and \textbf{IPC}$(I,\ve1,\ve1)$, respectively. We have 
$$ OPT \le (1+\delta)\widetilde{OPT}_{\tau} \le s_B (1+\delta)OPT.$$ 
\end{lemma}

\noindent\textbf{Step 2}. We handle large items. We first classify item profits into large, medium and small. We then classify item weights into large and small based on the largest coordinate in the weight vector, i.e., $\|\ve B_j\|_{\infty}$. We say an item is large if it has a large weight or a large profit, and small otherwise. Let $S^*$ be the {leader's optimal solution in \textbf{IPC}$_{\tau}(\tilde{I}_\delta,\ve1,\ve1)$.} Using a similar argument as the special case, we can prove the following.

\begin{lemma}\label{lemma:guess_l_ge}
With $\OO(s_B\delta)$ additive error, we can guess out all items in ${S}^*$ that have a large weight or a large profit by $n^{\tilde\OO(s_B/\delta^{s_B+1})}$ enumerations. 
\end{lemma}
 Utilizing the fact that coordinates in dimension $h\in [2,s_B]$ are all rounded, we can show that the follower only has a small number (i.e. $\tilde{\OO}(s_B/\epsilon^{s_B})$) of dominant choices on large items, denoted as $\Theta$. By restricting the follower's choices to $\Theta$, we can obtain a new bilevel integer programming $\textbf{MBi-IP}(\tilde{I}_\delta,\ve1,\ve1)$. Similar to \textbf{Bi-IP}$({I},\ve1,\ve1)$ in the special case, $\textbf{MBi-IP}(\tilde{I}_\delta,\ve1,\ve1)$ has a small number of constraints. 

\noindent\textbf{Step 3}. We handle small items. Since $\textbf{MBi-IP}(\tilde{I}_\delta,\ve1,\ve1)$ has a small number of constraints, we remove the integral constraint to obtain a relaxation $\textbf{MBi-IP}_r(\tilde{I}_\delta,\ve1,\ve1)$. Next, we transform this bilevel LP $\textbf{MBi-IP}_r(\tilde{I}_\delta,\ve1,\ve1)$ to a standard (single level) LP, denoted as $\textbf{cen-LP}_{\lambda}$. Note that here the transformation is much more complicated than that in the special case: in the special case we know that if the follower can choose items fractionally, then its optimal fractional solution is always obtained greedily with respect to the ratio (i.e., profit to weight), whereas it suffices to guess one single critical item. In the general case, if the follower can choose items fractionally, we can only guarantee that among all items whose rounded weight vector are the same except for the principal dimension (i.e., $\ve B_j[h]$'s have the same rounded value for every $2\le h\le s_B$), the follower selects items greedily with respect to the principal ratio (i.e., profit to weight coordinate in the principal dimension). Therefore, we need to guess a subset of critical items, and the subscript $\lambda$ in $\textbf{cen-LP}_{\lambda}$ corresponds to a set of parameters characterizing the subset of critical items. The most technical part is to show that the optimal solution to $\textbf{cen-LP}_{\lambda}$ gives a good approximation to $\textbf{MBi-IP}_r(\tilde{I}_\delta,\ve1,\ve1)$ (see Lemma~\ref{lemma:cen-LP}), where we need to create a sequence of "intermediate" LPs. Finally, we obtain an extreme point solution to $\textbf{MBi-IP}_r(\tilde{I}_\delta,\ve1,\ve1)$ by solving $\textbf{cen-LP}_{\lambda}$, and round it to an integral solution. The rounding error can be bounded due to that $\textbf{MBi-IP}_r(\tilde{I}_\delta,\ve1,\ve1)$ contains a small number of constraints.

\section{Approximation algorithm for IPC where $s_B$ and $s_A$ are arbitrary}\label{sec:general_case}
We consider the most general setting of IPC where $s_A$ and $s_B$ are arbitrary (not necessarily polynomial in the input size). 

We define $\max\{\vep\vey:\ve B\vey\le \veb,\vey\in\{0,1\}^{n}\}$ as the {\it follower's problem}. A separation oracle for the {\it leader's problem} is an oracle such that given any $\vex=\vex^0\in [0,1]^{n}$, it either asserts that $\vex^0\in \{\vex:\ve A\vex\le \vea,\vex\in [0,1]^{n}\}$, or returns a violating constraint. The goal of this section is to prove the following theorem. 

\begin{theorem}\label{theorem:general_1}
	Given a separation oracle $O_L$ for the leader's problem, and an oracle $O_F$ for the follower's problem that returns a $\rho$-approximation solution, 
	there exists a $(\frac{\rho}{1-\alpha}, \frac{1}{\alpha})$-bicriteria approximation algorithm for any $\alpha\in (0,1)$ that returns a solution $\vex^*\in\{0,1\}^{n}$ such that $\ve A\vex^*\le \frac{1}{\alpha}\cdot\vea$, and $$\max\{\vep\vey:\ve B\vey\le \veb,\vey\le \ve1-\vex^*,\vey \in \{0,1\}^n\}\le \frac{\rho T^*}{1-\alpha},$$ where $T^*$ is the optimal objective value of \textbf{IPC}$(I,\vea,\veb)$. Furthermore, the algorithm runs in polynomial oracle time.	
\end{theorem}
We omit the proof of Theorem~\ref{theorem:general_1} here, and refer the reader to Appendix~\ref{sec:general_case_app}.

\section{Conclusions}

In this paper, we consider a general two-player zero-sum Stackelberg game in which the leader interdicts some items to minimize the total profit that the follower could obtain from the remaining items. We obtain an $(s_B+\epsilon)$-approximation algorithm when $s_A$ and $s_B$ are both constant, and show that there does not exist any $(4/3-\epsilon)$-approximation algorithm when $s_B\ge 2$. Our algorithm is the best possible when $s_B=1$, however, it is not clear whether it is the best possible when $s_B$ is larger than or equal to $2$. In particular, it is not clear whether an approximation algorithm with a ratio independent of $s_B$ can be obtained. Furthermore, can we hope for a PTAS if $s_B\ge 2$ but the constraints of the leader or the follower are not given by general inequalities but follow from common optimization problems? For example, what if the follower's optimization problem is a bin packing problem? It would be interesting to investigate the bilevel generalization of well-known optimization problems, e.g., scheduling and bin packing.

\clearpage
\appendix

\section{Omitted contents in  Section~\ref{sec:constant_KIP_2}}\label{appsec:constant_KIP_2}
In this section, we aim to prove our main result, Theorem~\ref{theorem:s_B}.
\newtheorem*{T2}{Theorem \ref{theorem:s_B}}
\begin{T2} 
    When $s_B$ and $s_A$ are fixed constants, for an arbitrarily small number $\epsilon >0$, there exists an $(s_B+\OO(\epsilon))$-approximation polynomial time algorithm for IPC. 
\end{T2}

The rest of this section is dedicated to proving the following Lemma~\ref{lemma:residue-main}, which implies Theorem~\ref{theorem:s_B} by scaling item profits. 

\newtheorem*{la}{Lemma \ref{lemma:residue-main}}
\begin{la} 
    Let $OPT$ be the optimal objective value of \textbf{IPC}$(I,\vea,\veb)$. If $OPT\le 1$, then for an arbitrarily small number $\epsilon>0$, there exists a polynomial time algorithm that returns a feasible solution to $\textbf{IPC}(I,\vea,\veb)$ with an objective value of at most $s_B+\OO(\epsilon)$. 
\end{la}

Throughout the proof of Lemma~\ref{lemma:residue-main}, we consider the \textbf{IPC}$(I,\vea,\veb)$ whose optimal objective value is at most 1. Since the case of $s_B = 1$ has been considered in Section~\ref{sec:constant_KIP_1}, in the following, we assume that $s_B\ge2$. 

\subsection{Preprocessing}\label{subsec:prepro_2}
From now on, we assume $OPT \le 1$. Without loss of generality, we further assume that $\max_{j}p_j\le 1$.

\smallskip
\noindent\textbf{Scaling.}
We scale the matrix $\ve A$ and $\ve B$ such that $\vea = \ve1$ and $\veb = \ve1$. From now on we denote this IPC instance as \textbf{IPC}$(I,\ve1,\ve1)$. Without loss of generality, we further assume that $\max_{j} \|\ve B_j\|_{\infty} \le 1$.

\subsubsection{Generating a well-structured rounded item instance .}\label{subsec:decompose}

We aim to replace the original instance $I$ with a well-structured rounded instance $\tilde{I}_\delta$, where $\delta$ is the rounding precision. 

Before giving our rounding scheme, we first present the augment lemma (i.e., Lemma~\ref{lemma:augment_budget}) which measures the loss incurred by slightly increasing the follower's budget. 

\paragraph{Augment lemma}\label{subsec:augment_lemma}
We pick a small parameter $\tau>0$ and enlarge the weight budget on dimension $h\in [2,s_B]$ by a factor of $1+\tau$. By doing so we obtain: 
\begin{align*}
	\textbf{IPC}_\tau(I,\ve1,\ve1):
	\min\limits_{\vex}\hspace{2mm} & \vep\vey&\\
	s.t. \hspace{1mm}&\ve A\vex\le \ve1& \\
	&\hspace{1mm} \vex\in\{0,1\}^n&\\
	& \text{where }  \vey \text{ solves the following:}&\\
	&\max\limits_{\vey}\quad \vep\vey &\\
	& s.t. \hspace{5mm} \sum_{j=1}^n \ve B_j[1]y_j\le 1&\\
	& \hspace{10mm} \sum_{j=1}^n \ve B_j[i]y_j\le 1+\tau, \hspace{2mm} \forall 2\le i\le s_B&\\
	&\hspace{10mm} \vex+\vey\le \ve1 \\
	&\hspace{10mm} \vey\in\{0,1\}^n&
	\end{align*}

That is, \textbf{IPC}$_\tau(I,\ve1,\ve1)$ denotes the instance where the budget vector of the follower is augmented such that value in its first dimension is still $1$, but values in its other dimensions are increased to $1+\tau$. We are interested in to what extent the optimal objective value may change when the budget of the follower is augmented as such. Note that any feasible solution to \textbf{IPC}$_\tau(I,\ve1,\ve1)$ is also a feasible solution to \textbf{IPC}$(I,\ve1,\ve1)$, and vice versa. We have the following lemma.

\begin{lemma}[Augment Lemma]\label{lemma:augment_budget}
 Let $0<\tau\le 1/2 $. Let $\bar{\vex}$ be any feasible solution to \textbf{IPC}$_{\tau}(I,\ve1,\ve1)$. Then $\bar{\vex}$ is a feasible solution to \textbf{IPC}$(I,\ve1,\ve1)$. Let $Obj_{\tau}(\bar{\vex})$ and $Obj(\bar{\vex})$ be the objective values of \textbf{IPC}$_\tau(I,\ve1,\ve1)$ and \textbf{IPC}$(I,\ve1,\ve1)$ for $\vex = \bar{\vex}$, respectively. We have 
 $$ Obj_{\tau}(\bar{\vex}) \le s_B Obj(\bar{\vex}).$$
 
 Furthermore, let $OPT_\tau$ and $OPT$ be the optimal objective values of \textbf{IPC}$_\tau(I,\ve1,\ve1)$ and \textbf{IPC}$(I,\ve1,\ve1)$, respectively, then we have
 
 $$OPT_\tau\le s_B OPT.$$
 
\end{lemma}
\begin{proof}

Let $\bar{\vex}\in\{0,1\}^n$ be a feasible solution to \textbf{IPC}$_{\tau}(I,\ve1,\ve1)$. 
Note that by taking $\vex=\bar{\vex}$, the objective value of \textbf{IPC}$_\tau(I,\ve1,\ve1)$, denote by $Obj_{\tau}(\bar{\vex})$, is exactly the optimal objective value of the following integer program:
		\begin{align*}
		\textbf{IP}_\tau(\bar{\vex}): \quad 	\max\limits_{\vey}\quad  & \vep\vey&\\
		& \sum_{j=1}^n{\ve B}_j[1]y_j\le 1& \nonumber\\
		& \sum_{j=1}^n{\ve B}_j[i]y_j\le 1+\tau, \hspace{2mm} \forall 2\le i\le s_B& \\
		&   \vey \le \ve1-\bar{\vex} & \\
		& \vey\in \{0,1\}^n &
		\end{align*}
Now we consider an optimal solution to \textbf{IP}$_\tau(\bar{\vex})$ and let it be $\vey_\tau(\bar{\vex})\in\{0,1\}^n$. Let $S(\vey_\tau(\bar{\vex}))$ be the set of items selected by the follower in $\vey_\tau(\bar{\vex})$, then the total weight vector of items in $S(\vey_\tau(\bar{\vex}))$ is at most $(1,1+\tau,\cdots,1+\tau)$. We decompose $S(\vey_\tau(\bar{\vex}))$ into $s_B$ subsets $S_1,S_2,\cdots,S_{s_B}$ as follows. 
	
Initially, set $D_0=\{1\}$ and we call dimension 1 the satisfied dimension. We first construct $S_1$ in iteration 1. Consider each item $j\in S(\vey_\tau (\bar{\vex})) $ and its weight vector $\ve B_j$. We remove the $1$-st coordinate of $\ve B_j$ and obtain the projection of vector $\ve B_j$ onto dimensions 2 to $s_B$, and denote it as $\ve B_j^{1}$. If there exists any item $j$ such that $\|\ve B_j^{1}\|_{\infty}\ge \tau$, then $S_1$ consists of this single item. Otherwise, $\|\ve B_j^{1}\|_{\infty}< \tau$ for all items in $ S(\vey_{\tau}(\bar{\vex}))$. Then $S_1$ consists of some $\ell_1$ items $j^1_1,j^1_2,\cdots,j^1_{\ell_1}$ in $ S(\vey_{\tau}(\bar{\vex}))$ such that 
$$
\|\ve \sum_{k=1}^{\ell_{1}-1} \ve B_{j^1_k}^{1}\|_{\infty}< \tau\le \|\ve \sum_{k=1}^{\ell_{1}}\ve B_{j^1_k}^{1}\|_{\infty}.
$$
For each dimension $2\le i\le s_B$, if $\sum_{j\in S_1}\ve B_{j}[i]\ge \tau$, then we say dimension $i$ is satisfied (in iteration 1). Let $D_1\supset D_0$ be the set of all satisfied dimensions so far.
	
Suppose we have constructed $S_1,\cdots,S_h$, and $D_h$ is the set of all satisfied dimensions so far. We now construct $S_{h+1}$ in iteration $h+1$. Let $S_{[1:h]}=\cup_{\iota=1}^{h}S_\iota$. Consider each item $j\in S(\vey_\tau (\bar{\vex}))\backslash S_{[1:h]}$ and its weight vector $\ve B_j$. We remove the $i$-th coordinate of $\ve B_j$ if $i\in D_h$. By doing so we obtain the projection of vector $\ve B_j$ onto all unsatisfied dimensions and denote it as $\ve B_j^{h+1}$. If there exists any item $j\in S(\vey_\tau (\bar{\vex}))\backslash S_{[1:h]}$ such that $\|\ve B_j^{h+1}\|_{\infty}\ge \tau$, then $S_{h+1}$ consists of this single item. Otherwise, $\|\ve B_j^{h+1}\|_{\infty}< \tau$ for all items in $ S(\vey_\tau (\bar{\vex}))\backslash S_{[1:h]}$. Then $S_{h+1}$ consists of some $\ell_{h+1}$ items  $j^{h+1}_1,j^{h+1}_2,\cdots,j^{h+1}_{\ell_{h+1}}$ in $ S(\vey_\tau (\bar{\vex}))\backslash S_{[1:h]}$ such that 
	\begin{eqnarray}\label{eq:greedy}
	\|\ve \sum_{k=1}^{\ell_{h+1}-1} \ve B_{j^{h+1}_k}^{h+1}\|_{\infty}< \tau\le \|\ve \sum_{k=1}^{\ell_{h+1}}\ve B_{j^{h+1}_k}^{h+1}\|_{\infty}.
	\end{eqnarray}
For each dimension $i\not\in D_h$, if $\sum_{j\in S_{h+1}}\ve B_j[i]\ge \tau$, then we say dimension $i$ is satisfied (in iteration $h+1$). Let $D_{h+1}\supset D_h$ be the set of all satisfied dimensions so far.
	
Notice that in each iteration, at least one more dimension becomes satisfied, hence the above procedure stops after at most $s_B-1$ iterations and produces $S_1,S_2,\cdots,S_{s_B-1}$ (if it stops at iteration $h<s_B-1$, we may simply let $S_{h+1}$ to $S_{s_B-1}$ be $\emptyset$). Define $S_{s_B}=S(\vey_\tau (\bar{\vex}))\setminus (\cup_{\iota=1}^{s_B-1}S_\iota)$. 

We claim that, every $S_h$ induces a feasible solution to the following \textbf{IP}$(\bar{\vex}))$ (by setting $y_j=1$ if and only if $j\in S_h$, $y_j=0$ otherwise):
		\begin{align*}
		\textbf{IP}(\bar{\vex}): \quad 	\max\limits_{\vey}\quad  & \vep\vey&\\
		& \ve B\vey\le \ve1& \\
		&   \vey \le \ve1-\bar{\vex} &  \\
		& \vey\in \{0,1\}^n & 
		\end{align*}
To see the claim, recall the construction of each $S_{h+1}(h\ge 0)$. If $S_{h+1}$ only contains a single item, then since each item $j$ satisfies that $\ve B_j\le \ve1$, $S_{h+1}$ obviously gives a feasible solution to \textbf{IP}$(\bar{\vex})$. Otherwise, $S_{h+1}$ contains 2 or more items. It is easy to see that $\sum_{j\in S_{[1:h]}}\ve B_j[1] \le 1$. For every $k\in D_{h} \setminus \{1\}$, we know that $\sum_{j\in S_{[1:h]}}\ve B_j[k]\ge \tau$, and consequently
$$
\sum_{j\in S_{h+1}}\ve B_j[k]\le \sum_{j\in S(\vey_\tau (\bar{\vex}))\setminus S_{[1:h]}}\ve B_j[k]\le  1, \quad\forall k\in D_{h}\setminus \{1\}.
$$
Meanwhile, by the way we construct $S_{h+1}$, $S_{h+1}$ can contain more than one item only if $\|\ve B_j^{h}\|_{\infty}< \tau$ for any $j \in S(\vey_\tau (\bar{\vex}))\setminus S_{[1:h]} $, hence, by Eq~\eqref{eq:greedy}, we know that 
	$$\sum_{j\in S_{h+1}}\ve B_j[k] \le 2\tau \le 1, \quad\forall k\not\in D_{h}.$$
The last inequality holds by the fact that $0<\tau\le 1/2$. Consequently, $\sum_{j\in S_{h+1}}\ve B_j\le \ve1$, implying that $S_{h+1}$ always induces a feasible solution to \textbf{IP}$(\bar{\vex})$. 

Notice that the optimal objective value of \textbf{IP}$(\bar{\vex})$ is $Obj(\bar{\vex})$, whereas the objective value of the solution induced by $S_{h+1}$ is at most $Obj(\bar{\vex})$, implying that the solution induced by $\cup_{h=1}^{s_B} S_h$ is at most $s_B Obj(\bar{\vex})$, that is $Obj_{\tau}(\bar{\vex})\le s_B Obj(\bar{\vex})$. Particularly, given an optimal solution $\vex^*$ of \textbf{IPC}$(I,\ve1,\ve1)$, we have $Obj_{\tau}(\vex^*) \le s_B OPT$. Since the optimal solution of \textbf{IPC}$_\tau(I,\ve1,\ve1)$ may achieve an even smaller objective value, it follows that $OPT_\tau \le s_B OPT$. Hence,  Lemma~\ref{lemma:augment_budget} is proved.\end{proof}

\paragraph{Rounding Scheme}\label{appsec:RS_l}
Now, we are ready to generate a well-structured rounded instance. Let $\delta>0$ be some small parameter to be fixed later (in particular, we can choose $\delta = \epsilon^{2s_B+4}$). We say an item $j\in I$ has a large weight if $\|\ve B_j\|_{\infty}>\delta$, and a small weight otherwise.

\smallskip\noindent\textbf{1. Rounding up the weight vectors.} For each large-weight item $j\in I$, consider its weight vector $\ve B_j=(\ve B_j[1], \ve B_j[2],\dots, \ve B_j[s_B])$. We keep $\ve B_j[1] $ unchanged and round up the coordinate $\ve B_j[i] $ $(2\leq i\leq s_B)$ to the nearest value of the form $\frac{\delta^{2}}{s_B}(1+\delta)^{h}$. More precisely, for each $2\leq i\leq s_B$, if $\ve B_j[i] \le \frac{\delta^{2}}{s_B}$, we round it up to $\frac{\delta^{2}}{s_B}$; otherwise for some integer $h\in \tilde \OO(1/\delta)$, it holds that $\frac{\delta^{2}}{s_B}(1+\delta)^{h} < \ve B_j[i] \leq \frac{\delta^{2}}{s_B}(1+\delta)^{h+1}$, and we round it up to $\frac{\delta^{2}}{s_B}(1+\delta)^{h+1}$. For small-weight items, we keep their weight vectors unchanged. Let $\tilde{\ve B}_j (1 \le j \le n)$ be the rounded weight vector. 

\smallskip\noindent\textbf{2. Rounding down the profits.} 
For each item $j\in I$, consider its profit $p_j$. If $p_j \le \delta^2$, we keep it as it is; otherwise $p_j>\delta^2$, we round the profit down to the largest value of the form $\delta^2(1+\delta)^h$. For profits whose values are at least $\delta^2$, simple calculation shows there are at most $\tilde\OO(1/\delta)$ distinct rounded profits. Denote by $\tilde{p}_j(1\le j \le n)$ the rounded profit, we have $\tilde{p}_j\le p_j \le (1+\delta)\tilde{p}_j$. Let 
$\tilde{\vep}= (\tilde{p}_1,\tilde{p}_2,\cdots,\tilde{p}_n)$.

\smallskip\noindent\textbf{3. Generating $I$ with a well-structured instance.} 
Let $\tilde{I}_\delta=\{1,2,\cdots,n\}$ denote the \textit{rounded} instance where each item $j\in \tilde{I}_{\delta}$ is associated with a \textit{rounded} weight vector $\tilde{\ve B}_j$, a \textit{rounded} profit $\tilde{p}_j$ and a cost vector $\ve A_j$. 

We replace $I$ with $\tilde{I}_\delta$ in \textbf{IPC}$_{\tau}(I,\ve1,\ve1)$. By doing so we obtain:
\begin{subequations}
	\begin{align}
	\textbf{IPC}_{\tau}(\tilde{I}_{\delta},\ve1,\ve1):\nonumber
	\min\limits_{\vex} \hspace{2mm} & \tilde{\vep}\vey&\nonumber\\
	s.t. \hspace{1mm}&\ve A\vex\le \ve1& \nonumber\\
	&\hspace{1mm} \vex \in\{0,1\}^n&\nonumber\\
	& \text{where }  \vey \text{ solves the following:}&\nonumber\\
	&\max\limits_{\vey}\quad \tilde{\vep}\vey &\nonumber\\
	& s.t. \quad\hspace{1mm} \sum_{j=1}^n \tilde{\ve B}_j[1]y_j\le 1&\label{IPC_aug:b}\\
	&\hspace{10mm} \sum_{j=1}^n \tilde{\ve B}_j[i]y_j\le 1+\tau, \hspace{2mm} \forall 2\le i\le s_B&\label{IPC_aug:a}\\
	&\hspace{10mm} \vex+\vey\le \ve1\nonumber\\
	&\hspace{10mm} \vey\in\{0,1\}^n&\nonumber
	\end{align}
\end{subequations}

Since we do not change the cost vectors of items, any feasible solution to \textbf{IPC}$_{\tau}(\tilde{I}_{\delta},\ve1,\ve1)$ is also a feasible solution to \textbf{IPC}$(I,\ve1,\ve1)$, and vice versa. Furthermore, we have the following Lemma~\ref{lemma:r_theta} which ensures that solving \textbf{IPC}$_{\tau}(\tilde{I}_\delta,\ve1,\ve1)$ gives a good approximation solution to \textbf{IPC}$(I,\ve1,\ve1)$.

\newtheorem*{l2}{Lemma \ref{lemma:r_theta}}
\begin{l2} 
Let $0<\tau\le 1/2$. Let $\tilde{\vex}$ be any feasible solution to \textbf{IPC}$_{\tau}(\tilde{I}_\delta,\ve1,\ve1)$. Then $\tilde{\vex}$ is a feasible solution to \textbf{IPC}$(I,\ve1,\ve1)$. Let $\widetilde{Obj}_{\tau}(\tilde{\vex})$ and $Obj(\tilde{\vex})$ be the objective values of \textbf{IPC}$_{\tau}(\tilde{I}_\delta,\ve1,\ve1)$ and \textbf{IPC}$(I,\ve1,\ve1)$ for $\vex = \tilde{\vex}$, respectively. If $\tau\ge 2\delta$, we have
$$ 
Obj(\tilde{\vex}) \le (1+\delta)\widetilde{Obj}_{\tau}(\tilde{\vex}) \le s_{B} (1+\delta)Obj(\tilde{\vex}).
$$
Furthermore, let $\widetilde{OPT}_{\tau}$ and $OPT$ be the optimal objective values of \textbf{IPC}$_{\tau}(\tilde{I}_\delta,\ve1,\ve1)$ and \textbf{IPC}$(I,\ve1,\ve1)$, respectively. We have 
$$ OPT \le (1+\delta)\widetilde{OPT}_{\tau} \le s_B (1+\delta)OPT.$$ 
\end{l2}
Towards the proof of Lemma~\ref{lemma:r_theta},
we have the following observation.
\begin{observation}\label{obs:decom_wr}
Given $\vey \in [0,1]^n$ satisfying $\sum_{j=1}^n \ve B_j y_j \le \ve1$, we have 
\begin{itemize}
    \item $\sum_{j=1}^n \tilde{\ve B}_j[1] y_j \le 1$;
    \item $\sum_{j=1}^n \tilde{\ve B}_j[i] y_j \le 1+2\delta, \hspace{2mm}i =2,3,\cdots,s_B$.
\end{itemize}
\end{observation}
\begin{proof}
Given that $\vey \in [0,1]^n$ satisfies $\sum_{j=1}^n \ve B_j y_j \le \ve 1$, then there are at most $s_B / \delta$ large items in $I$ picked in $\vey$. According to the rounding above, we have
\begin{itemize}
    \item[] $\sum_{j=1}^n \tilde{\ve B}_j[1] y_j=\sum_{j=1}^n {\ve B}_j[1] y_j \leq 1$;
    \item[] $\sum_{j=1}^n \tilde{\ve B}_j[i] y_j \leq (1+\delta)\sum_{j=1}^n {\ve B}_j[i] y_j+\frac{\delta^{2}}{s_B}\cdot \frac{s_B}{\delta} \leq 1+2\delta, \hspace{2mm} i =2,3,\cdots,s_B$.
\end{itemize}\end{proof}

\begin{proof}[Proof of Lemma~\ref{lemma:r_theta}]
Let $\tilde{\vex}$ be any feasible solution to \textbf{IPC}$_{\tau}(\tilde{I}_{\delta},\ve1,\ve1)$. Denote by $\widetilde{Obj}_{\tau}(\tilde{\vex})$ and $Obj(\tilde{\vex})$ the objective values of \textbf{IPC}$_{\tau}(\tilde{I}_{\delta},\ve1,\ve1)$ and \textbf{IPC}$(I,\ve1,\ve1)$ for $\vex = \tilde{\vex}$, respectively. Note that $\widetilde{Obj}_{\tau}(\tilde{\vex})$ is exactly the optimal objective value of the following integer program:
		\begin{align*}
		\widetilde{\textbf{IP}}_{\tau}(\tilde{\vex}): \quad 	\max\limits_{\vey}\quad  & \tilde{\vep}\vey&\\
		& \sum_{j=1}^n \tilde{\ve B}_j[1]y_j\le 1&  \\
		&  \sum_{j=1}^n \tilde{\ve B}_j[i]y_j\le 1+\tau, \hspace{2mm} \forall 2\le i\le s_B&  \\
		&   \vey \le \ve1-\tilde{\vex} &  \\
		& \vey\in \{0,1\}^n & 
		\end{align*}
	
Meanwhile, $Obj(\tilde{\vex})$ is the optimal objective value of the following integer program:
		\begin{align*}
		\textbf{IP}(\tilde{\vex}): \quad 	\max\limits_{\vey}\quad  & \vep\vey&\\
		& \sum_{j=1}^{n}\ve B_j[i]y_j\le 1, \hspace{2mm} \forall 1\le i\le s_B&  \\
		&   \vey \le \ve1-\tilde{\vex} &  \\
		& \vey\in \{0,1\}^n & 
		\end{align*}
	
We claim that $Obj(\tilde{\vex}) \le (1+\delta)\widetilde{Obj}_{\tau}(\tilde{\vex})$. Towards this, let $\vey^*(\tilde{\vex})$ be an optimal solution of $ \textbf{IP}(\tilde{\vex})$. By the definition of $\tilde{\vep}$, we have $\vep\vey^*(\tilde{\vex})\le(1+\delta)\tilde{\vep}\vey^*(\tilde{\vex})$. It suffices to show that $\vey^*(\tilde{\vex})$ is a feasible solution to $\widetilde{\textbf{IP}}_{\tau}(\tilde{\vex})$, which is guaranteed by  Observation~\ref{obs:decom_wr} and the fact that $\tau \ge 2\delta$.

Let $OPT$ and $\widetilde{OPT}_{\tau}$ be the optimal objective values of $\textbf{IPC}(I,\ve1,\ve1)$ and ${\textbf{IPC}}_{\tau}(\tilde{I}_{\delta},\ve1,\ve1)$, respectively, then we have $OPT \le (1+\delta)\widetilde{OPT}_{\tau}$. This is because that the optimal solution to ${\textbf{IPC}}_{\tau}(\tilde{I}_{\delta},\ve1,\ve1)$ is also a feasible solution to $\textbf{IPC}(I,\ve1,\ve1)$ while the optimal solution to $\textbf{IPC}(I,\ve1,\ve1)$ may achieve an even smaller value. 

It remains to prove that $\widetilde{Obj}_{\tau}(\tilde{\vex}) \le s_B {Obj}(\tilde{\vex})$ and $\widetilde{OPT}_{\tau} \le s_B OPT$. Towards this, we introduce a new bilevel program ${\textbf{IPC}}(\tilde{I}_{\delta},\ve1,\ve1)$, which is constructed by replacing (\ref{IPC_aug:a}) in \textbf{IPC}$_{\tau}(\tilde{I}_{\delta},\ve1,\ve1)$ with $ \sum_{j=1}^n \tilde{\ve B}_j[i]y_j\le 1( \forall 2\le i\le s_B$). That is, we reduce the budget of the follower to $\ve 1$. Let $\widetilde{OPT}$ be the optimal objective value of ${\textbf{IPC}}(\tilde{I}_{\delta},\ve1,\ve1)$ and let $\widetilde{Obj}(\tilde{\vex})$ be the objective value of ${\textbf{IPC}}(\tilde{I}_{\delta},\ve1,\ve1)$ for $\vex = \tilde{\vex}$.

Compare ${\textbf{IPC}}(\tilde{I}_{\delta},\ve1,\ve1)$ with ${\textbf{IPC}}({I},\ve1,\ve1)$. We see that in this two bilevel programs, the budgets of the leader and follower are the same, while the weight vectors in $\tilde{I}_{\delta}$ is always larger and the profits in $\tilde{I}_{\delta}$ is always smaller. Hence, for the same leader's solution, the follower in ${\textbf{IPC}}(\tilde{I}_{\delta},\ve1,\ve1)$ achieves a smaller objective, thus $\widetilde{Obj}(\tilde{\vex})\le Obj(\tilde{\vex})$. Furthermore, we have $\widetilde{OPT}\le OPT$. This is that because that the optimal solution to ${\textbf{IPC}}(I,\ve1,\ve1)$ is also a feasible solution to $\textbf{IPC}(\tilde{I}_{\delta},\ve1,\ve1)$ while the optimal solution to $\textbf{IPC}(\tilde{I}_{\delta},\ve1,\ve1)$ may achieve an even smaller value. 

Compare ${\textbf{IPC}}(\tilde{I}_{\delta},\ve1,\ve1)$ with \textbf{IPC}$_{\tau}(\tilde{I}_{\delta},\ve1,\ve1)$. According to Lemma~\ref{lemma:augment_budget}, we have $\widetilde{Obj}_{\tau}(\tilde{\vex}) \le s_{B}\widetilde{Obj}(\tilde{\vex})$ and $\widetilde{OPT}_{\tau}\le s_B\widetilde{OPT}$.

To summarize, we know $Obj(\tilde{\vex})\le (1+\delta)\widetilde{Obj}_{\tau}(\tilde{\vex}) \le s_{B}(1+\delta)\widetilde{Obj}(\tilde{\vex}) \le  s_{B}(1+\delta) {Obj}(\tilde{\vex})$ and $OPT\le (1+\delta)\widetilde{OPT}_{\tau} \le s_B(1+\delta)\widetilde{OPT}\le s_B (1+\delta)OPT$. Lemma~\ref{lemma:r_theta} is proved. 
\end{proof}

In the rest, we focus on \textbf{IPC}$_{2\delta}(\tilde{I}_{\delta},\ve1,\ve1)$ which is obtained by setting $\tau = 2\delta$ in \textbf{IPC}$_{\tau}(\tilde{I}_{\delta},\ve1,\ve1)$. Assuming $OPT\le 1$, we aim to design a polynomial time algorithm which returns a feasible solution to \textbf{IPC}$_{2\delta}(\tilde{I}_{\delta},\ve1,\ve1)$ with an objective value of at most $s_B+\OO(\epsilon)$. According to Lemma~\ref{lemma:r_theta}, this solution is exactly the solution that Lemma~\ref{lemma:residue-main} requires, and Theorem~\ref{theorem:s_B} is proved.

\subsubsection{Item Classification of $\tilde{I}_{\delta}$}\label{subsec:fur_prepro}
Recall that every item $j \in \tilde{I}_{\delta}$ is associated with a profit $\tilde{p}_j$ and a weight vector $\tilde{\ve B}_j$.

\textbf{Classifying Profits:} We say an item $j\in \tilde{I}_{\delta}$ has a large profit if $\tilde{p}_j>\delta$; a medium profit if $\delta^2<\tilde{p}_j\le \delta$; and a small profit if $\tilde{p}_j\le \delta^2$.

\textbf{Classifying Weights:}
We say an item $j \in \tilde{I}_{\delta}$ has a large weight if $ \|\tilde{\ve B}_j\|_{\infty} > \delta$; and a small weight otherwise. 

We say an item in $\tilde{I}_{\delta}$ is {\em large}(for the follower) if it has a large-profit or a large-weight. Otherwise, the item is {\em small}. To find a desired approximation solution to $\textbf{IPC}_{2\delta}(\tilde{I}_{\delta},\ve1,\ve1)$ in polynomial time 
, we will handle large items and small items separately. For simplicity, from now on, we use $\tilde{I}$ to denote $\tilde{I}_{\delta}$.

\subsection{Handling Large Items}\label{subsec:large_items_2}
\subsubsection{Determining the leader’s choice on large items}\label{subsec:large_items_l}
Instead of considering $\textbf{IPC}_{2\delta}(\tilde{I},\ve1,\ve1)$ directly, we introduce a bilevel program \textbf{IPC}$_{\overline{2\delta}}(\tilde{I},\ve1,\ve1)$, which is obtained by replacing~\eqref{IPC_aug:a} in \textbf{IPC}$_{2\delta}(\tilde{I},\ve1,\ve1)$ with $ \sum_{j=1}^n \tilde{\ve B}_j[i]y_j \le 1+\overline{2\delta} ( \forall 2\le i\le s_B$), where $\overline{2\delta} = 2\delta+\epsilon(1+2\delta)$. It is easy to see that for fixed $\vex$, the objective value of \textbf{IPC}$_{2\delta}(\tilde{I},\ve1,\ve1)$ is smaller than that of \textbf{IPC}$_{\overline{2\delta}}(\tilde{I},\ve1,\ve1)$. Let $\widetilde{OPT}_{\overline{2\delta}}$ be the optimal objective value of \textbf{IPC}$_{\overline{2\delta}}(\tilde{I},\ve1,\ve1)$. Recall that $OPT \le 1$, according to Lemma~\ref{lemma:r_theta}, we have $\widetilde{OPT}_{\overline{2\delta}}\le s_B OPT\le s_B$. 

Denote by ${S}^*$ the items selected by the leader in an optimal solution of \textbf{IPC}$_{\overline{2\delta}}(\tilde{I},\ve1,\ve1)$. The goal of this section is to guess \textit{large items} in ${S}^*$ in polynomial time.

\smallskip
\noindent\textbf{Large-profit small-weight items.} Notice that if there are at least $s_B/\delta$ such items in $\tilde{I}\setminus S^*$ for the follower to select, then selecting any $s_B/\delta$ of them gives a solution to \textbf{IPC}$_{\overline{2\delta}}(\tilde{I},\ve1,\ve1)$ with an objective value strictly larger than $s_B$, violating the fact that ${\widetilde{OPT}}_{\overline{2\delta}}\le s_B$.
Thus $S^*$ must include all except at most $s_B/\delta-1$ such items, which can be guessed out via $n^{\OO({s_B/\delta})}$ enumerations. Hence, we have the following observation. 

\begin{observation}\label{obs:lpsw_2}
    With $n^{\OO({s_B/\delta})}$ enumerations,  we can guess out all large-profit small-weight items in ${S}^*$. 
\end{observation}

\smallskip
\noindent\textbf{Large-weight small-profit items.} Notice that the follower can select at most $\frac{s_B(1+\overline{2\delta})}{\delta}$ items from this subgroup and their total profit is at most $\delta^2\cdot\frac{s_B(1+\overline{2\delta})}{\delta} = s_B(1+\overline{2\delta})\delta$. Hence, even if the leader does not select any such item, the objective value can increase by at most $\OO(s_B\delta)$, which leads to the following observation.

\begin{observation}\label{obs:splw_2}
    With an additive error of $\OO(s_B\delta)$, we may assume that ${S}^*$ does not contain large-weight small-profit items. 
\end{observation}

\smallskip
\noindent\textbf{Large-weight large/medium-profit items.} Notice that the follower can select at most $\frac{s_B(1+\overline{2\delta})}{\delta}$ items from this group. Recall the weight vector $\tilde{\ve B}_j$ of a large weight item. For every $2\le i\le s_B$, $\tilde{\ve B}_j[i]$ may take at most $\tilde{\OO}(1/\delta)$ distinct values. Furthermore, large/medium profits can also take $\tilde{\OO}(1/\delta)$ distinct values. Hence, we can divide large/medium-profit large-weight items into $\tilde\OO(1/\delta^{s_B})$ subgroups, such that items in the same subgroup have the same profit and the same rounded weight vector (except dimension 1), i.e., if items $j_1$ and $j_2$ belong to the same subgroup, then $\tilde{p}_{j_1}=\tilde{p}_{j_2}$, and $(\tilde{\ve B}_{j_1}[2], \cdots,\tilde{\ve B}_{j_1}[s_B])=(\tilde{\ve B}_{j_2}[2], \cdots,\tilde{\ve B}_{j_2}[s_B])$. Let $\{S_h: h\in \tilde\OO(1/\delta^{s_B})\}$ be the set of these subgroups. If there are two items in $S_h$ which are not selected by the leader, then the follower always prefers the item whose weight vector has a smaller value in the first dimension. Hence we have the following lemma. 

\begin{lemma}\label{lemma:lmplw2_2}
    With $n^{\tilde\OO(s_B/\delta^{s_B+1})}$ enumerations, we can guess out all large-weight large/medium-profit items in ${S}^*$. 
\end{lemma}
\begin{proof}
We can apply the same argument as the proof of Lemma~\ref{lemma:lmplw2}. That is, for each subgroup $S_{h}$, $ S^*\cap S_{h}$ can be determined through guessing out the following $ \frac{s_B(1+\overline{2\delta})}{\delta} $ key items in $S_{h}$: among items in $S_h\setminus S^*$, which is the item whose weight vector has the $k$-th smallest value in the first dimension for $k = 1,2,\cdots,\frac{s_B(1+\overline{2\delta})}{\delta}$? 
Thus via $n^{\OO(s_B/\delta)}$ enumerations, we can guess out all large-weight large/medium-profit items in ${S}^* \cap S_h$. Moreover, via total $n^{\tilde\OO(s_B/\delta^{s_B+1})}$ enumerations, we can guess out all large-weight large/medium-profit items in ${S}^*$. \end{proof}

To summarize, our above analysis leads to the following lemma:
\newtheorem*{l_gl}{Lemma \ref{lemma:guess_l_ge}}
\begin{l_gl} 
    With $\OO(s_B\delta)$ additive error, we can guess out all items in ${S}^*$ that have a large weight or a large profit by $n^{\tilde\OO(s_B/\delta^{s_B+1})}$ enumerations. 
\end{l_gl}

Let $\hat{I}$ 
be the set of small items, i.e., items of medium/small-profit and small-weight. Then $\tilde{I}\setminus \hat{I}$ is the set of large items. Denote by $\vex^*$ the optimal solution to \textbf{IPC}$_{\overline{2\delta}}(\tilde{I},\ve1,\ve1)$, which is corresponding to $S^*$. Then $\vex^*$ is a feasible solution to  \textbf{IPC}$_{2\delta}(\tilde{I},\ve1,\ve1)$ with an objective value of at most $s_B$. This is because that $\widetilde{OPT}_{\overline{2\delta}}\le s_B$,  and the follower in \textbf{IPC}$_{\overline{2\delta}}(\tilde{I},\ve1,\ve1)$ is stronger. In the following, we assume that large items in $S^*$ have been correctly guessed. Hence, the values in $\{x_j^*: \forall j\in \tilde{I}\setminus \hat{I}\}$ are known. We let $\vea'$ be the sum of the cost vectors of these guessed-out large items.

\subsubsection{Finding the follower’s dominant choices on large items}\label{subsection:connect_m}

Consider all the large items. Even if the leader's choice on large items is fixed, the follower may still have exponentially many different choices on the remaining large items. The goal of this subsection is to show that, among these choices of the follower, it suffices to restrict our attention to a few "dominant" choices that always outperform other choices.    

For simplicity, we re-index items such that $\hat{I}= \{1,2,\cdots,\hat{n}\}$, then $\tilde{I} \setminus\hat{I} = \{\hat{n}+1,\cdots, n\}$, where $\hat{n}\le n$. 
Let $\veb_{2\delta}= (1,1+2\delta,\cdots, 1+2\delta)$, which is the follower's budget vector in \textbf{IPC}$_{2\delta}(\tilde{I},\ve1,\ve1)$. 
Denote by $I'$ the set of items in $\tilde{I}\setminus \hat{I}$ which are \textit{not} selected by the leader. Note that we have correctly guessed out large items in $S^*$ and the objective value of \textbf{IPC}$_{2\delta}(\tilde{I},\ve1,\ve1)$ for $\vex = \vex^*$ is at most $s_B$, thus the follower cannot select items from $I'$ with total profit larger than $s_B$. Since $I'$ is a subset of large-profit or large-weight items, the follower can select at most $\OO({s_B}/{\delta})$ items from $I'$. 

Let $V$ denote the largest integer satisfying $\epsilon(1+\epsilon)^{V}\le 1+2\delta$. For each integer $ k \in [1,1+s_B/{\epsilon}] $ and each $\ve v := (v_2,v_3,\cdots,v_{s_B})$ where $ v_i \in \{-1,0,1,2,\cdots,V\} $, we define the following sub-problem: Among all follower's choices on $I^{\prime}$, find out one whose summation of weight vectors has the smallest value in first dimension, and this choice must satisfy: (i) the summation of profits is within $[(k-1)\epsilon,k\epsilon)$; (ii) in each dimension $i\in [2,s_B]$, the \textit{remaining} follower's budget is within $[l_{v_i},r_{v_i})$, where $[l_{v_i}, r_{v_i}) = [0,\epsilon)$ if $v_i = -1 $, otherwise $[l_{v_i}, r_{v_i}] = [\epsilon(1+\epsilon)^{v_{i}},\epsilon(1+\epsilon)^{v_{i}+1})$. The sub-problem is formulated as follows:
    	\begin{align*}
    	\textbf{SP}(k;\vev):
    	\min\limits_{y_j:j\in I'} \hspace{2mm} & \sum_{j\in I'} \tilde{\ve B}_j[1] y_j&\\
    	s.t.\hspace{2mm} & \sum_{j\in I'}\tilde{p}_jy_j\in [(k-1)\epsilon,k\epsilon) &\\
    	& \sum_{j\in I'}\tilde{\ve B}_j[1]y_j \le 1& \\
    	& 1+2\delta -\sum_{j\in I'} \tilde{\ve B}_j[i] y_j \in [l_{v_i}, r_{v_i}), \indent i = 2,3,\cdots, s_B& \\
    	&y_j \in \{0,1\}, j\in I'&
    	\end{align*}
Note that there are at most $\tilde\OO({s_B}/{\epsilon^{s_B}})$ distinct sub-problems. Denote by $\textbf{opt}(k;\ve v)$ an optimal solution of $\textbf{SP}(k;\vev)$, if it exists. Let $\Theta = \{\textbf{opt}(k;\vev)|k,v_i \in \mathbb{Z}; k \in [1, \frac{\epsilon+s_B}{\epsilon}]; v_i \in [-1, V] \text{ for  } i= 2,3,\cdots,s_B\}$, which contains the follower's $\tilde\OO({s_B}/{\epsilon^{s_B}})$ possible choices on $I'$. Recall that the follower can select at most $\OO({s_B}/{\delta})$ items from $I'$. Thus via $n^{\OO({s_B}/{\delta})}$ enumerations, for each sub-problem $\textbf{SP}(k;\vev)$, we can return its optimal solution or assert there does not exist a feasible solution for this sub-problem, which means that $\Theta$ can be determined in polynomial time. We will show that to select small items to accompany the large items in $S^*$, it suffices to consider the choices in $\Theta$. 

For the $l$-th choice in $\Theta$, let $\veb_l \le \veb_{2\delta}$ and $P_l$ be the summation of weight vectors and the summation of profits, respectively. Then the follower has a residual budget of $\veb_{2\delta}-\veb_l$. Assuming that the $l$-th choice in $\Theta$ is $\textbf{opt}(k;\vev)$,
define $\ved_l  = (1-\veb_l[1],r_{v_2}, \cdots, r_{v_{s_B}})$, it follows that $\veb_{2\delta}-\veb_l \le \ved_l$, moreover, we have the following observation.

\begin{observation}\label{obs:bl+dl}
      Given the $l$-th choice in $\Theta$, say $\textbf{opt}(k,\vev)$, then any feasible solution $\{y_j:j\in I'\}$ of \textbf{SP}$(k,\vev)$ satisfies 
      \begin{itemize}
          \item $\sum_{j\in I'}\tilde{\ve B}_j[1] y_j + \ved_{l}[1] \le 1$;
          \item $\sum_{j\in I'}\tilde{\ve B}_j[i] y_j + \ved_{l}[i] \le 1+\overline{2\delta}, \hspace{2mm} i=2,3,\cdots,s_B$,
      \end{itemize}
      where $\overline{2\delta} = 2\delta+\epsilon(1+2\delta)$.
\end{observation}

Recall that the leader has a residual budget of $\ve1 - \vea'$ for items in $\hat{I}$, and by guessing we already know the value of $x^*_j$ for $j\in \tilde{I}\setminus \hat{I}$. Define $\vey[\hat{I}]= (y_1,y_2,\cdots,y_{\hat{n}})$. Consider the following bilevel program:
\begin{subequations}
	\begin{align}
	\textbf{MBi-IP}(\tilde{I},\ve1,\ve1):
	\min_{\vex }\quad & P_{l}+ \sum_{j\in \hat{I}}\tilde{p}_{j}y_{j}& \nonumber \\
	s.t. \hspace{2mm}& \sum_{j\in \hat{I}}\ve A_{j}x_{j}\le \ve1-\vea' &\label{MBi-IP:a1}\\
	&x_j = x_j^*, \hspace{7mm} \forall j\in \tilde{I} \setminus \hat{I}& \label{MBi-IP:b1}\\
	& x_j \in \{0,1\}, \hspace{2mm} \forall j\in \hat{I} & \label{MBi-IP:c1}\\
	& \text{where integer } l, \vey[\hat{I} ] \text{ solves the following:}& \nonumber\\
	&\max_{ 1\le l \le |\Theta| } \max_{\vey[\hat{I} ]} \quad P_{l}+ \sum_{j\in \hat{I}}\tilde{p}_{j}y_{j} & \label{MBi-IP:d}\\
	&{ }\hspace{10mm} s.t. \hspace{2mm}  \sum_{j\in \hat{I}} \tilde{\ve B}_{j}y_{j}\le \ved_l &\label{MBi-IP:a2}\\
	&\hspace{17mm} y_{j} \le 1-x_j, \hspace{2mm} \forall j\in \hat{I} \label{MBi-IP:b2} \\
	&\hspace{17mm}  y_j\in \{0,1\}, \hspace{2mm} \forall j\in \hat{I} \label{MBi-IP:c2}&
	\end{align}
\end{subequations}

We have the following lemma.
\begin{lemma}\label{lemma:MBi_1}
 Assume that $OPT \le 1$. Let $\vex'$ be any feasible solution to $\textbf{MBi-IP}(\tilde{I},\ve1,\ve1)$, then $\vex'$ is also a feasible solution to \textbf{IPC}$_{2\delta}(\tilde{I},\ve1,\ve1)$.
    Let $Obj^{MBi}(\vex')$ and $\widetilde{Obj}_{2\delta}(\vex')$ be the objective values of $\textbf{MBi-IP}(\tilde{I},\ve1,\ve1)$ and \textbf{IPC}$_{2\delta}(\tilde{I},\ve1,\ve1)$ for $\vex = \vex'$, respectively. We have 
    $$
    \widetilde{Obj}_{2\delta}(\vex') \le Obj^{MBi}(\vex')+\epsilon.
    $$
    Furthermore, let $OPT^{Mbi}$ be the optimal objective value of $\textbf{MBi-IP}(\tilde{I},\ve1,\ve1)$. We have 
    $$
    OPT^{MBi} \le s_B.
    $$
\end{lemma}
\begin{proof}

Let $\vex'$ be any feasible solution to $\textbf{MBi-IP}(\tilde{I},\ve1,\ve1)$, denote by $Obj^{MBi}(\vex')$ and $\widetilde{Obj}_{2\delta}(\vex')$ the objective values of $\textbf{MBi-IP}(\tilde{I},\ve1,\ve1)$ and \textbf{IPC}$_{2\delta}(\tilde{I},\ve1,\ve1)$ for $\vex = \vex'$, respectively. Note that $\widetilde{Obj}_{2\delta}(\vex')$ is the optimal objective of the following program:
		\begin{align*}
		\widetilde{\textbf{IP}}_{2\delta}(\vex'): \quad 	\max\limits_{\vey}\quad  & \tilde{\vep}\vey&\\
		& \sum_{j=1}^n \tilde{\ve B}_j[1]y_j\le 1&  \\
		&  \sum_{j=1}^n \tilde{\ve B}_j[i]y_j\le 1+2\delta, \hspace{2mm} \forall 2\le i\le s_B&  \\
		&   \vey \le \ve1-\vex' &  \\
		& \vey\in \{0,1\}^n & 
		\end{align*}

Notice that $x'_j = x_j^* (\forall j\in \tilde{I} \setminus \hat{I})$ and the objective value of \textbf{IPC}$_{2\delta}(\tilde{I},\ve1,\ve1)$ for $\vex = \vex^*$ is at most $s_B$. Thus in $\widetilde{\textbf{IP}}_{2\delta}(\vex')$, the maximum profit the follower could achieve from $\tilde{I} \setminus \hat{I}$ is at most $s_B$. Let $\vey'$ be an optimal solution to $\widetilde{\textbf{IP}}_{2\delta}(\vex')$, then there exists a choice $\textbf{opt}(k;\vev) \in \Theta$ such that 
	\begin{itemize}
	    \item $\sum_{j \in {\tilde{I}\setminus \hat{I}}} \tilde{p}_jy'_j \in [(k-1)\epsilon, k\epsilon) $;
	    \item $1+2\delta - \sum_{j \in {\tilde{I}\setminus \hat{I}}} \tilde{\ve B}_j[i]y'_j \in [l_{v_i}, r_{v_i})$, $\quad i= 2,3,\cdots, s_B.$
	\end{itemize}

Assuming that $\textbf{opt}(k;\vev)$ is the $l$-th choice in $\Theta$, then it follows that
\begin{itemize}
    \item $P_l \in [(k-1)\epsilon, k\epsilon) $;
    \item $\sum_{j\in \hat{I}} \tilde{\ve B}_j[1]y'_j \le 1 - \veb_{l}[1] = \ved_{l}[1]$;
    \item $\sum_{j\in \hat{I}} \tilde{\ve B}_j[i]y'_j \le 1+2\delta - \sum_{j \in {\tilde{I}\setminus \hat{I}}} \tilde{\ve B}_j[i]y'_j < r_{v_{i}}= \ved_{l}[i]$, $\quad i= 2,3,\cdots, s_B.$
\end{itemize}
Thus $\sum_{j \in {\tilde{I}\setminus \hat{I}}} \tilde{p}_jy'_j \le P_l+\epsilon$, and $\vey'[\hat{I}]=(y'_1,y'_2,\cdots, y'_{\hat{n}})$ is a feasible solution of the following program:
    		\begin{align*}
    		{\textbf{IP}}(l;\vex'): \quad\max\limits_{\vey[\hat{I}]}\quad  & \sum_{j\in \hat{I}}\tilde{p}_{j}y_{j}&\\
    		&  \sum_{j\in \hat{I}} \tilde{\ve B}_jy_j\le \ved_l&  \\
    		&   y_j \le 1-x'_j , \forall j \in \hat{I} & \\
    		& y_j\in \{0,1\} , \forall j \in \hat{I}& 
    		\end{align*}
This means that $l$ and $\vey'[\hat{I}]$ form a feasible solution to the program \eqref{MBi-IP:d} - \eqref{MBi-IP:c2} when $\vex = \vex'$. Then we have 
$$\tilde{\vep}\vey' = \sum_{j\in {\tilde{I}\setminus}\hat{I}}\tilde{p}_jy'_j+\sum_{j\in \hat{I}}\tilde{p}_jy'_j \le P_l+\epsilon+\sum_{j\in \hat{I}}\tilde{p}_jy'_j\le Obj^{MBi}(\vex')+\epsilon.$$
Thus the first part of Lemma~\ref{lemma:MBi_1} is proved. 

It remains to prove that the optimal objective value of \textbf{MBi-IP}$(\tilde{I},\ve1,\ve1)$ is at most $s_B$. Towards this, we compare $\textbf{MBi-IP}(\tilde{I},\ve1,\ve1)$ with   \textbf{IPC}$_{\overline{2\delta}}(\tilde{I},\ve1,\ve1)$. Recall that $\vex^*$ and $\widetilde{OPT}_{\overline{2\delta}} $ are the optimal solution and the optimal objective value of \textbf{IPC}$_{\overline{2\delta}}(\tilde{I},\ve1,\ve1)$, respectively. Notice that $\widetilde{OPT}_{\overline{2\delta}} \le s_B$ and $\vex^*$ is feasible for \textbf{MBi-IP}$(\tilde{I},\ve1,\ve1)$, it suffices to prove that the objective value of \textbf{MBi-IP}$(\tilde{I},\ve1,\ve1)$ for $\vex = \vex^*$ is no larger than $\widetilde{OPT}_{\overline{2\delta}}$.

Note that $\widetilde{OPT}_{\overline{2\delta}}$ is the optimal objective value of the following program:
		\begin{align*}
		\widetilde{\textbf{IP}}_{\overline{2\delta}}(\vex^{*}): \quad 	\max\limits_{\vey}\quad  & \tilde{\vep}\vey&\\
		& \sum_{j=1}^n \tilde{\ve B}_j[1]y_j\le 1&  \\
		&  \sum_{j=1}^n \tilde{\ve B}_j[i]y_j\le 1+\overline{2\delta}, \hspace{2mm} \forall 2\le i\le s_B&  \\
		&   \vey \le \ve1-\vex^{*} &  \\
		& \vey\in \{0,1\}^n & 
		\end{align*}

Given $\vex^*$, there exist $l' \in \{ 1,2,\cdots, |\Theta|\}$ and $\vey^{l'}[\hat{I}] :=(y^{l'}_1,y^{l'}_2,\cdots,y^{l'}_{\hat{n}})$ which form an optimal solution to the program \eqref{MBi-IP:d} - \eqref{MBi-IP:c2} when $\vex = \vex^*$. Note that the objective value of \textbf{MBi-IP}$(\tilde{I},\ve1,\ve1)$ for $\vex =  \vex^*$ is 
$P_{l'}+\sum_{j\in \hat{I}}\tilde{p}_j y^{l'}_j$. Define $\vey\in \{0,1\}^n$ as follows: $y_j = y^{l'}_j$ if $j\in \hat{I}$; $y_j = 1$ if item $j$ is selected in the $l'$-th choice in $\Theta$; $y_j = 0$ otherwise. Then we have
$$\sum\limits_{j\in \tilde{I}\setminus \hat{I}} \tilde{\ve B}_j[i]y_j +\sum\limits_{j\in \hat{I}} \tilde{\ve B}_j[i]y_j \le \veb_{l'}[i]+\ved_{l'}[i],\quad i = 1,2,\cdots,s_B.$$
It follows that $\vey$ is a feasible solution to $\widetilde{\textbf{IP}}_{\overline{2\delta}}(\vex^{*})$ (by Observation~\ref{obs:bl+dl}). Thus $P_{l'} + \sum_{j\in \hat{I}}\tilde{p}_j y^{l'}_j$ is no larger than $\widetilde{OPT}_{\overline{2\delta}}$, and the second part of Lemma~\ref{lemma:MBi_1} is proved.\end{proof}

In the following, we focus on $\textbf{MBi-IP}(\tilde{I},\ve1,\ve1)$. Assuming $OPT\le1$, our goal is to design a polynomial time algorithm which returns a feasible solution to $\textbf{MBi-IP}(\tilde{I},\ve1,\ve1)$ with an objective value of at most $s_B+\OO(\epsilon)$. According to Lemma~\ref{lemma:MBi_1} and Lemma~\ref{lemma:r_theta}, this solution is exactly the solution that Lemma~\ref{lemma:residue-main} requires.

\subsection{Handling Small Items} \label{subsec:small-items_2}
The goal of this section is to select \textit{small items} to accompany the large items of ${S}^*$ that have been guessed out in Section~\ref{subsec:large_items_l}. It suffices to solve $\textbf{MBi-IP}(\tilde{I},\ve1,\ve1)$. Towards this, we first obtain a linear relaxation of $\textbf{MBi-IP}(\tilde{I},\ve1,\ve1)$ where both the leader and the follower can select items fractionally. To solve the relaxation, we construct a single level LP and find its extreme point optimal fractional solution (see Section~\ref{sub:frac}). Finally we round this fractional solution to a desired approximation solution to $\textbf{MBi-IP}(\tilde{I},\ve1,\ve1)$ (see Section~\ref{sub:rounding_frac}).

\subsubsection{Finding a fractional solution}\label{sub:frac}
Replace \eqref{MBi-IP:c1} and \eqref{MBi-IP:c2} in \textbf{MBi-IP}$(\tilde{I},\ve1,\ve1)$ with $x_j \in [0,1]( \forall j\in \hat{I})$ and $y_j\in [0,1] (\forall j\in \hat{I})$, respectively, we obtain a relaxation of \textbf{MBi-IP}$(\tilde{I},\ve1,\ve1)$ as follows:
\begin{subequations}
	\begin{align}
	\textbf{MBi-IP}_r(\tilde{I},\ve1,\ve1):
	\min_{\vex }\quad & P_{l}+ \sum_{j\in \hat{I}}\tilde{p}_{j}y_{j}& \nonumber \\
	s.t. \hspace{2mm}& \sum_{j\in \hat{I}}\ve A_{j}x_{j}\le \ve1-\vea' &\label{MBi-IP_r:a1}\\
	&x_j = x_j^*, \hspace{7mm} \forall j\in \tilde{I} \setminus \hat{I}& \label{MBi-IP_r:b1}\\
	& x_j \in [0,1], \hspace{4mm} \forall j\in \hat{I} & \label{MBi-IP_r:c1}\\
	& \text{where integer } l, \vey[\hat{I} ] \text{ solves the following:}& \nonumber\\
	&\max_{ 1\le l\le |\Theta| } \max_{\vey[\hat{I} ]} \quad P_{l}+ \sum_{j\in \hat{I}}\tilde{p}_{j}y_{j} & \label{MBi-IP_r:d}\\
	&{ }\hspace{10mm} s.t. \hspace{2mm}  \sum_{j\in \hat{I}} \tilde{\ve B}_{j}y_{j}\le \ved_l &\label{MBi-IP_r:a2}\\
	&\hspace{17mm} y_{j} \le 1-x_j, \hspace{2mm} \forall j\in \hat{I} \label{MBi-IP_r:b2} \\
	&\hspace{17mm}  y_j\in [0,1], \hspace{4mm} \forall j\in \hat{I} \label{MBi-IP_r:c2}&
	\end{align}
\end{subequations}

We have the following lemma. 
\begin{lemma}\label{lemma:opt^MBi_r}
        Assume that $OPT \le 1$. Let $OPT^{MBi}_r$ and $OPT^{MBi}$ be the optimal objective values of $\textbf{MBi-IP}_r(\tilde{I},\ve1,\ve1)$ and  $\textbf{MBi-IP}(\tilde{I},\ve1,\ve1)$, respectively. Then we have $$OPT^{MBi}_r\le OPT^{MBi}+\OO(s_B \delta)\le s_B + \OO(s_B \delta).$$%
\end{lemma}
\begin{proof}
It is not straightforward if we compare $\textbf{MBi-IP}_r(\tilde{I},\ve1,\ve1)$ with $\textbf{MBi-IP}(\tilde{I},\ve1,\ve1)$ directly, as both the follower and the leader become stronger in the relaxation (in the sense they can pack items fractionally). Towards this, we introduce an intermediate bilevel program $\textbf{MBi-IP}_{in}(\tilde{I},\ve1,\ve1)$, which is obtained by replacing \eqref{MBi-IP_r:c1} in $\textbf{MBi-IP}_r(\tilde{I},\ve1,\ve1)$ with $x_j\in \{0,1\} (\forall j\in \hat{I})$, that is, we only allow the follower to select items fractionally but not the leader. Denote by $OPT^{MBi}_{in}$ the optimal objective value of $\textbf{MBi-IP}_{in}(\tilde{I},\ve1,\ve1)$.
    
First, we compare $\textbf{MBi-IP}_{in}(\tilde{I},\ve1,\ve1)$ with $\textbf{MBi-IP}_r(\tilde{I},\ve1,\ve1)$. We see that in $\textbf{MBi-IP}_r(\tilde{I},\ve1,\ve1)$ the follower is facing a stronger leader who is allowed to pack items fractionally, and it thus follows that $OPT^{MBi}_{r}\le OPT^{MBi}_{in}$.
    
Next, we compare $\textbf{MBi-IP}_{in}(\tilde{I},\ve1,\ve1)$ with $\textbf{MBi-IP}(\tilde{I},\ve1,\ve1)$. Note that the leader’s solution must be integral in both programs. 
Let $\vex^{MBi}$ be an optimal solution to $\textbf{MBi-IP}(\tilde{I},\ve1,\ve1)$, then $\vex^{MBi}$ is also a feasible solution to $\textbf{MBi-IP}_{in}(\tilde{I},\ve1,\ve1)$. Once the leader fixes his solution as $\vex^{MBi}$ in $\textbf{MBi-IP}_{in}(\tilde{I},\ve1,\ve1)$, there exist $l' \in \{1,2,\cdots, |\Theta|\}$ and $\vey'[\hat{I}]$ which form an optimal solution to the program \eqref{MBi-IP_r:d} - \eqref{MBi-IP_r:c2} when $\vex = \vex^{MBi}$. Then the objective value of $\textbf{MBi-IP}_{in}(\tilde{I},\ve1,\ve1)$ for $\vex = \vex^{MBi}$ is $Obj^{MBi}_{in} =  P_{l'}+\sum_{j\in \hat{I}} \tilde{p}_{j}y'_{j}.$ We have the following two observations:
    \begin{itemize}
        \item $OPT^{MBi}_{in} \le Obj^{MBi}_{in}$. This is because that 
        $\vex^{MBi}$ is just a feasible solution to $\textbf{MBi-IP}_{in}(\tilde{I},\ve1,\ve1)$, while the optimal solution of the leader may achieve an even smaller objective value.
        \item $Obj^{MBi}_{in} \le OPT^{MBi}+s_B \delta$.
    \end{itemize}
Towards the second observation, notice that $\sum_{j\in \hat{I}} \tilde{p}_{j}y'_{j}$ is the optimal objective value of the following linear program:
    		\begin{align*}
    		{\textbf{LP}}(l';\vex^{MBi}): \quad\max\limits_{\vey[\hat{I}]}\quad  & \sum_{j\in \hat{I}}\tilde{p}_{j}y_{j}&\\
    		&  \sum_{j\in \hat{I}} \tilde{\ve B}_jy_j\le \ved_{l'}&  \\
    		&   y_j \le 1-x^{MBi}_j, \hspace{1mm} \forall j \in \hat{I} & \\
    		& y_j\in [0,1], \hspace{6mm} \forall j \in \hat{I}& 
    		\end{align*}
Let $\vey'[\hat{I}]$ be an extreme point optimal fractional solution to {\textbf{LP}}$(l';\vex^{MBi})$, we define $\vey[\hat{I}] \in \{0,1\}^{\hat{n}}$ as follows: $y_j = 1$ if $y'_j = 1$; $y_j = 0$ otherwise. Note that $\vey[\hat{I}]$ is a feasible solution to the following integer program:
    		\begin{align*}
    		{\textbf{IP}}(l';\vex^{MBi}): \quad\max\limits_{\vey[\hat{I}]}\quad  & \sum_{j\in \hat{I}}\tilde{p}_{j}y_{j}&\\
    		&  \sum_{j\in \hat{I}} \tilde{\ve B}_jy_j\le \ved_{l'}&  \\
    		&   y_j \le 1-x^{MBi}_j, \hspace{1mm}\forall j \in \hat{I} & \\
    		& y_j\in \{0,1\},\hspace{4mm} \forall j \in \hat{I}& 
    		\end{align*}
Then $l'$ and $\vey[\hat{I}]$ form a feasible solution to the program \eqref{MBi-IP:d} - \eqref{MBi-IP:c2} when $\vex = \vex^{MBi}$, thus we have $P_{l'} + \sum_{j\in \hat{I}} \tilde{p}_{j}y_{j} \le OPT^{MBi}$. 
Notice that there are at most $s_B$ variables in $\vey'[\hat{I}]$ taking fractional values and $\tilde{p}_j \le \delta $ for $j\in \hat{I}$, it thus follows that $\sum_{j\in \hat{I}} \tilde{p}_{j}y'_{j} \le \sum_{j\in \hat{I}} \tilde{p}_{j}y_{j}+s_B\delta$. 
    
To summarize, we have $OPT^{MBi}_r \le OPT^{MBi}_{in} \le Obj^{MBi}_{in} \le OPT^{MBi}+s_B \delta \le s_B+s_B\delta$, where the last inequality is guaranteed by Lemma~\ref{lemma:MBi_1}.\end{proof}

The rest of this section is to solve $\textbf{MBi-IP}_r(\tilde{I},\ve1,\ve1)$. Since $\textbf{MBi-IP}_r(\tilde{I},\ve1,\ve1)$ is a bilevel linear program, which is hard to solve (Jeroslow~\cite{DBLP:journals/mp/Jeroslow85}), we consider to find an approximation solution to $\textbf{MBi-IP}_r(\tilde{I},\ve1,\ve1)$ with an objective value of at most $s_B+\OO(\epsilon)$. Towards this, in the following Section~\ref{sub:prepro_MBi_r} - Section~\ref{sub:single-level-LP}, we transform $\textbf{MBi-IP}_r(\tilde{I},\ve1,\ve1)$ to a standard (single level) LP whose extreme point optimal solution is what we want. Moreover, both the construction of the single level program and the algorithm to solve the single level program run in polynomial time.

\paragraph{Preprocessing of small items}\label{sub:prepro_MBi_r}
\smallskip
\noindent\textbf{Scaling.} Recall that $\{\tilde{\ve B}_j:j\in \hat{I}\}$ is the set of weight vectors of small items in $\tilde{I}$. Given $l \in \{1,2,\cdots,|\Theta|\}$, for each item $j\in \hat{I}$, we scale its weight vector $\tilde{\ve B}_j$ by $\ved_{l}$, and let $\ve B^{l}_j$ denote the scaled weight vector. That is, $\ve B^{l}_j[i] = \tilde{\ve B}_j[i] / \ved_{l}[i] (\forall 1\le i \le s_B)$.

It is easy to see that $\{y_j:\sum_{j\in \hat{I}} \tilde{\ve B}_j y_j\le \ved_l, \forall j\in \hat{I}\}$ is equal to $\{y_j:\sum_{j\in \hat{I}} {\ve B}^{l}_j y_j\le \ve1, \forall j\in \hat{I}\}$, thus we can replace \eqref{MBi-IP_r:a2} in \textbf{MBi-IP}$_r(\tilde{I},\ve1,\ve1)$ with $\sum_{j\in \hat{I}} {\ve B}^l_{j}y_{j}\le \ve 1$. Without loss of generality, we assume that $ \|\ve B^l_j\|_{\infty}\le 1$ for all $j\in \hat{I}$.\footnote{If $ \|\ve B^l_j\|_{\infty} > 1$, then item $j$ can not be selected by the follower, thus $y_j= 0$.}

In the following we will round the scaled weight vectors of small items, however, we do {\it not} really change the input instance. We still keep the original weight vectors of items but only use the rounded weight vectors to classify items. This should be differentiated from the rounding procedure we applied before to large items, where we replace the original instance $I$ with $\tilde{I}$. 

\smallskip
\noindent\textbf{Rounding weight vectors.}
We round the scaled weight vectors of small items as follows.

\textbf{1. Scaling:} For every item $j\in \hat{I}$, we define vector $\ve g^l_j = \ve B^l_j / w^l_j,$ where $w^l_j = \| \ve B^l_j\|_{\infty}$. Consequently, $\|\veg^l_j\|_{\infty}=1$. In other word, $\ve g^l_j$ is the vector obtained by scaling up the largest coordinate of $\ve B^l_j$ to 1, and then scale up other coordinates proportionally.

\textbf{2. Rounding up the scaled vectors:} For every item $j$, we round up each coordinate of $\ve g^l_j$ to the smallest value of the form $\epsilon(1+\epsilon)^{h}$. Precisely, for $1 \leq i \leq s_B$, if $\ve g^l_j[i] \le \epsilon$, we round it up to $\epsilon$; otherwise it holds that $\epsilon(1+\epsilon)^{h} < \ve g^l_j[i] \leq \epsilon(1+\epsilon)^{h+1}$ for some $h\in \tilde \OO(1/\epsilon)$, we round it up to $\epsilon(1+\epsilon)^{h+1}$. Denote by $\bar{\ve g}^l_j$ the rounded vector, we {call it the \textit{shape} of item $j$.} Since every coordinate of $\bar{\ve g}^l_j$ can take at most $\tilde\OO(1/ \epsilon)$ distinct values, thus there are at most $\tilde\OO(1/\epsilon^{s_B})$ distinct shapes.

\textbf{3. Generating rounded weight vectors:} For every item $j$, we round up its weight weight to $w^l_j\cdot \bar{\ve g}^l_j$ and denote by $\bar{\ve B}^l_j$ the rounded weight vector. It follows that $\ve B^l_j \le \bar{\ve B}^l_j$. 
Moreover, we have the following observation.
\begin{observation} \label{observation: obs_frac1}
    Given $(y_1,\cdots,y_{\hat{n}})\in [0,1]^{\hat{n}}$ satisfying $\sum_{j\in \hat{I}}\ve B^l_jy_j \le \ve1$, we have $\sum_{j\in \hat{I}}\bar{\ve B}^l_jy_j \le (1+(s_B+1)\epsilon) \cdot \ve1$.
\end{observation}
\begin{proof}
Given $(y_1,\cdots,y_{\hat{n}}) \in[0,1]^{\hat{n}}$ satisfying $\sum_{j\in \hat{I}}\ve B^l_jy_j \le \ve1$, the observation follows directly from the following two formulas:
\begin{itemize}
    \item[ ]$\sum_{j:\ve g^l_j[i] \le \epsilon}\bar{\ve B}^l_j[i]{y}_j = \epsilon\sum_{j:\ve g^l_j[i] \le \epsilon}w_j{y}_j \leq \epsilon \sum_{i=1}^{s_B}\sum_{j\in \hat{I}}\ve B^l_j[i]{y}_j \le \epsilon s_B, \hspace{2mm}1 \le i \le s_B;$
    \item[ ]$\sum_{j:\ve g^l_j[i] > \epsilon}\bar{\ve B}^l_j[i]{y}_j\le (1+\epsilon)
    \sum_{j:\ve g^l_j[i] > \epsilon}\ve B^l_j[i]{y}_j\le 1+\epsilon ,\hspace{2mm}1 \le i \le s_B.$
\end{itemize} \end{proof}

\smallskip
\noindent\textbf{Item classification of $\hat{I}$.} We classify small items in the following way.

\textbf{1. Classifying Shapes:} We first classify items according to their \textit{shapes}. Precisely, we define the index sets $I^l_h$'s such that $j \in {I}^l_h$ if and only if $\bar{\ve g}^l_j$ is the $h$-th rounded vector. 

\textbf{2. Classifying Ratios:} Define the \textit{ratio} of each item $j$ as $\rho^l_j = \tilde{p}_j / w^l_j$, where $w^l_j = \| \ve B^l_j\|_{\infty}$. Then in each $I^l_h$, we first divide items into three mega sub-groups:  let $I^{(l,s)}_h$ be the set of items whose ratio is smaller or equal to $\epsilon$; let $I^{(l,b)}_{h}$ be the set of items whose ratio is larger than $1/\epsilon;$ let $I^{(l,m)}_h$ be the set of remaining items. Let $I^{(l,s)} = \mathop{\cup}\limits_{h}I^{(l,s)}_h$, $I^{(l,b)} =\mathop{\cup}\limits_{h}I^{(l,b)}_h$ and $I^{(l,m)} = \mathop{\cup}\limits_{h}I^{(l,m)}_h$.

\textbf{3. Sub-Classifying Medium Ratios:} In each $I^{(l,m)}_h$, we round down the ratio of the each item to the nearest value of the form $\epsilon(1+\epsilon)^k$. More precisely, for item $j \in I^{(l,m)}_h$, it holds that $\epsilon(1+\epsilon)^k \le \rho^l_j < \epsilon(1+\epsilon)^{k+1}$ for some $k\in \tilde\OO(1/\epsilon)$, we round down the ratio $\rho^l_j$ to $\epsilon(1+\epsilon)^k$. 

After rounding, in each $I^{(l,m)}_h$, we divide items into subgroups, such that in each subgroup, items have the same shape and the same rounded ratio. These subgroups are denoted by $T^{(l,h)}_{1}, T^{(l,h)}_2, \dots, T^{(l,h)}_{\gamma_{(l,h)}}$. 
Let ${\veg}_k^{(l,h)}$ and $\rho^{(l,h)}_k$ denote the shape and the rounded ratio of items in the same sub-group $T^{(l,h)}_{k}$, respectively. It holds that $\rho^l_j \le \rho^{(l,h)}_k \le (1+\epsilon)\rho^l_j$ for $j\in T^{(l,h)}_{k}$.

In conclude, we have divided items in $I^{(l,m)}$ into at most $\tilde\OO(1/\epsilon^{s_B+1})$ sub-groups such that in each sub-group, items have the same shape and the same rounded ratio. This completes the preprocessing part. Again notice that we do not substitute item weights with rounded weights, and the instance is thus still $\hat{I}$. 

\paragraph{Generating a well-structured follower's solution.}\label{sub:well_stru}
Let $\vex^{r^*}$ be an optimal solution to \textbf{MBi-IP}$_r(\tilde{I},\ve1,\ve1)$. Given $l \in \{1,2,\cdots,|\Theta|\}$, consider the following linear program:
    		\begin{align*}
    		\textbf{LP}(l;\vex^{r^*}): \quad\max\limits_{\vey[\hat{I}]}\quad  & \sum_{j\in \hat{I}}\tilde{p}_{j}y_{j}&\\
    		&  \sum_{j\in \hat{I}} \ve B^l_jy_j\le \ve1&  \\
    		&  y_j \le 1-x^{r^*}_j,\quad \forall j \in \hat{I} & \\
    		& y_j\in [0,1],\hspace{7mm} \forall j \in \hat{I}& 
    		\end{align*}

In the following, we are going to generate a well-structured near-optimal solution to \textbf{LP}$(l;\vex^{r^*})$. According to the classification of items, items in $I^{(l,s)}$ are the ones that the follower least want to select and items in $I^{(l,b)}$ are the ones that the follower most want to select, whereas for the items in $I^{(l,m)}$, we are not clear about the preference of the follower, but we have subdivided items in $I^{(l,m)}$ into subgroups so that we may approximately guess the follower's choice on each subgroup. Based on this, we have the following lemma.
\begin{lemma}\label{lemma:rounding_m}
    If $OPT \le 1$, then there exists a well-structured feasible solution to \textbf{LP}$(l;\vex^{r^*})$, say $\tilde{\vey}^{(l;\vex^{r^*})}= (\tilde{y}^{(l;\vex^{r^*})}_1,\tilde{y}^{(l;\vex^{r^*})}_2,\cdots,\tilde{y}^{(l;\vex^{r^*})}_{\hat{n}})$, such that 
    \begin{itemize}
        \item $\tilde{y}^{(l;\vex^{r^*})}_j=1-x_j^{r^*}$ for $j\in I^{(l,b)}$;
        \item $\tilde{y}^{(l;\vex^{r^*})}_j=0$ for $j\in I^{(l,s)}$;
        \item for each $T^{(l,h)}_k$, either $\sum_{j\in T^{(l,h)}_k} w^l_j \tilde{y}^{(l;\vex^{r^*})}_j=0$ or $\sum_{j\in T^{(l,h)}_k} w^l_j \tilde{y}^{(l;\vex^{r^*})}_j= \epsilon^{s_B+4}(1+\epsilon)^t$ for some $t \in \tilde{\OO}(1/\epsilon).$
    \end{itemize}
    Furthermore, let $OPT^{(l;\vex^{r^*})}$ be the optimal objective value of $\textbf{LP}(l;\vex^{r^*})$, then the objective value of $\textbf{LP}(l;\vex^{r^*})$ for $\vey[\hat{I}] = \tilde{\vey}^{(l;\vex^{r^*})}$ is at least $OPT^{(l;\vex^{r^*})} -\OO(s_B \epsilon)$.
\end{lemma}
Let $\vey^{(l;\vex^{r^*})}$ and $OPT^{(l;\vex^{r^*})}$ be an optimal solution and the optimal objective value of $\textbf{LP}(l;\vex^{r^*})$, respectively.
Towards the proof of Lemma~\ref{lemma:rounding_m},
we have the following observations.
\begin{observation}\label{obs:r_m_1}
        If $OPT \le 1$, then $P_l + OPT^{(l;\vex^{r^*})}\le s_B+\OO(s_B\delta)$ for any $l \in \{1,2,\cdots,|\Theta|\}$.
\end{observation}
\begin{proof}   
    Recall that $\vex^{r^*}$ is an optimal solution to $\textbf{MBi-IP}_r(\tilde{I},\ve1,\ve1)$ and the optimal objective value of $\textbf{MBi-IP}_r(\tilde{I},\ve1,\ve1)$ is at most $s_B+\OO(s_B\delta)$. Note that $l$ and $\vey^{(l;\vex^{r^*})}$ form a feasible solution to program \eqref{MBi-IP_r:d} - \eqref{MBi-IP_r:c2} when $ \vex = \vex^{r^*}$, it thus follows that $P_l + OPT^{(l;\vex^{r^*})} \le s_B+\OO(s_B\delta)$.\end{proof}

\begin{observation}\label{obs:ropt_a}
For any $\vey[\hat{I}]\in[0,1]^{\hat{n}}$ satisfying $\sum_{j\in \hat{I}}\ve B^l_j y_j \le \ve1$, we have $\sum\limits_{j \in I^{(l,s)}}\tilde{p}_jy_j \le \OO(s_B\epsilon)$.
\end{observation}
\begin{proof}
For any $\vey[\hat{I}]\in[0,1]^{\hat{n}}$ satisfying $\sum_{j\in \hat{I}}\ve B^l_j y_j \le \ve1$, it suffices to observe that
$$
\sum\limits_{j \in I^{(l,s)}}\tilde{p}_jy_j = \sum\limits_{j \in I^{(l,s)}}\rho^l_jw^l_jy_j \le \epsilon \sum\limits_{j \in I^{(l,s)}}w^l_jy_j,
$$ 
$$ 
\sum\limits_{j \in I^{(l,s)}}w^l_jy_j\le \sum_{i=1}^{s_B}\sum_{j\in \hat{I}}\ve B^l_j[i]y_j\le s_B.
$$\end{proof}

\begin{observation}\label{obs:ropt_b}
    If $OPT \le 1$, then for any $l\in \{1,2,\cdots,|\Theta|\}$, we have $\sum_{j\in I^{(l,b)}} w^l_j(1-x_j^{r^*}) \le 2s_B\epsilon$. Furthermore, there exists a feasible solution $\hat{\vey}$ to \textbf{LP}$(l;\vex^{r^*})$ such that $\hat{y}_j=1-x_j^{r^*}$ for $j\in I^{(l,b)}$. 
\end{observation}
\begin{proof}
We first show that $\sum_{j\in I^{(l,b)}} w^l_j(1-x_j^{r^*}) \le 2s_B\epsilon$. 
If $\sum_{j \in I^{(l,b)}} w^l_j(1-x_j^{r^*}) > 2s_B \epsilon$, then we can define a feasible solution $\vey[\hat{I}]$ to \textbf{LP}$(l;\vex^{r^*})$ such that $y_j=\frac{2s_B \epsilon(1-x_j^{r^*})}{\sum_{j\in I^{(l,b)}} w^l_j(1-x_j^{r^*})}$ for $j\in I^{(l,b)}$, and $y_j=0$ otherwise. It is easy to see that the objective value of \textbf{LP}$(l;\vex^{r^*})$ for this solution is greater than $2s_B$, contradicting the fact that $OPT^{(l;\vex^{r^*})}\le s_B+\OO(\epsilon)$ when $OPT \le 1$. Hence, $\sum_{j\in I^{(l,b)}} w^l_j(1-x_j^{r^*}) \le 2s_B \epsilon$. 

For the second part of Observation~\ref{obs:ropt_b}. Recall that $\vey^{(l;\vex^{r^*})}$ is an optimal solution to \textbf{LP}$(l;\vex^{r^*})$. Define $\hat{\vey}[\hat{I}]\in[0,1]^{\hat{n}}$ such that $\hat{y}_j=1-x_j^{r^*}$ if $j\in I^{(l,b)}$; otherwise $\hat{y}_j=(1-2s_B \epsilon) y_j^{(l;\vex^{r^*})}$. We claim that $\hat{\vey}[\hat{I}]$ is a feasible solution to \textbf{LP}$(l;\vex^{r^*})$. It suffices to observe that
    		\begin{align*}
    		\sum_{j\in \hat{I}} \ve B^l_j\hat{y}_j &=\sum_{j\not\in I^{(l,b)}} \ve B^l_j\hat{y}_j+\sum_{j\in I^{(l,b)}} \ve B^l_j\hat{y}_j&\\
    		& = (1-2s_B \epsilon)\sum_{j\not\in I^{(l,b)}} \ve B^l_j y_j^{(l;\vex^{r^*})} +\sum_{j\in I^{(l,b)}} \ve B^l_j(1-x_j^{r^*})&  \\
    		&\le (1-2s_B \epsilon)\cdot \ve1+2s_B \epsilon\cdot \ve1= \ve1.& 
    		\end{align*}
Hence Observation~\ref{obs:ropt_b} is proved. \end{proof}

Now we are ready to prove Lemma~\ref{lemma:rounding_m}.
\begin{proof}[Proof of Lemma~\ref{lemma:rounding_m}]
Recall that $\vey^{(l;\vex^{r^*})}$ is an optimal solution to $\textbf{LP}(l;\vex^{r^*})$. For each $T^{(l,h)}_k$, we define $L(T^{(l,h)}_k)=\sum_{j\in T^{(l,h)}_k} w^l_jy_j^{(l;\vex^{r^*})}$. Note that each item $j\in T^{(l,h)}_k$ satisfies $\rho^l_j\in (\epsilon,\frac{1}{\epsilon}]$. Since $\sum_{j\in T^{(l,h)}_k} \rho^l_jw^l_jy_j^{(l;\vex^{r^*})} \le \sum_{j\in \hat{I}}\tilde{p}_j y^{(l;\vex^{r^*})}_j \le s_B+\OO(s_B\delta)$, we have $L(T^{(l,h)}_k)\le 2s_B/\epsilon$. If $L(T^{(l,h)}_k) \le \epsilon^{s_B+4}$, then  $\sum_{j\in T^{(l,h)}_k} \tilde{p}_jy_j^{(l;\vex^{r^*})} = \sum_{j\in T^{(l,h)}_k} \rho^l_j w^l_jy_j^{(l;\vex^{r^*})}\le \epsilon^{s_B+3}$. Hence, the overall contribution to profit from these groups is at most $\OO(\epsilon)$. If $L(T^{(l,h)}_k)> \epsilon^{s_B+4}$, note that $L(T^{(l,h)}_k)\le 2s_B/\epsilon$,  we may round down its value to the nearest value of the form $\epsilon^{s_B+4}(1+\epsilon)^t$ for some $t \in\tilde{\OO}(1/\epsilon)$ and let it be $\tilde{L}(T^{(l,h)}_k)$. Define $\tilde{\vey}^{(l;\vex^{r^*})}\in [0,1]^{\hat{n}}$ such that: 
\begin{itemize}
\item[(i)]$\tilde{y}^{(l;\vex^{r^*})}_j = 0$ if $j\in T^{(l,h)}_k$ and $L(T^{(l,h)}_k)\le \epsilon^{s_B+4}$; 
\item [(ii)]$\tilde{y}^{(l;\vex^{r^*})}_j=(1-2s_B \epsilon) \theta^{(l,h)}_k y^{(l;\vex^{r^*})}_j$ if $j\in T^{(l,h)}_k$ and $L(T^{(l,h)}_k)> \epsilon^{s_B+4}$, where $\theta^{(l,h)}_k=\frac{\tilde{L}(T^{(l,h)}_k)}{{L}(T^{(l,h)}_k)}\in [1/(1+\epsilon),1]$;
\item [(iii)]$\tilde{y}^{(l;\vex^{r^*})}_j = 1-x^{r^*}_j $ if $j\in I^{(l,b)}$;
\item [(iv)]$\tilde{y}^{(l;\vex^{r^*})}_j = 0$ if $j\in I^{(l,s)}$.
\end{itemize}
We claim that $\tilde{\vey}^{(l;\vex^{r^*})}$ is a feasible solution to \textbf{LP}$(l;\vex^{r^*})$ with an objective value of at least $OPT^{(l;\vex^{r^*})} -\OO(s_B \epsilon)$. The feasibility is guaranteed by 
    		\begin{align*}
    		\sum_{j\in \hat{I}} \ve B^l_j \tilde{y}^{(l;\vex^{r^*})}_j&=\sum_{j\not\in I^{(l,b)}}\ve B^l_j\tilde{y}^{(l;\vex^{r^*})}_j+ \sum_{j\in I^{(l,b)}} \ve B^l_j\tilde{y}^{(l;\vex^{r^*})}_j&\\
    		& \le (1- 2s_B \epsilon) \sum_{j\not\in I^{(l,b)}}\ve B^l_jy^{(l;\vex^{r^*})}_j + \sum_{j\in I^{(l,b)}} \ve B^l_j(1-x_j^{r^*})&  \\
    		&\le (1-2s_B \epsilon)\cdot \ve1+2s_B \epsilon\cdot \ve1= \ve1.& 
    		\end{align*}
On the other hand, we have
    		\begin{align*}
    		\sum\limits_{j\in \hat{I}}\tilde{p}_j  \tilde{y}^{(l;\vex^{r^*})}_j &=\sum_{j\not\in I^{(l,b)}} \tilde{p}_j\tilde{y}^{(l;\vex^{r^*})}_j + \sum_{j\in I^{(l,b)}} \tilde{p}_j\tilde{y}^{(l;\vex^{r^*})}_j \ge  \frac{1-2s_B \epsilon}{1+\epsilon}\sum_{j\not\in I^{(l,b)}} \tilde{p}_jy_j^{(l;\vex^{r^*})} +\sum_{j\in I^{(l,b)}} \tilde{p}_j y_j^{(l;\vex^{r^*})}-\OO(s_B \epsilon) &\\
    		&\ge \frac{1-2s_B \epsilon}{1+\epsilon}OPT^{(l;\vex^{r^*})}-\OO(s_B \epsilon)\ge OPT^{(l;\vex^{r^*})}-\OO(s_B \epsilon),& 
    		\end{align*}

Hence Lemma~\ref{lemma:rounding_m} is proved. \end{proof}

\paragraph{Finding a solution to $\textbf{MBi-IP}_r(\tilde{I},\ve1,\ve1)$}\label{sub:single-level-LP}

The goal of this section to construct a single level linear program whose extreme point optimal solution is a solution to $\textbf{MBi-IP}_r(\tilde{I},\ve1,\ve1)$ with an objective value of at most $s_B+\OO(\epsilon)$.

Consider an optimal solution $\vex^{r^*}$ to $\textbf{MBi-IP}_r(\tilde{I},\ve1,\ve1)$. For each $l\in \{1,2,\cdots,|\Theta|\}$, we define $\lambda^{(l,h)}_k$ for every $T^{(l,h)}_k$ as follows: 
\begin{itemize}
    \item If $\sum_{j\in T_k^{(l,h)}}w^l_j(1-x^{r^*}_j)\le\epsilon^{s_B+4}$, then $\lambda^{(l,h)}_k=0$;
    \item If $\sum_{j\in T_k^{(l,h)}}w^l_j(1-x^{r^*}_j)\ge 2s_B/\epsilon$, then $\lambda^{(l,h)}_k = 2s_B/\epsilon$; 
    \item If $\epsilon^{s_B+4} < \sum_{j\in T_k^{(l,h)}}w^l_j(1-x^{r^*}_j) < 2s_B/\epsilon$, then $\lambda^{(l,h)}_k$ is  $\sum_{j\in T^{(l,h)}_k}w^l_j(1-x^{r^*}_j)$ rounded down to the nearest value of the form $\epsilon^{s_B+4}(1+\epsilon)^t$ where $t\in \tilde\OO(1/\epsilon)$.
\end{itemize}
For each $T_k^{(l,h)}$, $\lambda^{(l,h)}_k$ may take $\tilde{\OO}(1/\epsilon)$ distinct values. Therefore given $\vex^{r^*}$, with $2^{\tilde{\OO}(s_B/\epsilon^{2s_B+1})}$ enumerations we can guess out $\{\lambda^{(l,h)}_k: \forall l\in \tilde\OO(\frac{s_B}{\epsilon^{s_B}}); \forall  h\in \tilde{\OO}(\frac{1}{ \epsilon^{s_B}}) ;\forall k \in \tilde\OO(\frac{1}{\epsilon})\}$. Furthermore, for each well-structured solution $\tilde{\vey}^{(l;\vex^{r^*})}$ defined in Lemma~\ref{lemma:rounding_m}, we let $\bar{\lambda}^{(l,h)}_k$ denote $\sum_{j\in T^{(l,h)}_k} w_j^l\tilde{y}^{(l;\vex^{r^*})}_j$, recall that $\bar{\lambda}^{(l,h)}_k$ may take $\tilde{\OO}(1/\epsilon)$ distinct values, thus with $2^{\tilde{\OO}(s_B/\epsilon^{2s_B+1})}$ enumerations we can also guess out $\{\bar{\lambda}^{(l,h)}_k: \forall l\in \tilde\OO(\frac{s_B}{\epsilon^{s_B}}); \forall  h\in \tilde{\OO}(\frac{1}{ \epsilon^{s_B}});\forall k \in \tilde\OO(\frac{1}{\epsilon})\}$. It is easy to see that $\bar{\lambda}^{(l,h)}_k \le \lambda^{(l,h)}_k$. Then we construct the following single level LP(given the guessed-out $\lambda^{(l,h)}_k$'s and $\bar{\lambda}^{(l,h)}_k$'s ):
\begin{align*}
    			\textbf{cen-LP}_{\lambda}: \quad \min\limits_{\vex, M} \hspace{2mm} & M &\\
    			s.t. \hspace{2mm} & \sum_{j\in \tilde{I}}
    			\ve A_jx_j \le \ve1 &\\
    			&P_l+\sum_{j\in I^{(l,b)}}\tilde{p}_j(1-x_j)+\sum_{h}\sum_{k}
    			\rho^{(l,h)}_k\bar{\lambda}^{(l,h)}_k \le M, \hspace{2mm} \forall l &\\
    			& \lambda^{(l,h)}_k \leq \sum_{j\in T^{(l,h)}_k}w^l_j(1-x_j) \leq \lambda^{(l,h)}_k(1+\epsilon),\hspace{2mm} \forall l,k, h \text{ such that }\lambda^{(l,h)}_k \neq 0,2s_B/\epsilon&\\
    			&\sum_{j\in T^{(l,h)}_k}w^l_j(1-x_j) \leq \epsilon^{s_B+4},\hspace{2mm} \forall l, k, h \text{ such that } \lambda^{(l,h)}_k=0 & \\
    			&\sum_{j\in T^{(l,h)}_k}w^l_j(1-x_j)\geq 2s_B/\epsilon, \hspace{2mm} \forall l,k, h \text{ such that } \lambda^{(l,h)}_k=2s_B/\epsilon& \\
    			&  x_j = x^*_j, j \in \tilde{I}\setminus \hat{I}&\\
    			&  x_j\in [0,1], j\in \hat{I}&
\end{align*}

Where $\rho^{(l,h)}_{k}$ is the rounded ratio of items in $T^{(l,h)}_k$. Since the number of constraints in \textbf{cen-LP}$_{\lambda}$ is a polynomial with respect to the number of variables, there exists polynomial time algorithm which returns an extreme point optimal solution to \textbf{cen-LP}$_{\lambda}$, denote by $\{\vex^{\lambda},M^{\lambda}\}$ this solution. It is easy to see that $\vex^{\lambda}$ is a feasible solution to $\textbf{MBi-IP}_r(\tilde{I},\ve1,\ve1)$, furthermore, we have the following.

\begin{lemma}\label{lemma:cen-LP}
If $OPT \le 1$, let $\{\vex^{\lambda},M^*\}$ be an extreme point optimal solution to $\textbf{cen-LP}_{\lambda}$, then $\vex^{\lambda}$ is a feasible solution to $\textbf{MBi-IP}_r(\tilde{I},\ve1,\ve1)$ with an objective value of at most $s_B+\OO(\epsilon)$.
\end{lemma}
The rest of this section is dedicated to proving Lemma~\ref{lemma:cen-LP}, which consists three main steps: (i)
given $\vex^{r^*}$ and integer $l\in [1,\Theta]$, we first construct $\textbf{LP}(l;\vex^{r^*})$. Then we construct another two LPs, i.e.,  $\overline{\textbf{LP}}(l;\vex^{r^*})$ and $\textbf{LP}'(l;\vex^{r^*})$. By comparing the optimal objectives of these three LPs, showing that the optimal objective value of $\textbf{LP}(l;\vex^{r^*})$ is approximate to the optimal objective value of $\textbf{LP}'(l;\vex^{r^*})$; (ii) given $\vex^{\lambda}$ and integer $l\in [1,\Theta]$, we first construct $\textbf{LP}(l;\vex^{\lambda})$. Then we construct another two LPs, i.e.,  $\overline{\textbf{LP}}(l;\vex^{\lambda})$ and $\textbf{LP}'(l;\vex^{\lambda})$. By comparing the optimal objectives of these three LPs, showing that with an additive error of $\OO(s_B\epsilon)$, the optimal objective value of $\textbf{LP}(l;\vex^{\lambda})$ is less than $(1+\OO(s_B \epsilon))$ times the optimal 
objective value of $\textbf{LP}'(l;\vex^{\lambda})$; (iii) we complete the proof of Lemma~\ref{lemma:cen-LP} by the results obtained in (i) and (ii). 

\paragraph*{Step 1}
Recall that $\vex^{r^*}$ is an optimal solution to $\textbf{MBi-IP}_r(\tilde{I},\ve1,\ve1)$. Given $\vex^{r^*}$ and $l\in \{1,2,\cdots,|\Theta|\}$, we consider the following LP:
    		\begin{align*}
    		\textbf{LP}(l;\vex^{r^*}): \quad\max\limits_{\vey[\hat{I}]}\quad  & \sum_{j\in \hat{I}}\tilde{p}_{j}y_{j}&\\
    		&  \sum_{j\in \hat{I}} \ve B^l_jy_j\le \ve1&  \\
    		&  y_j \le 1-x^{r^*}_j,\hspace{2mm} \forall j \in \hat{I} & \\
    		& y_j\in [0,1],\hspace{5mm} \forall j \in \hat{I}& 
    		\end{align*}
Let $OPT^{(l;\vex^{r^*})}$ be the optimal objective value of ${\textbf{LP}}(l;\vex^{r^*})$. We have $P_l + OPT^{(l;\vex^{r^*})} \le OPT^{MBi}_r$, where $OPT^{MBi}_r$ is the optimal objective value of $\textbf{MBi-IP}_r(\tilde{I},\ve1,\ve1)$.

Then we define the following two LPs:
    		\begin{align*}
    		\overline{\textbf{LP}}(l;\vex^{r^*}): \quad \max\limits_{\vey[\hat{I}]}\quad  &  \sum_{j\in I^{(l,b)}}\tilde{p}_j(1-x_j^{r^*})+\sum_{j\in I^{(l,m)}}\tilde{p}_{j}y_{j}&\nonumber\\
    		& \sum_{j\in \hat{I}} \ve B^l_jy_j\le \ve1&  \\
    		&y_j = 0,\hspace{10mm} j\in I^{(l,s)} &\\
    		& y_j = 1-x^{r^*}_{j}, \hspace{2mm} j\in I^{(l,b)}&  \\
    		& y_j \le 1-x^{r^*}_{j}, \hspace{2mm} j\in I^{(l,m)}&  \\
    		& y_j \in [0,1], \hspace{5mm} j\in \hat{I} & \nonumber
    		\end{align*}

    		\begin{align*}
    		\textbf{LP}'(l;\vex^{r^*}): \quad \max\limits_{\vey[\hat{I}]}\quad  &  \sum_{j\in I^{(l,b)}}\tilde{p}_j(1-x_j^{r^*})+\sum_{j\in I^{(l,m)}}\tilde{p}_{j}y_{j}&\nonumber\\
    		& \sum_{j\in \hat{I}} \ve B^l_jy_j\le \ve1& \nonumber \\
    		& \sum_{j\in T^{(l,h)}_k}  w^l_jy_j\le \sum_{j\in T^{(l,h)}_k} w^l_j(1-x^{r^*}_j), \hspace{2mm}  \forall k, h &\\
    		&y_j = 0,\hspace{10mm} j\in I^{(l,s)} &\\
    		& y_j = 1-x^{r^*}_{j}, \hspace{2mm} j\in I^{(l,b)}&  \\
    		& y_j \in [0,1],\hspace{5mm} j\in I^{(l,m)} & \nonumber
    		\end{align*}
Let $OPT^{(l;\vex^{r^*})}_1$ and $OPT^{(l;\vex^{r^*})}_2$ be the optimal objective values of $\overline{\textbf{LP}}(l;\vex^{r^*})$ and $\textbf{LP}'(l;\vex^{r^*})$, respectively. Then we compare  $OPT^{(l;\vex^{r^*})}_1$ with $OPT^{(l;\vex^{r^*})}$ and $OPT^{(l;\vex^{r^*})}_2$ separately. 

\subparagraph*{Step 1.1 Compare $OPT^{(l;\vex^{r^*})}_1$ with $OPT^{(l;\vex^{r^*})}$.}{ }~{ }\newline

We have the following observation.
\begin{observation}\label{obs:opt1}
    If $OPT \le 1$, then we have
    $$OPT^{(l;\vex^{r^*})}-\OO(s_B\epsilon)\le OPT^{(l;\vex^{r^*})}_1 \le OPT^{(l;\vex^{r^*})}\le s_B+\OO(s_B\delta).$$
\end{observation}
\begin{proof}
It follows directly from Lemma~\ref{lemma:rounding_m} and Observation~\ref{obs:r_m_1}.
\end{proof}

\subparagraph*{Step 1.2 Compare $OPT^{(l;\vex^{r^*})}_1$ with $OPT_2^{(l;\vex^{r^*})}$.}{ }~{ }\newline

We have the following lemma.
\begin{lemma}\label{lemma:opt1-opt2}
        If $OPT \le 1$, then we have  
        $$OPT^{(l;\vex^{r^*})}_1 \le OPT^{(l;\vex^{r^*})}_2 \le OPT^{(l;\vex^{r^*})}_1+\OO(s_B^2 \epsilon).$$
\end{lemma}
\begin{proof}
It is easy to see that $OPT^{(l;\vex^{r^*})}_1 \le OPT^{(l;\vex^{r^*})}_2$ since any feasible solution to $\overline{\textbf{LP}}(l;\vex^{r*})$ is also a feasible solution to $\textbf{LP}'(l;\vex^{r*})$. In the next, we show that $OPT^{(l;\vex^{r^*})}_2 \le (1+\OO(s_B \epsilon))OPT^{(l;\vex^{r^*})}_1$. Then by the fact that $OPT^{(l;\vex^{r^*})}_1 \le s_B+\OO(s_B\delta)$, we have $OPT^{(l;\vex^{r^*})}_2 \le OPT^{(l;\vex^{r^*})}_1+\OO(s_B^2 \epsilon)$, which completes the proof.

Let $\tilde{\vey}[\hat{I}]$ be an optimal solution to $\textbf{LP}'(l;\vex^{r*})$, it suffices to show that there exists a feasible solution $\vey[\hat{I}]$ to $\overline{\textbf{LP}}(l;\vex^{r*})$ such that $
\sum_{j\in \hat{I}}\tilde{p}_j \tilde{y}_j\le (1+\OO(s_B \epsilon))\sum_{j\in \hat{I}}\tilde{p}_j y_j$.  
Towards this, for each $h$ and each $k$, we define the following LP (given $\tilde{\vey}[\hat{I}]$):
    		\begin{align*}
    		\textbf{LP}^{(l,h)}_k:\quad &\sum_{j\in T_k^{(l,h)}}w^l_{j}y_{j} = \sum_{j\in T_k^{(l,h)}}w^l_{j}\tilde{y}_{j} &\\
    		&y_{j}\leq 1-x_{j}^{r*} ,\hspace{2mm} \forall j\in T_k^{(l,h)} & 
    		\end{align*}
Since $\sum_{j\in T_k^{(l,h)}} w^l_j\tilde{y}_j\le \sum_{j\in T_k^{(l,h)}} w^l_j(1-x_j^{r*})$, then one feasible solution to \textbf{LP}$^{(l,h)}_k$ can be obtained by a simple greedy algorithm that selects remaining
fractional items in $T_k^{(l,h)}$ in the natural order of indices, until the budget $\sum_{j\in T_k^{(l,h)}} w^l_j\tilde{y}_j$ is exhausted. 

Define $\vey^{l}\in [0,1]^{|\hat{I}|}$ such that for each $h$ and each $k$, $\{y^{l}_j:j\in T_k^{(l,h)}\}$ is a feasible solution to \textbf{LP}$^{(l,h)}_k$. Note that items in $T_k^{(l,h)}$ have the same shape ${\veg}_k^{(l,h)}$, thus we have $$
    \sum_{j\in T_k^{(l,h)}}{\ve B}^l_jy^{l}_j \le \sum_{j\in T_k^{(l,h)}}\bar{\ve B}^l_jy^{l}_j = {\veg}_k^{(l,h)}\sum_{j\in T_k^{(l,h)}}w^l_jy^{l}_j = {\veg}_k^{(l,h)}\sum_{j\in T_k^{(l,h)}}w^l_j\tilde{y}_j = \sum_{j\in T_k^{(l,h)}}\bar{\ve B}^l_j\tilde{y}_j.$$
It follows that
$$ 
\sum_{h}\sum_{k}\sum_{j\in T_k^{(l,h)}}{\ve B}^l_jy^{l}_j \le\sum_{h}\sum_{k}\sum_{j\in T_k^{(l,h)}}\bar{\ve B}^l_j\tilde{y}_j
     \le (1+(s_B+1)\epsilon)\cdot \ve1,
$$
where the last inequality holds by Observation~\ref{observation: obs_frac1}. 
    
According to Observation~\ref{obs:ropt_b}, we have $\sum_{j\in I^{(l,b)}}w^l_j(1-x_j^{r^*}) \le 2s_B \epsilon$. Define $\vey[\hat{I}] $ such that: (i) $y_j = \frac{(1-2s_B \epsilon)y^{l}_j}{1+(s_B+1)\epsilon}$ if $j\in T_k^{(l,h)}$; (ii) $y_j = 1- x^{r*}_j $ if $j\in I^{(l,b)}$; (iii)$y_j = 0 $ if $j\in I^{(l,s)}$. We claim that $\vey[\hat{I}]$ is a feasible solution to $\overline{\textbf{LP}}(l;x^{r*})$. It suffices to observe that 
$$
\sum_{j\in \hat{I}} \ve B^l_jy_j = \sum_{j\in I^{(l,b)}}\ve B^l_jy_j +\sum_{h}\sum_{k}\sum_{j\in T_k^{(l,h)}}\ve B^l_jy_j \le 2s_B  \epsilon \cdot \ve1 + (1-2s_B  \epsilon)\cdot \ve1 \le \ve1.
$$

On the other hand, we have
    		\begin{align*}
    		\sum_{j\in T_k^{(l,h)}}\tilde{p}_j\tilde{y}_j & = \sum_{j\in T_k^{(l,h)}}\rho^l_j w^l_j\tilde{y}_j
    		\le (1+\epsilon) \rho^{(l,h)}_k \sum_{j\in T_k^{(l,h)}}w^l_j\tilde{y}_j&\\
    		&= (1+\epsilon)\frac{1+(s_B+1)\epsilon}{1-2s_B \epsilon}\rho^{(l,h)}_k \sum_{j\in T_k^{(l,h)}}w^l_jy_j &\\
    		&\le (1+\epsilon)\frac{1+(s_B+1)\epsilon}{1-2s_B \epsilon}\sum_{j\in T_k^{(l,h)}}\tilde{p}_jy_j.&
    		\end{align*}
It follows that 
$
\sum\limits_{j\in \hat{I}}\tilde{p}_j \tilde{y}_j \le  \sum\limits_{j\in I^{(l,b)}}\tilde{p}_j{y}_j  + (1+\epsilon)\frac{1+(1+s_B)\epsilon}{1-2s_B \epsilon}\sum\limits_{h}\sum\limits_{k}\sum\limits_{j\in T_k^{(l,h)}}\tilde{p}_j{y}_j\le (1+\OO(s^2_B\epsilon))\sum\limits_{j\in \hat{I}}\tilde{p}_j y_j.
$
\end{proof}

Besides, we have the following observation.
\begin{observation}\label{obs:lambda}
    If $OPT \le 1$, then for any $l\in \{1,2,\cdots,|\Theta|\}$, we have
    $$
    OPT^{(l;\vex^{r^*})}-\OO(s_B \epsilon) \le \sum_{j\in I^{(l,b)}}\tilde{p}_j(1-x^{r^*}_j)+\sum_{h}\sum_{k} 	\rho^{(l,h)}_k\bar{\lambda}^{(l,h)}_k \le  OPT^{(l;\vex^{r^*})}.$$ Moreover, we have 
    $$
    P_l+\max\{OPT^{(l;\vex^{r^*})}_1, OPT^{(l;\vex^{r^*})}_2\} \le s_B+\OO(s_B^2 \epsilon).$$  
\end{observation}

\begin{proof}
Given $\vey^{(l;\vex^{r^*})}$ an optimal solution to \textbf{LP}$(l;\vex^{r^*})$, let $\tilde{\vey}^{(l;\vex^{r^*})}$ be the well-structured feasible solution to \textbf{LP}$(l;\vex^{r^*})$ defined in Lemma~\ref{lemma:rounding_m}, then we have 
    		\begin{align*}
    		OPT^{(l;\vex^{r^*})} &\ge \sum_{j\in I^{(l,b)}}\tilde{p}_j(1-x^{r^*}_j)+\sum_{h}\sum_{k}\sum_{j \in T^{(l,h)}_k} \tilde{p}_j \tilde{y}^{(l;\vex^{r^*})}_j&\\
    		&\ge \sum_{j\in I^{(l,b)}}\tilde{p}_j(1-x^{r^*}_j)+\sum_{h}\sum_{k} \rho^{(l,h)}_k\bar{\lambda}^{(l,h)}_k &\\
    		&\ge \sum_{j\in I^{(l,b)}}\tilde{p}_j(1-x^{r^*}_j)+\frac{1}{1+\epsilon}\sum_{h}\sum_{k}\sum_{j \in T^{(l,h)}_k}\rho^l_j w^l_j \tilde{y}^{(l;\vex^{r^*})}_j&\\
    		&\ge \frac{1}{1+\epsilon}\left(\sum_{j\in I^{(l,b)}}\tilde{p}_j(1-x^{r^*}_j)+\sum_{h}\sum_{k}\sum_{j \in T^{(l,h)}_k}\tilde{p}_j \tilde{y}^{(l;\vex^{r^*})}_j\right)&\\
    		&\ge \frac{1}{1+\epsilon} (OPT^{(l;\vex^{r^*})}-\OO(s_B \epsilon)).&
    		\end{align*}
According to Observation~\ref{obs:r_m_1}, $ OPT^{(l;\vex^{r^*})} \le s_B+\OO(s_B\delta)$, it thus follows that $$
OPT^{(l;\vex^{r^*})}-\OO(s_B \epsilon) \le \sum_{j\in I^{(l,b)}}\tilde{p}_j(1-x^{r^*}_j)+\sum_{h}\sum_{k} 	\rho^{(l,h)}_k\bar{\lambda}^{(l,h)}_k\le  OPT^{(l;\vex^{r^*})}.
$$

For the second part of the observation, notice that $\max\{OPT^{(l;\vex^{r^*})}_1,OPT^{(l;\vex^{r^*})}_2 \}
\le OPT^{(l;\vex^{r^*})}+\OO(s_B^2\epsilon)$ and $P_l + OPT^{(l;\vex^{r^*})}\le s_B+\OO(s_B\delta)$, thus 
$$P_l+ \max\{OPT^{(l;\vex^{r^*})}_1,OPT^{(l;\vex^{r^*})}_2 \}\le P_l+OPT^{(l;\vex^{r^*})}+\OO(s_B^2\epsilon)\le s_B+\OO(s_B^2\epsilon). $$ \end{proof}

\paragraph*{Step 2}
Recall that $\vex^{\lambda}$ is an optimal basic solution of $\textbf{cen-LP}_{\lambda}$. Given $\vex^{\lambda}$ and any $l\in \{1,2,\cdots,|\Theta|\}$, we consider the following LP:
    		\begin{align*}
    		\textbf{LP}(l;\vex^{\lambda}): \quad\max\limits_{\vey[\hat{I}]}\quad  & \sum_{j\in \hat{I}}\tilde{p}_{j}y_{j}&\\
    		&  \sum_{j\in \hat{I}} \ve B^l_j[i]y_j\le \ve1&  \\
    		&  y_j \le 1-x^{\lambda}_j, \hspace{2mm} \forall j \in \hat{I} & \\
    		& y_j\in [0,1], \hspace{4mm} \forall j \in \hat{I}& 
    		\end{align*}
Let $OPT^{(l;\vex^{\lambda})}$ be the optimal objective value of $\textbf{LP}(l;\vex^{\lambda})$. 
Let $Obj^{MBi}_r(\vex^{\lambda})$ be the objective value of $\textbf{MBi-IP}_r(\tilde{I},\ve1,\ve1)$ for $\vex = \vex^{\lambda}$, we have $P_l + OPT^{(l;\vex^{\lambda})} \le Obj^{MBi}_r(\vex^{\lambda})$.

Then we define the following two LPs:
    		\begin{align*}
    		\overline{\textbf{LP}}(l;\vex^{\lambda}): \quad \max\limits_{\vey[\hat{I}]}\quad  &  \sum_{j\in I^{(l,b)}}\tilde{p}_j(1-x_j^{\lambda})+\sum_{j\in I^{(l,m)}}\tilde{p}_{j}y_{j}&\nonumber\\
    		& \sum_{j\in \hat{I}} \ve B^l_jy_j\le \ve1&  \\
    		&y_j = 0,\hspace{9mm} j\in I^{(l,s)} &\\
    		& y_j = 1-x^{\lambda}_{j}, \hspace{2mm} j\in I^{(l,b)}& \\
    		& y_j \le 1-x^{\lambda}_{j}, \hspace{2mm} j\in I^{(l,m)}&\\
    		& y_j \in [0,1],\hspace{5mm} j\in \hat{I} & \nonumber
    		\end{align*}

    		\begin{align*}
    		\textbf{LP}'(l;\vex^{\lambda}): \quad \max\limits_{\vey[\hat{I}]}\quad  &  \sum_{j\in I^{(l,b)}}\tilde{p}_j(1-x_j^{\lambda})+\sum_{j\in I^{(l,m)}}\tilde{p}_{j}y_{j}&\nonumber\\
    		& \sum_{j\in \hat{I}} \ve B^l_jy_j\le \ve1&  \\
    		& \sum_{j\in T^{(l,h)}_k}  w^l_jy_j\le \sum_{j\in T^{(l,h)}_k} w^l_j(1-x^{\lambda}_j),\hspace{2mm}  \forall k, h & \label{LP_4:b}\\
    		&y_j = 0,\hspace{10mm} j\in I^{(l,s)} &\\
    		& y_j = 1-x^{\lambda}_{j}, \hspace{3mm} j\in I^{(l,b)}& \\
    		& y_j \in [0,1],\hspace{6mm} j\in I^{(l,m)} & \nonumber
    		\end{align*}
Let $OPT^{(l;\vex^{\lambda})}_1$ and $OPT^{(l;\vex^{\lambda})}_2$ be the optimal objective values of $\overline{\textbf{LP}}(l;\vex^{\lambda})$ and $\textbf{LP}'(l;\vex^{\lambda})$, respectively. Then we compare $OPT^{(l;\vex^{\lambda})}_1$ with $OPT^{(l;\vex^{\lambda})}$ and $OPT^{(l;\vex^{\lambda})}_2$ separately.

\subparagraph*{Step 2.1 Compare $OPT^{(l;\vex^{\lambda})}_1$ with $OPT^{(l;\vex^{\lambda})}$}{ }~{ }\newline

We claim that $(1-\OO(s_B\epsilon))OPT^{(l;\vex^{\lambda})}-\OO(s_B\epsilon) \le  OPT^{(l;\vex^{\lambda})}_1 \le OPT^{(l;\vex^{\lambda})}$. It follows directly from the following lemma.
\begin{lemma}\label{lemma:lambda}
If $OPT\le 1$, then there exists a well-structured solution $\bar{\vey}^{(l;\vex^{\lambda})}$ to  \textbf{LP}$(l;\vex^{\lambda})$ such that
\begin{itemize}
    \item $\bar{\vey}_j^{(l;\vex^{\lambda})} = 1- x_j^{\lambda}$ for $j\in I^{(l,b)}$;
    \item $\bar{\vey}^{(l;\vex^{\lambda})}_j = 0$ for $j\in I^{(l,s)}$.
\end{itemize}
Furthermore, the objective value of
\textbf{LP}$(l;\vex^{\lambda})$ for $\vey[\hat{I}] = \bar{\vey}^{(l;\vex^{\lambda})}$ is at least $(1-\OO(s_B\epsilon))OPT^{(l;\vex^{\lambda})}-\OO(s_B\epsilon)$.


\end{lemma}

Towards the proof of Lemma~\ref{lemma:lambda},
we first give the following observations.

\begin{observation}\label{obs:M^*}
Let $M^{r^*}= \max\{P_l+\sum_{j\in I^{(l,b)}}\tilde{p}_j(1-x^{r^*}_j)+\sum_{h}\sum_{k}\rho^{(l,h)}_k\bar{\lambda}^{(l,h)}_k : \forall l \in \mathbb{Z} \cap [1,|\Theta|] \}$. Then $\vex^{r^*}$ and $M^{r^*}$ form a feasible solution to \textbf{cen-LP}$_{\lambda}$. Furthermore, if $OPT \le 1$, then $M^{r^*} \le s_B+\OO(s_B\delta)$.
\end{observation}
\begin{proof}
The first part of the observation is easy to see. Consider the second part of the observation. Let $\vex^{r^*}$ be an optimal solution to $\textbf{MBi-IP}_r(\tilde{I},\ve1,\ve1)$, we assume that $M^{r^*} = P_{l'}+\sum_{j\in I^{(l',b)}}\tilde{p}_j(1-x^{r^*}_j)+\sum_{h}\sum_{k}\rho^{(l',h)}_k\bar{\lambda}^{(l',h)}_k$ for some $l'\in \{1,2,\cdots,|\Theta|\}$. Let $\vey^{(l';\vex^{r^*})}$ be an optimal solution of the following linear program:
    		\begin{align*}
    		\textbf{LP}(l';\vex^{r^*}): \quad\max\limits_{\vey[\hat{I}]}\quad  & \sum_{j\in \hat{I}}\tilde{p}_{j}y_{j}&\\
    		&  \sum_{j\in \hat{I}} \ve B^{l'}_jy_j\le \ve1&  \\
    		&  y_j \le 1-x^{r^*}_j , \forall j \in \hat{I} & \\
    		& y_j\in [0,1] , \forall j \in \hat{I}& 
    		\end{align*}

It is easy to see that $P_{l'}+ \sum_{j\in \hat{I}}\tilde{p}_j y_j^{(l';\vex^{r^*})} \le OPT^{MBi}_{r}$. According to Lemma~\ref{lemma:rounding_m}, there exists a well-structured feasible solution to \textbf{LP}$(l';\vex^{r^*})$, say $\tilde{\vey}^{(l'; \vex^{r^*})}$, such that 
$\sum_{j\in \hat{I}}\tilde{p}_j \tilde{y}^{(l'; \vex^{r^*})}_j =
\sum_{j\in I^{(l',b)}} \tilde{p}_j(1-x^{r^*}_j)+\sum_{h}\sum_{k}\sum_{j \in T^{(l',h)}_k} \rho^{l'}_j w^{l'}_j \tilde{y}^{(l'; \vex^{r^*})}_j$ and $\sum_{j \in T^{(l',h)}_k} w^{l'}_j \tilde{y}^{(l'; \vex^{r^*})}_j = \bar{\lambda}^{(l',h)}_k$. Since $\rho^{(l',h)}_k \le \rho^{l'}_j$ for $j\in T^{(l',h)}_k$, we have
$$P_{l'}+ \sum_{j\in I^{(l',b)}} \tilde{p}_j(1-x^{r^*}_j)+\sum_{h}\sum_{k} \rho^{(l',h)}_k\bar{\lambda}^{(l',h)}_k \le 
P_{l'}+ \sum_{j\in \hat{I}} \tilde{p}_j \tilde{y}^{(l';\vex^{r^*})}_{j}\le  P_{l'}+\sum_{j\in \hat{I}}\tilde{p}_jy^{(l';\vex^{r^*})}_{j}.$$
Then the observation follows by the fact that $OPT^{MBi}_r \le s_B+\OO(s_B\delta)$ if $OPT \le 1$ (Lemma~\ref{lemma:opt^MBi_r}).
\end{proof}

\begin{observation}\label{obs:l-w}
If $OPT\le 1$, then $\sum_{j\in I^{(l,b)}}w^l_j(1-x^{\lambda}_j)\le 2s_B \epsilon$ for each $l\in \{1,2,\cdots,|\Theta|\}$.
\end{observation}
\begin{proof}
Let $\{\vex^{\lambda}, M^{\lambda}\}$ be an optimal solution to \textbf{cen-LP}$_{\lambda}$. For each $l\in \{1,2,\cdots,|\Theta|\}$, it holds that $$P_l+\sum_{j\in I^{(l,b)}}\tilde{p}_j(1-x^{\lambda}_j)+\sum_{h}\sum_{k}
    			\rho^{(l,h)}_k\bar{\lambda}^{(l,h)}_k \le M^{\lambda} \le M^{r^*}\le s_B+ \OO(\epsilon),$$
where the last inequality holds by Observation~\ref{obs:M^*}. Recall that $\rho^l_j > {1 }/{ \epsilon}$ for $j\in I^{(l,b)}$, thus we have 
$$\sum_{j\in I^{(l,b)}}w^l_j(1-x^{\lambda}_j)\le \epsilon \sum_{j\in I^{(l,b)}}\tilde{p}_j(1-x^{\lambda}_j)\le \epsilon(s_B+\OO(\epsilon))\le 2s_B \epsilon.$$\end{proof}

Now we are ready to prove Lemma~\ref{lemma:lambda}.
\begin{proof}[Proof of Lemma~\ref{lemma:lambda}]

Since $OPT \le 1$, we have $\sum_{j\in I^{(l,b)}}w^l_j(1-x^{\lambda}_j) \le 2s_B\epsilon$ by Observation~\ref{obs:l-w}. 
Let $\vey^{(l;\vex^{\lambda})}$ be an optimal solution of $\textbf{LP}(l;\vex^{\lambda})$, define $\bar{\vey}^{(l;\vex^{\lambda})} \in [0,1]^{\hat{n}}$ such that:
\begin{itemize}
    \item[(i)]$\bar{y}^{(l;\vex^{\lambda})}_{j} = 1-x_{j}^{\lambda}$ if $j\in I^{(l,b)}$;
    \item[(ii)] $\bar{y}^{(l;\vex^{\lambda})}_{j} = 0$ if $j \in I^{(l,s)}$;
    \item[(iii)] $\bar{y}^{(l;\vex^{\lambda})}_{j} = (1-2s_B\epsilon)y_{j}^{(l;\vex^{\lambda})}$ if $j \in I^{(l,m)}$.
\end{itemize}
We claim that $\bar{\vey}^{(l;\vex^{\lambda})}$ is a feasible solution to \textbf{LP}$(l;\vex^\lambda)$ with an objective value of at least $(1-\OO(s_B\epsilon))\sum_{j\in \hat{I}} \tilde{p}_jy_j^{(l;\vex^{\lambda})}-\OO(s_B \epsilon)$. The feasibility is guaranteed by
    \begin{align*}
    \sum_{j\in \hat{I}} \ve B^l_j\bar{y}^{(l;\vex^{\lambda})}_{j} &\le \sum_{j\in I^{(l,b)}} \ve B^l_j(1-x_j^{\lambda})+(1-2s_B\epsilon)\sum_{j\not\in I^{(l,b)}} \ve B^l_j y_j^{(l;\vex^{\lambda})}\\
    &\le  \sum_{j\in I^{(l,b)}} w^l_j(1-x_j^{\lambda})\cdot \ve1+(1-2s_B\epsilon)\sum_{j\not\in I^{(l,b)}} \ve B^l_j y_j^{(l,\lambda)} \\
    &\le 2s_B\epsilon \cdot \ve1 +(1-2s_B\epsilon)\cdot \ve1=\ve1.
    \end{align*}
On the other hand, 
we have
    \begin{align*}
    \sum_{j\in \hat{I}}\tilde{p}_j\bar{y}^{(l;\vex^{\lambda})}_j&=\sum_{j\in I^{(l,b)}} \tilde{p}_j\bar{y}^{(l;\vex^{\lambda})}_{j}+\sum_{j\in I^{(l,m)}} \tilde{p}_j\bar{y}^{(l;\vex^{\lambda})}_{j}+\sum_{j\in I^{(l,s)}} \tilde{p}_j\bar{y}^{(l;\vex^{\lambda})}_{j}\\
    &\ge \sum_{j\in I^{(l,b)}} \tilde{p}_j y_j^{(l;\vex^{\lambda})} +(1-2s_B \epsilon)\sum_{j\not\in I^{(l,b)}} \tilde{p}_jy_j^{(l;\vex^{\lambda})} -\OO(s_B \epsilon)\\
    &\ge (1-2s_B \epsilon)\sum_{j\in \hat{I}} \tilde{p}_jy_j^{(l,\lambda)}-\OO(s_B \epsilon).
    \end{align*}
Where $\OO(s_B \epsilon)$ comes from Observation~\ref{obs:ropt_a}. Hence, Lemma~\ref{lemma:lambda} is proved. 
\end{proof}

\subparagraph*{Step 2.2 Compare $OPT^{(l;\vex^{\lambda})}_1$ with $OPT^{(l;\vex^{\lambda})}_2$}{ }~{ }\newline

We have the following observation.
\begin{observation}
        If $OPT \le 1$, then we have 
        $OPT^{(l;\vex^{\lambda})}_1 \le  OPT^{(l;\vex^{\lambda})}_2$.
\end{observation}
\begin{proof}
It follows directly from the fact that any feasible solution to $\overline{\textbf{LP}}(l;\vex^{\lambda})$ is also a feasible solution to $\textbf{LP}'(l;\vex^{\lambda})$. 
\end{proof}

\paragraph*{Step 3} Recall that $\vex^{\lambda}$ is an extreme point optimal solution to \textbf{cen-LP}$_{\lambda}$. The objective value of $\textbf{MBi-IP}_r(\tilde{I},\ve1,\ve1)$ for $\vex = \vex^{\lambda}$ is the optimal objective value of the following program:
	\begin{align*}
	\textbf{Lower-LP:} \quad \max_{1\le l\le |\Theta|} \max_{\vey[\hat{I} ]} \quad& P_{\ell}+ \sum_{j\in \hat{I}}\tilde{p}_{j}y_{j} & \\
	 s.t. \hspace{3mm}  &\sum_{j\in \hat{I}} \ve B^l_{j}y_{j}\le \ve1 &\\
	& y_{j} \le 1-x^{\lambda}_j, \hspace{2mm} \forall j\in \hat{I} & \\
	&   y_j\in [0,1], \hspace{4mm} \forall j\in \hat{I} &
	\end{align*}
Assume that $l'$ and $\vey'[\hat{I}]$ together form an optimal solution to $\textbf{Lower-LP}$. If $OPT\le 1$, according to the previous discussions, we have 
$$
(1-\OO(s_B \epsilon))(P_{l'}+\sum_{j\in\hat{I}}\tilde{p}_jy_j' )-\OO(s_B \epsilon)\le P_{l'}+ OPT_1^{(l';\vex^{\lambda})} \le P_{l'}+ OPT_2^{(l';\vex^{\lambda})}.
$$
If we can prove that the $P_{l'}+OPT_2^{(l';\vex^{\lambda})} \le s_B+\OO(\epsilon)$, then we have $P_{l'}+\sum_{j\in\hat{I}}\tilde{p}_jy'_j \le s_B+\OO(s^2_B \epsilon)$, which completes the proof of Lemma~\ref{lemma:cen-LP}. From now on, we aim to prove the following lemma.
\begin{lemma}\label{claim:fra_3}
    If $OPT\le 1$, assuming that $l'$ and $\vey'[\hat{I}]$ form an optimal solution to $\textbf{Lower-LP}$, then $P_{l'}+OPT_2^{(l';\vex^{\lambda})} \le s_B+\OO(\epsilon)$.
\end{lemma}
\begin{proof}
    Let $\hat{\vey}[\hat{I}]$ be an optimal solution to $\textbf{LP}'(l';\vex^{\lambda})$. 
    
    When $\lambda^{(l',h)}_k \neq 0$ and $\lambda^{(l',h)}_k \neq \frac{2s_B}{\epsilon}$, recall that both $\{\vex^{r*},M^{r^*}\}$ and $\{\vex^{\lambda},M^{\lambda}\}$ are feasible solutions to $\textbf{cen-LP}_{\lambda}$, we have $$\sum_{j\in T_k^{(l',h)}}w^{l'}_{j}\hat{y}_{j} \leq\sum_{j\in T_k^{(l',h)}}w^{l'}_{j}(1-x^{\lambda}_{j}) \leq (1+\epsilon)\sum_{j\in T_k^{(l',h)}}w^{l'}_{j}(1-x^{r*}_{j}).$$
    
    When $ \lambda_k^{(l',h)} = 2s_B/\epsilon$, we have $ \sum_{j\in T_k^{(l',h)}}w^{l'}_{j}(1-x^{r*}_{j}) \ge 2s_B/\epsilon \ge 1 $. Note that items in $T_k^{(l',h)}$ have the same shape, thus there exists $i(1 \le i \le s_B)$ such that 
    $$ \sum_{j\in T_k^{(l',h)}}w^{l'}_{j}\hat{y}_{j} = \sum_{j\in T_k^{(l',h)}}\ve B^{l'}_{j}[i]\hat{y}_{j} \le 1 \le \sum_{j\in T_k^{(l',h)}}w^{l'}_{j}(1-x^{r*}_{j}).$$
    
    When $\lambda_k^{(l',h)} = 0$, we have $$\sum_{j\in T_k^{(l',h)}}w^{l'}_{j}\hat{y}_{j} \le  \sum_{j\in T_k^{(l',h)}}w^{l'}_{j}(1-x^{\lambda}_{j}) \le \epsilon^{s_B+4}.$$
    Define $H^{l'}=\{(h,k)|\lambda_k^{(l',h)} = 0,\forall k,h\}$. 
    There are at most $\tilde\OO(1/\epsilon^{s_B+1})$ elements in $H^{l'}$.
    The overall contribution to profit from items in $\cup_{(k,h) \in H^{l'}} T_k^{(l',h)}$ is 
    \begin{align*}
    \sum_{(h,k)\in H^{l'}}\sum_{j\in T_k^{(l',h)}}\tilde{p}_{j}\hat{y}_{j} &\le  \sum_{(h,k)\in H^{l'}}(1+\epsilon)\rho^{(l',h)}_{k}\sum_{j\in T_k^{(l',h)}}w^{l'}_{j}\hat{y}_{j}\\
    &\le (1+\epsilon)/\epsilon \sum_{(h,k)\in H^{l'}}\sum_{j\in T_k^{(l',h)}} w^{l'}_{j}\hat{y}_{j} \le {\OO}(\epsilon).
    \end{align*}
    According to Observation~\ref{obs:ropt_b}, we have $\sum_{j\in I^{(l',b)}}w^{l'}_j(1-x_j^{r^*}) \le 2s_B \epsilon$. 
    Define $\vey[\hat{I}]\in [0,1]^{\hat{n}}$ such that:
    \begin{itemize}
        \item[(i)] $y_{j} = \frac{1-2s_B \epsilon}{1+\epsilon}\hat{y}_{j}$ if $j \in T_k^{(l',h)} $ and $\lambda_k^{(l',h)} \neq 0 $;
        \item[(ii)] $y_{j} = 0$ if $j \in T_k^{(l',h)} $ and $\lambda_k^{(l',h)} = 0$;
        \item[(iii)] $y_{j} = 1-x^{r*}_{j}$ if $j \in I^{(l',b)} $;
        \item [(iv)]$y_{j} = 0$ if $j \in I^{(l',s)}$.
    \end{itemize}
    We claim that $\vey[\hat{I}]$ is a feasible solution to $\textbf{LP}'({l'}, \vex^{r*})$. It suffices to observe that
        \begin{itemize}
        \item $\sum_{j\in T_k^{(l',h)}}w^{l'}_{j}y_{j} \le \frac{1-2s_B \epsilon}{1+\epsilon}\sum_{j\in T_k^{(l',h)}}w^{l'}_{j}\hat{y}_{j} \le (1-2s_B \epsilon)\sum_{j\in T_k^{(l',h)}}w^{l'}_{j}(1-x^{r*}_{j}) \le \sum_{j\in T_k^{(l',h)}}w^{l'}_{j}(1-x^{r*}_{j}) ,\hspace{2mm} \forall k , h;$
        \item $ \sum_{j\in \hat{I}}\ve B^{l'}_{j}y_{j} \le \left(\sum_{j\in I^{(l',b)}}w^{l'}_{j}(1-x^{r*}_{j})\right)\cdot \ve1+\sum_{j\in I^{(l',m)}}\ve B^{l'}_{j}y_{j} \leq 2s_B \epsilon \cdot\ve1 + (1-2s_B \epsilon)\sum_{j\in I^{(l',m)}}\ve B^{l'}_{j}\hat{y}_j \leq \ve1.$
    \end{itemize}
    According to Observation~\ref{obs:lambda}, we have
    \begin{align*}
    \sum_{j\in \hat{I}}\tilde{p}_jy_j &=  \sum_{j\in I^{(l',b)}}\tilde{p}_j(1-x_j^{r^*})+  \sum_{j\not\in I^{(l',b)}}\tilde{p}_jy_j \le OPT_2^{(l';\vex^{r^*})} \le OPT^{(l';\vex^{r^*})}+\OO(s_B^2 \epsilon)\\
    &\le \sum_{j\in I^{(l',b)}}\tilde{p}_j(1-x_j^{r^*})+\sum_{h}\sum_{k}\rho^{(l',h)}_k\bar{\lambda}_{k}^{(l',h)} + \OO(s_B^2 \epsilon).
    \end{align*}
    It thus follows that $\sum_{j\not\in I^{(l',b)}}\tilde{p}_jy_j \le \sum_{h}\sum_{k}\rho^{(l',h)}_k\bar{\lambda}_{k}^{(l',h)} + \OO(s_B^2 \epsilon)$. Note that $M^{r^*} \le s_B+\OO(s_B\delta)$ by Observation~\ref{obs:M^*}, then we have
    
      \begin{align*}
    &P_{l'}+\frac{1-2s_B \epsilon}{1+\epsilon}\left(\sum_{j\in I^{(l',b)}}\tilde{p}_{j}(1-x^{\lambda}_{j})+ \sum_{j\in I^{(l',m)}}\tilde{p}_{j}\hat{y}_{j}\right) \\
    &\leq P_{l'}+ \sum_{j\in I^{(l',b)}}\tilde{p}_{j}(1-x^{\lambda}_{j})+ \frac{1-2s_B \epsilon}{1+\epsilon}\sum_{j\in I^{(l',m)}}\tilde{p}_{j}\hat{y}_{j}\\
    &\le P_{l'}+ \sum_{j\in I^{(l',b)}}\tilde{p}_{j}(1-x^{\lambda}_{j})+ 
    \sum_{j\not\in I^{(l',b)}}\tilde{p}_{j}y_j+\sum_{(h,k)\in H^{l'}}\sum_{j\in T_k^{(l',h)}}\tilde{p}_{j}\hat{y}_j\\
    &\le P_{l'}+ \sum_{j\in I^{(l',b)}}\tilde{p}_{j}(1-x^{\lambda}_{j})+ 
    \sum_{j\not\in I^{(l',b)}}\tilde{p}_{j}y_j+\OO(\epsilon)\\
    &\le P_{l'}+ \sum_{j\in I^{(l',b)}}\tilde{p}_{j}(1-x^{\lambda}_{j})+ \sum_{h}\sum_{k}\rho^{(l',h)}_k\bar{\lambda}_{k}^{(l',h)} +\OO(s^2_B \epsilon)\\
    &\leq M^{\lambda} + \OO(s_B^2 \epsilon)\le M^{r^*}+ \OO(s_B^2 \epsilon) \le s_B+\OO(s_B^2\epsilon).
    \end{align*}
    Thus $ P_{l'}+OPT_2^{(l';\vex^{\lambda})} \le  s_B+{\OO}(s_B^3 \epsilon)$, and Lemma~\ref{claim:fra_3} is proved.  
\end{proof}

\subsubsection{Rounding the fractional solution}\label{sub:rounding_frac}

Recall that $\vex^{\lambda}$ is an extreme point optimal solution to \textbf{cen-LP}$_{\lambda}$ and is also a feasible solution to $\textbf{MBi-IP}_r(\tilde{I},\ve1,\ve1)$. We round down $\vex^{\lambda}$ to $\tilde{\vex}^{\lambda}$ such that $\tilde{x}^{\lambda}_j = 0$ if $x^{\lambda}_j$ is a fractional number, otherwise $\tilde{x}^{\lambda}_j = x^{\lambda}_j$. Note that both the construction of $\textbf{cen-LP}_{\lambda}$ and the ellipsoid method to solve $\textbf{cen-LP}_{\lambda}$ run in polynomial time, i.e., the above algorithm we have designed to return $\tilde{\vex}^{\lambda}$ runs in polynomial time. We still need to prove that $\tilde{\vex}^{\lambda}$ is a feasible solution to $\textbf{IPC}_{2\delta}(\tilde{I},\ve1,\ve1)$ with an objective value of at most $s_B + \OO(\epsilon)$.

Note that in $\{x^{\lambda}_j:j
\in \hat{I}\}$, there are at most $s_A+\tilde\OO(\frac{s_B}{\epsilon^{2s_B+1}})$ variables take fractional values. Since $s_A$ and $s_B$ are fixed constants, $\epsilon_{>0}$ could be extremely small such that $s_A+ 
\tilde\OO(\frac{s_B}{\epsilon^{2s_B+1}}) \le 1 / \epsilon^{2s_B+3}$. Let $\delta = \epsilon^{2s_B+4}$, then we have the following.


\begin{lemma}\label{lemma:delta-mbi}
	Let $\widetilde{Obj}_{2\delta}(\tilde{\vex}^{\lambda})$ be the objective value of $\textbf{IPC}_{2\delta}(\tilde{I},\ve1,\ve1)$ for $\vex = \tilde{\vex}^{\lambda}$. Let $Obj^{MBi}_r(\vex^{\lambda})$ be the objective value of $\textbf{MBi-IP}_r(\tilde{I},\ve1,\ve1)$ for $\vex = \vex^{\lambda}$. We have 
    $$
        \widetilde{Obj}_{2\delta}(\tilde{\vex}^{\lambda}) \le  Obj^{MBi}_r(\vex^{\lambda})+ \OO(\epsilon)\le s_B+\OO(\epsilon).
    $$
\end{lemma}
\begin{proof}
Let $\vex^{\lambda}$ be an extreme point optimal solution to \textbf{cen-LP}$_{\lambda}$, we define $\tilde{\vex}^{\lambda}$ such that $\tilde{x}^{\lambda}_j = 0$ if $x^{\lambda}_j$ is a fractional number and $\tilde{x}^{\lambda}_j = x^{\lambda}_j$ otherwise. 
Let $Obj^{MBi}_r(\vex^{\lambda})$ be an objective value of $\textbf{MBi-IP}_r(\tilde{I},\ve1,\ve1)$ for $\vex = \vex^{\lambda}$. Let $\widetilde{Obj}_{2\delta}(\tilde{\vex}^{\lambda})$, $Obj^{MBi}(\tilde{\vex}^{\lambda})$ and $Obj^{MBi}_r(\tilde{\vex}^{\lambda})$ be the objective values of $\textbf{IPC}_{2\delta}(\tilde{I},\ve1,\ve1)$, $\textbf{MBi-IP}(\tilde{I},\ve1,\ve1)$ and $\textbf{MBi-IP}_r(\tilde{I},\ve1,\ve1)$ for $\vex = \tilde{\vex}^{\lambda}$, respectively. 

Recall that $\delta = \epsilon^{2s_B+4}$ and $\tilde{p}_j \le \delta$ for $j\in \hat{I}$. Since $\sum_{j\in \hat{I}: x^{\lambda}_j \in (0,1)} 1 \le 1 / s^{2s_B+3}$ and $\tilde{x}^{\lambda}_j = x^{*}_j = x^{\lambda}_j$ for $j\in \tilde{I}\setminus \hat{I}$, thus we have $$ Obj^{MBi}_r(\tilde{\vex}^{\lambda}) \le Obj^{MBi}_r(\vex^{\lambda})+ \sum\limits_{j\in \hat{I}: x^{\lambda}_j \in (0,1)} \tilde{p}_j \le Obj^{MBi}_r(\vex^{\lambda})+ \epsilon.$$
According to Lemma~\ref{lemma:cen-LP}, we have $Obj^{MBi}_r(\vex^{\lambda})\le s_B+\OO(\epsilon)$, thus $Obj^{MBi}_r(\tilde{\vex}^{\lambda}) \le s_B+\OO(\epsilon)$.

On the other hand, according to Lemma~\ref{lemma:MBi_1}, we have $\widetilde{Obj}_{2\delta}(\tilde{\vex}^{\lambda}) \le Obj^{MBi}(\tilde{\vex}^{\lambda})+\epsilon$. And it holds that $ Obj^{MBi}(\tilde{\vex}^{\lambda}) \le Obj^{MBi}_r(\tilde{\vex}^{\lambda})$ by the fact that in $\textbf{MBi-IP}_r(\tilde{I},\ve1,\ve1)$, the leader is facing a stronger follower who can fractionally pack items. Hence the Lemma~\ref{lemma:delta-mbi} is proved.\end{proof}

In conclude, we have designed a polynomial time algorithm that returns a feasible solution to $\textbf{IPC}_{2\delta}(\tilde{I},\ve1,\ve1)$ with an objective value of at most $s_B + \OO(\epsilon)$. According to Lemma~\ref{lemma:r_theta}, we have proved Lemma~\ref{lemma:residue-main}.


\section{Omitted contents in Section~\ref{sec:general_case}}\label{sec:general_case_app}
\newtheorem*{T3}{Theorem \ref{theorem:general_1}}
\begin{T3} 
	Given a separation oracle $O_L$ for the leader's problem, and an oracle $O_F$ for the follower's problem that returns a $\rho$-approximation solution, 
	there exists a $(\frac{\rho}{1-\alpha}, \frac{1}{\alpha})$-bicriteria approximation algorithm for any $\alpha\in (0,1)$ that returns a solution $\vex^*\in\{0,1\}^{n}$ such that $\ve A\vex^*\le \frac{1}{\alpha}\cdot\vea$, and $$\max\{\vep\vey:\ve B\vey\le \veb,\vey\le \ve1-\vex^*\}\le \frac{\rho T^*}{1-\alpha},$$ where $T^*$ is the optimal objective value of \textbf{IPC}$(I,\vea,\veb)$. Furthermore, the algorithm runs in polynomial oracle time.
\end{T3}

Towards the proof of Theorem~\ref{theorem:general_1}, we start with an equivalent formulation of \textbf{IPC}$(I,\vea,\veb)$ as follows:


	\begin{align*}
	\textbf{IPC}^1(I,\vea,\veb):
	\min \hspace{1mm} & T&\\
	s.t. \hspace{1mm}& \ve A\vex \le \vea& \\
	& \vex\in\{0,1\}^n&\\
	& \sum_{j=1}^n p_j(1-x_j)y_j\leq T, \ \forall \vey\in\Omega:=\{\ve B\vey\leq\veb, \vey\in\{0,1\}^{n}\} &
	\end{align*}
Note that \textbf{IPC}$^1(I,\vea,\veb)$ is an IP with possibly exponential number of constraints. Replacing $\vex \in \{0,1\}^n$ with $\vex \in [0,1]^n$, we obtain a linear relaxation and denote it as \textbf{IPC}$_f^1(I,\vea,\veb)$.

We start with a weaker result.

\begin{lemma}\label{lemma:opt-sep}
	Given a separation oracle $O_L$ for the leader's problem, and an oracle $O_F$ for the follower's problem that always returns the optimal solution, \textbf{IPC}$_f^1(I,\vea,\veb)$ can be solved in polynomial oracle time.
\end{lemma}
\begin{proof}
 Note that \textbf{IPC}$_f^1(I,\vea,\veb)$ is equivalent to the following determine problem: given an arbitrary fixed target value $T_0$, does there exist $\vex \in [0,1]^{n}$ such that $\ve A\vex\le \vea$ and $\vep(\ve1-\vex)\vey\leq T_0$ for all $\vey\in\Omega$? The determine problem can be formulated as the following linear program:
    		\begin{align*}
    		\textbf{LP}_{\vex}(T_0):		& \ve A\vex \le \vea& \\
    		& \vex\in[0,1]^n&\\
    		& \sum_{j=1}^n p_j(1-x_j)y_j\leq T_0, \  \forall \vey\in\Omega:=\{\ve B\vey\leq\veb, \vey\in\{0,1\}^{n}\} &
    		\end{align*}
    	
  In the following we test the feasibility of	$\textbf{LP}_{\vex}(T_0)$ through ellipsoid method. We give a very brief description of ellipsoid method and the reader may refer to, e.g.,~\cite{grotschel2012geometric} for details. 
  
The ellipsoid method iteratively computes a sequence of ellipsoids $E_0,E_1,\cdots$. In each iteration, it implements a separation oracle to check whether the center of the current ellipsoid $E_k$ is feasible. If it is, then it terminates and returns this point; Otherwise, it finds out a violating constraint as a cut. Incorporating the cut output by the separation oracle, the ellipsoid method computes a new ellipsoid $E_{k+1}$ and guarantees that the volume of the new ellipsoid is smaller than $E_k$ by a significant factor. After a polynomial number of iterations, the ellipsoid method either returns a feasible solution, or asserts that there is no feasible solution and returns a polynomial number of cuts as a certificate.
  
Towards solving $\textbf{LP}_{\vex}(T_0)$, it suffices to design a {separation oracle} such that given any $\vex^{0}\in[0,1]^{n}$, it either asserts that  $\vex^{0}$ is feasible to $\textbf{LP}_{\vex}(T_0)$, or returns a constraint of $\textbf{LP}_{\vex}(T_0)$ that is violated by $\vex^{0}$. 
  
  We construct the separation oracle \textbf{SEP} as follows. Consider the following IP:
		\begin{align*}
		\textbf{IP}_{\vey}(\vex^{0}): \max \hspace{1mm} & \sum_{j=1}^n p_j(1-x^{0}_j)y_j&\\
		&\ve B\vey\leq\veb&\\
		& \vey\in\{0,1\}^{n}&
		\end{align*}
	Note that \textbf{IP}$_{\vey}(\vex^{0})$ is the follower's problem with item profits being $p_j(1-x^0_j)$ for $j = 1,\cdots,n$. Apply the oracle $O_F$ for the follower's problem, we get an optimal solution to $\textbf{IP}_{\vey}(\vex^{0})$ and denote it as $\vey^*$. 
If the optimal objective value of $\textbf{IP}_{\vey}(\vex^{0})$ is larger than $T_{0}$, then we know $\vex^0$ violates the constraint $\sum_{j=1}^n p_j(1-x_j)y_j^*\leq T_{0}$ and \textbf{SEP} returns this violating constraint. Otherwise, we apply the leader's separation oracle $O_L$ for $\vex=\vex^{0}$. If $\vex^0$ does not satisfy $\ve A\vex^0\le \vea$,  $O_L$ returns a violating constraint, and we let \textbf{SEP} return the same constraint. Otherwise, $\vex^0$ is feasible to $\textbf{LP}_{\vex}(T_{0})$. 
\end{proof}

Generalizing our technique for Lemma~\ref{lemma:opt-sep}, we obtain the following strengthened lemma.
\begin{lemma}\label{lemma:apx-sep}
Given a separation oracle $O_L$ for the leader's problem, and an oracle $O_F$ for the follower's problem that returns a $\rho$-approximation solution, a $\rho$-approximation solution to \textbf{IPC}$_f^1(I,\vea,\veb)$ can be found in polynomial oracle time.
\end{lemma}


\begin{proof}
Again, for any target value $T_0$, we consider the following:
    		\begin{align*}
    		\textbf{LP}_{\vex}(T_0):		& \ve A\vex \le \vea& \\
    		& \vex\in[0,1]^n&\\
    		& \sum_{i=1}^n p_j(1-x_j)y_j\leq T_0, \  \forall \vey\in\Omega:=\{\ve B\vey\leq\veb, \vey\in\{0,1\}^{n}\} &
    		\end{align*}
We first present an ellipsoid method based algorithm $Al$ and then show its property.
$Al$ works as follows. Consider the following IP:
		\begin{align*}
		\textbf{IP}_{\vey}(\vex^{0}): \max \hspace{1mm} & \sum_{j=1}^n p_j(1-x^{0}_j)y_j&\\
		&\ve B\vey\leq\veb&\\
		& \vey\in\{0,1\}^{n}&
		\end{align*}
	Apply the oracle for the follower's problem, we get a $\rho$-approximation solution to $\textbf{IP}_{\vey}(\vex^{0})$ and denote it as $\vey^0$. If $\sum_{j=1}^n p_j(1-x^{0}_j)y_j^0\le  T_0$ and $\vex^0$ satisfies that $\ve A\vex^0\le \vea$, $Al$ stops and returns $\vex=\vex^0$. Otherwise, either $\sum_{j=1}^n p_j(1-x^{0}_j)y_j^0> T_0$, then $Al$ returns the violating constraint $\sum_{j=1}^n p_j(1-x_j)y_j^0\le T_0$ to $\textbf{LP}_{\vex}(T_0)$ and continues; or $\vex^0$ does not satisfy that $\ve A\vex^0\le \vea$, then $Al$ returns the same violating constraint that is returned by $O_L$ and continues.
	
	Now we analyze $Al$. We know the ellipsoid method based algorithm $Al$ either returns some $\vex^0$ at some iteration and stops, or after a polynomial number of iterations it ends up with an empty volume ellipsoid and terminates without finding any feasible solution.
	
	Let $\vey^*$ be an optimal solution to $\textbf{IP}_{\vey}(\vex^{0})$ and let $T^*(\vex^0)=\sum_{j=1}^n p_j(1-x^{0}_j)y_j^*$. We have the following observations.
	
	\begin{observation}
	If $T^*(\vex^0)>  \rho T_0$, then $Al$ must return a violating constraint $\sum_{j=1}^n p_j(1-x_j)y_j^0\le T_0$. 
	\end{observation}
	The observation follows as $\vey^0$ is a $\rho$-approximation solution, and thus $\sum_{j=1}^n p_j(1-x^{0}_j)y_j^0\ge  T^*(x^0)/\rho > T_0$.
	\begin{observation}
	If $T^*(\vex^0)\le T_0$ and $A\vex^0\le \vea$, then $Al$ must return $\vex=\vex^0$.
	\end{observation}
	The observation follows as the feasible solution $\vey^0$ cannot have an objective value larger than that of the optimal solution.
	
	If, however, $T^*(\vex^0)\in (T_0,\rho T_0]$, $Al$ either returns some $\vex^0$ or terminates without finding any solution. Consequently, if $Al$ returns some $\vex=\vex^0$, then $T^*(\vex^0)\le  \rho T_0$; if $Al$ terminates without finding any solution, then $Al$ has found a polynomial number of violating constraints $\sum_{j=1}^n p_j(1-x_j)y_j^k\le T_0$ for $k=0,1,\cdots,m$, together with a polynomial number of constraints from $\ve A\vex\le \vea$, such that there does not exist $\vex\in[0,1]^n$ satisfying these constraints simultaneously. 
That is, $Al$ returns a polynomial number of constraints a certificate that $\textbf{LP}_{\vex}(T_0)$ does not admit any feasible solution. Hence, the following is true.
	
	\begin{claim}
	$Al$ either returns some solution $\vex^0$ which guarantees that $\sum_{j=1}^n p_j(1-x^{0}_j)y_j\le \rho T_0$ for any $\vey\in\Omega$; or it asserts that there does not exist any feasible solution to $\textbf{LP}_{\vex}(T_0)$, and hence the optimal objective value of  \textbf{IPC}$_f^1(I,\vea,\veb)$ is larger than $T_0$.
	\end{claim}
Hence, with binary search, we obtain a $\rho$-approximation solution for \textbf{IPC}$_f^1(I,\vea,\veb)$.	\end{proof}
	
Finally, we are able to obtain a bicriteria-approximation solution to \textbf{IPC}$^1(I,\vea,\veb)$ by rounding the $\rho$-approximation solution to \textbf{IPC}$_f^1(I,\vea,\veb)$, thus finalizing the proof of Theorem \ref{theorem:general_1}.  

\begin{proof}[Proof of Theorem \ref{theorem:general_1}]
Denote by $T^*_f$ and $T^*$ the optimal objective values of \textbf{IPC}$_f^1(I,\vea,\veb)$ and \textbf{IPC}$^1(I,\vea,\veb)$, respectively. Then it follows that $T^*_f \le T^*$.
Let $\vex$ be a  $\rho$-approximation solution to \textbf{IPC}$_f^1(I,\vea,\veb)$. Let $\alpha\in(0,1)$ be any fixed parameter. We round up $x_{j}$ to 1 if $x_{j}\geq \alpha$, otherwise, round down $x_{j}$ to 0. Denote by $\hat{\vex}$ this integral solution, we have: (i) $\hat{\vex}$ violates the leader's budget by a factor of $\frac{1}{\alpha}$ since $A( \alpha \hat{\vex} )\le A\vex \le \vea$; (ii) $\sum\limits_{j=1}^{n}p_{j}(1-\hat{x}_{j})y_{j}=\sum\limits_{j:x_{j}<\alpha}p_{j}y_{j}+\sum\limits_{j:x_{j}\geq\alpha}0\leq\frac{1}{1-\alpha}\sum\limits_{j=1}^{n}p_{j}(1-x_{j})y_{j}\leq\frac{\rho}{1-\alpha}T^*_f\leq \frac{\rho}{1-\alpha}T^*$. Thus $\hat{\vex}$ is a $(\frac{\rho}{1-\alpha},\frac{1}{\alpha})$-bicriteria approximation solution of \textbf{IPC}$(I,\vea,\veb)$. Theorem \ref{theorem:general_1} is proved. \end{proof}

\bibliographystyle{plain}
\bibliography{kip}

\end{document}